\documentclass[a4paper]{article}
\usepackage{geometry}
\usepackage{amsmath,amsfonts,amssymb,amsthm,stmaryrd,mathrsfs,mathdots}
\usepackage[small]{diagrams}
\usepackage{hyperref}
\usepackage{bm}
\usepackage{tikz}

\usepackage{sectsty}
\sectionfont{\sffamily}
\subsectionfont{\sffamily}
\subsubsectionfont{\sffamily}

\newtheoremstyle{myplain}
  {\topsep}   
  {\topsep}   
  {\itshape}  
  {0pt}       
  {\bfseries\sffamily} 
  {.}         
  {5pt plus 1pt minus 1pt} 
  {}          
  
\newtheoremstyle{mydefinition}
  {\topsep}   
  {\topsep}   
  {\normalfont}  
  {0pt}       
  {\bfseries\sffamily} 
  {.}         
  {5pt plus 1pt minus 1pt} 
  {}          
  
\newtheoremstyle{myremark}
  {0.5\topsep}   
  {0.5\topsep}   
  {\normalfont}  
  {0pt}       
  {\sffamily} 
  {.}         
  {5pt plus 1pt minus 1pt} 
  {}          

\theoremstyle{mydefinition}
\newtheorem{definition}{$\blacktriangleright$ Definition}[subsection]
\theoremstyle{myplain}
\newtheorem{theorem}[definition]{$\blacktriangleright$ Theorem}
\newtheorem{lemma}[definition]{$\blacktriangleright$ Lemma}
\newtheorem{corollary}[definition]{$\blacktriangleright$ Corollary}
\newtheorem{proposition}[definition]{$\blacktriangleright$ Proposition}
\theoremstyle{myremark}
\newtheorem*{remark}{$\blacktriangleright$ Remark}
\newtheorem*{notation}{$\blacktriangleright$ Notation}
\newtheorem{example}[definition]{$\blacktriangleright$ Example}

\newtheorem*{convention}{$\blacktriangleright$ Convention}

\newarrow{Corresponds}<--->
\newarrow{Dotscorresponds}<...>

\begin{document}
\title{\Huge \bf Game-theoretic Investigation of Intensional Equalities}
\author{\Large \bf Norihiro Yamada \\
{\tt norihiro1988@gmail.com} \\
University of Oxford
}

\maketitle

\begin{abstract}
We present a game semantics for \emph{Martin-L\"of type theory (\textsf{MLTT})} that interprets \emph{propositional equalities} in a non-trivial manner in the sense that it refutes the principle of \emph{uniqueness of identity proofs (UIP)} for the first time as a game semantics in the literature.
Specifically, each of our games is equipped with (selected) invertible strategies representing \emph{(computational) proofs of (intensional) equalities} between strategies on the game; these invertible strategies then interpret propositional equalities in \textsf{MLTT}, which is roughly how our model achieves the non-trivial interpretation of them.
Consequently, our game semantics provides a natural and intuitive yet mathematically precise formulation of the \emph{BHK-interpretation} of propositional equalities. 
From an extensional viewpoint, the algebraic structure of our model is similar to the classic \emph{groupoid model} by Hofmann and Streicher, but the former is distinguished from the latter by its computational and intensional nature. 
In particular, our game semantics refutes the axiom of \emph{function extensionality (FunExt)} and the \emph{univalence axiom (UA)}.
This provides a sharp contrast to the recent \emph{cubical set model} as well.
Similar to the path from the groupoid model to the \emph{$\omega$-groupoid models}, the present work is also intended to be a stepping stone towards a game semantics for the \emph{infinite hierarchy of propositional equalities} in \textsf{MLTT}.
\end{abstract}

\tableofcontents

\section{Introduction}
\label{Intro}
This paper presents a game semantics for \emph{Martin-L\"of type theory (\textsf{MLTT})} that refutes the principle of \emph{uniqueness of identity proofs (UIP)}, the axiom of \emph{function extensionality (FunExt)} and the \emph{univalence axiom (UA)}.
Our motivation is to give a computational and intensional explanation of \textsf{MLTT}, particularly its \emph{propositional equalities}, in a mathematically precise and syntax-independent manner.
This work is a continuation and an improvement of \cite{yamada2016game}.

\subsection{Martin-L\"{o}f Type Theory}
\emph{\textsf{MLTT}} \cite{martin1982constructive,martin1984intuitionistic,martin1998intuitionistic} is an extension of the simply-typed $\lambda$-calculus that, under the \emph{Curry-Howard isomorphism (CHI)} \cite{curry1934functionality,rosser1967curry,howard1980formulae}, corresponds to intuitionistic predicate logic, for which the extension is made by \emph{dependent types}, i.e., types that depend on terms.
It was proposed by Martin-L\"{o}f as a foundation of constructive mathematics, but also it has been an object of active research in computer science because it can be seen as a programming language, and one may extract programs that are ``correct by construction'' from its proofs. 
Moreover, based on the homotopy-theoretic interpretation, an extension of \textsf{MLTT}, called \emph{homotopy type theory (\textsf{HoTT})}, was recently proposed, providing new topological insights and having potential to be a powerful and practical foundation of mathematics \cite{hottbook}. 

Conceptually, \textsf{MLTT} is based on the \emph{BHK-interpretation} of intuitionistic logic \cite{troelstra1988constructivism} which interprets proofs as \emph{``constructions''}. 
Note that the BHK-interpretation is \emph{informal} in nature, and formulated \emph{syntactically} as \textsf{MLTT}; also it is reasonable to think of such constructions as some computational processes which must be an \emph{intensional} concept.
Thus, a \emph{semantic} (or \emph{syntax-independent}) model of \textsf{MLTT} that interprets proofs as some \emph{intensional} constructions \emph{in a mathematically precise sense} should be considered as primarily important for clarification and justification of \textsf{MLTT}. 
Also, it may give new insights for meta-theoretic study of the syntax.
Just for convenience, let us call such a semantics a \emph{computational semantics} for \textsf{MLTT}.
In particular, a computational semantics would interpret a proof of a \emph{propositional equality} $\mathsf{a_1 =_A a_2}$ as a computational process which ``witnesses'' that the interpretations of $\mathsf{a_1}$ and $\mathsf{a_2}$ are equal processes.

\subsection{Game Semantics}
\emph{Game semantics} \cite{abramsky1997semantics,abramsky1999game} refers to a particular kind of \emph{semantics of logics and programming languages} in which types and terms are interpreted as \emph{``games''} and \emph{``strategies''}, respectively. 
One of its distinguishing characteristics is to interpret syntax as dynamic interactions between two ``players'' of a game, providing a computational and intensional explanation of proofs and programs in a natural and intuitive yet mathematically precise and syntax-independent manner.

Remarkably, game semantics for \textsf{MLTT} was not addressed until Abramsky et~al. recently constructed such a model in \cite{abramsky2015games} based on  AJM-games \cite{abramsky2000full} and their extension \cite{abramsky2005game}. 
Moreover, the present author developed another game semantics for \textsf{MLTT} in the unpublished paper \cite{yamada2016game} that particularly achieves an interpretation of the (cumulative hierarchy of) \emph{universes} as games. 
By the nature of game semantics described above, these game-theoretic models may be seen as possible realizations of a computational semantics for \textsf{MLTT}.

\subsection{Problem to Solve: Game Semantics for Propositional Equalities}
However, although these game models \cite{abramsky2015games,yamada2016game} interpret large parts of \textsf{MLTT}, neither interprets propositional equalities in a non-trivial manner in the sense that they both interpret each propositional equality as an identity strategy, and so they validate \emph{UIP} \cite{hofmann1998groupoid}.
This is a problem as the \emph{groupoid model} \cite{hofmann1998groupoid} proved that UIP is not derivable in \textsf{MLTT}.

On the other hand, propositional equalities have been a mysterious and intriguing concept; they recently got even more attentions by a new topological interpretation of them, namely as \emph{paths between points} \cite{hottbook}. 
In particular, it posed a central open problem: a ``constructive justification'' of \emph{UA} \cite{hottbook}; a significant step towards this goal is taken in \cite{bezem2014model}, even leading to new syntax for equalities called \emph{path types} \cite{cohen2016cubical}. 
Nevertheless, the topological interpretation is a priori very different from the BHK-interpretation: It is spatial and extensional.
In particular, it does not interpret proofs of a propositional equality as computational processes.

Therefore we believe that it would lead to a better understanding of propositional equalities in \textsf{MLTT} to give a game semantics (or more generally a computational semantics) that interprets them in a non-trivial manner.

\subsection{Our Approach: Game-theoretic Groupoid Interpretation}
This paper solves the open problem mentioned above: It presents a game semantics for \textsf{MLTT} that \emph{does} refute UIP.
Specifically, we refine the previous work \cite{yamada2016game} by equipping each game $G$ with a set $=_G$ of \emph{invertible} strategies $\rho$ between\footnote{It is not necessarily the game-semantic sense but the category-theoretic one.} strategies $\sigma, \sigma' : G$, where $\rho$ is regarded as a ``\emph{(computational) proof of the (intensional) equality}'' between $\sigma$ and $\sigma'$, and used as a computational semantics for the corresponding propositional equality.
To interpret various phenomena in \textsf{MLTT}, we require that the structure $G = (G, =_G)$ forms a \emph{groupoid}, i.e., a category whose morphisms are all isomorphisms.

\subsection{Related Work and Our Contribution}
The primary contribution of this paper is to give a game semantics for \textsf{MLTT} that refutes UIP.

The algebraic structure of our game semantics is similar to the classic groupoid model \cite{hofmann1998groupoid} by Hofmann and Streicher; thus, our main technical contribution is to add \emph{intensional} structures and operations in such a way that the operations preserve the structures.
In fact, our game model stands in a sharp contrast to the groupoid model in its intensional nature: The game model refutes \emph{FunExt} \cite{hottbook} and UA.
Hence, it is also distinguished from the recent cubical set model \cite{bezem2014model,cohen2016cubical}, and conceptually closer to the BHK-interpretation.

Our hope is that their topological and our computational perspectives are complementary and fruitful, rather than opposing, to each other. 

\subsection{Structure of the Paper}
The rest of the paper proceeds as follows.
First, Section~\ref{McGamesAndStrategies} reviews necessary backgrounds in game semantics, and Section~\ref{PredicativeGames} recalls the games and strategies in \cite{yamada2016game}.
Then we refine this variant in Sections~\ref{GamesWithEqualities}, \ref{DependentGamesWithEqualities}, leading to the central notion of \emph{games with equalities (GwEs)}.
Finally, we construct a model of \textsf{MLTT} by GwEs in Section~\ref{InterpretationOfMLTT}, and make a conclusion in Section~\ref{ConclusionAndFutureWork}.

\section{Games with Equalities}
\label{GamesWithEqualities}
This section presents our games and strategies.
Let us first fix notation:
\begin{itemize}

\item We use bold letters $\bm{s}, \bm{t}, \bm{u}, \bm{v}$, etc. for sequences, in particular $\bm{\epsilon}$ for the \emph{empty sequence}, and letters $a, b, c, d, m, n, p, q, x, y, z$, etc. for elements of sequences.

\item A \emph{concatenation} of sequences is represented by a juxtaposition of them, but we write $a \bm{s}$, $\bm{t} b$, $\bm{u} c \bm{v}$ for $(a) \bm{s}$, $\bm{t} (b)$, $\bm{u} (c) \bm{v}$, etc. We sometimes write $\bm{s} . \bm{t}$ for $\bm{s t}$ for readability.

\item We write $\mathsf{even}(\bm{s})$ (resp. $\mathsf{odd}(\bm{s})$) if $\bm{s}$ is of \emph{even-length} (resp. \emph{odd-length}). For a set $S$ of sequences, we define $S^\mathsf{even} \stackrel{\mathrm{df. }}{=} \{ \bm{s} \in S \ \! | \! \ \mathsf{even}(\bm{s}) \}$ and $S^\mathsf{odd} \stackrel{\mathrm{df. }}{=} \{ \bm{t} \in S \ \! | \! \ \mathsf{odd}(\bm{t}) \}$.

\item $\bm{s} \preceq \bm{t}$ means $\bm{s}$ is a \emph{prefix} of $\bm{t}$, and $\mathsf{pref}(S) \stackrel{\mathrm{df. }}{=} \{ \bm{s} \ \! | \ \! \exists \bm{t} \in S . \ \! \bm{s} \preceq \bm{t} \ \! \}$.


\item For a function $f : A \to B$ and a subset $S \subseteq A$, $f \upharpoonright S$ denotes the \emph{restriction} of $f$ to $S$. 

\item Given sets $X_1, X_2, \dots, X_n$, $\pi_i : X_1 \times X_2 \times \dots \times X_n \to X_i$ denotes the \emph{$i^{\text{th}}$-projection function} $(x_1, x_2, \dots, x_n) \mapsto x_i$ for each $i \in \{ 1, 2, \dots, n \}$.

\item Let $X^* \stackrel{\mathrm{df. }}{=} \{ x_1 x_2 \dots x_n \ \! | \ \! n \in \mathbb{N}, x_i \in X \ \! \}$ for each set $X$.

\item Given a sequence $\bm{s}$ and a set $X$, $\bm{s} \upharpoonright X$ denotes the subsequence of $\bm{s}$ that consists of elements in $X$. When $\bm{s} \in Z^*$ with $Z = X + Y$ for some set $Y$, we abuse notation: The operation deletes the ``tags'' for the disjoint union, so that $\bm{s} \upharpoonright X \in X^*$.

\item For a poset $P$ and a subset $S \subseteq P$, $\mathsf{sup}(S)$ denotes the \emph{supremum} of $S$. 


\end{itemize}

\subsection{Review: McCusker's Games and Strategies}
\label{McGamesAndStrategies}
As mentioned above, our game semantics is a refinement of the previous work \cite{yamada2016game}, which is based on McCusker's games and strategies \cite{abramsky1999game,mccusker1998games}.
Therefore let us begin with a quick review of McCusker's variant, focusing mainly on basic definitions. 
See \cite{abramsky1999game,mccusker1998games} for a more detailed treatment of this variant, and \cite{abramsky1997semantics,abramsky1999game,hyland1997game} for a general introduction to the field of game semantics.

\subsubsection{Games}
A \emph{game}, roughly, is a certain kind of a rooted forest whose branches represent \emph{plays} in the ``game in the ordinary sense'' (such as chess) it represents. These branches are finite sequences of \emph{moves} of the game; a play of the game proceeds as its participants alternately make moves. Thus, a game is what specifies possible interactions between the participants, and so it interprets a \emph{type} in computation (resp. a \emph{proposition} in logic) which specifies terms of the type (resp. proofs of the proposition). For our purpose, it suffices to focus on games between two participants, \emph{Player} (who represents a ``computer'' or a ``mathematician'') and \emph{Opponent} (who represents an ``environment'' or a ``rebutter''), where Opponent always starts a play.  

Technically, games are based on two preliminary concepts: \emph{arenas} (Definition~\ref{DefArenas}) and \emph{legal positions} (Definition~\ref{DefLegalPositions}). An arena defines the basic components of a game, which in turn induces legal positions that specify the basic rules of the game. 

\begin{definition}[Arenas \cite{abramsky1999game,mccusker1998games}]
\label{DefArenas}
An \emph{\bfseries arena} is a triple $G = (M_G, \lambda_G, \vdash_G)$, where $M_G$ is a set whose elements are called \emph{\bfseries moves}, $\lambda_G$ is a function $M_G \to \{ \mathsf{O}, \mathsf{P} \} \times \{\mathsf{Q}, \mathsf{A} \}$, where $\mathsf{O}, \mathsf{P}, \mathsf{Q}, \mathsf{A}$ are some distinguished symbols, called the \emph{\bfseries labeling function}, and $\vdash_G \ \subseteq (\{ \star \} + M_G) \times M_G$, where $\star$ is some fixed element, is called the \emph{\bfseries enabling relation}, which satisfies:
\begin{itemize}

\item {\sffamily (E1)} If $\star \vdash_G m$, then $\lambda_G (m) = \mathsf{OQ}$ and $n \vdash_G m \Leftrightarrow n = \star$

\item {\sffamily (E2)} If $m \vdash_G n$ and $\lambda_G^\mathsf{QA} (n) = \mathsf{A}$, then $\lambda_G^\mathsf{QA} (m) = \mathsf{Q}$

\item {\sffamily (E3)} If $m \vdash_G n$ and $m \neq \star$, then $\lambda_G^\mathsf{OP} (m) \neq \lambda_G^\mathsf{OP} (n)$

\end{itemize}
in which $\lambda_G^\mathsf{OP} \stackrel{\mathrm{df. }}{=} \pi_1 \circ \lambda_G : M_G \to \{ \mathsf{O}, \mathsf{P} \}$ and $\lambda_G^\mathsf{QA} \stackrel{\mathrm{df. }}{=} \pi_2 \circ \lambda_G : M_G \to \{ \mathsf{Q}, \mathsf{A} \}$.
A move $m \in M_G$ is called \emph{\bfseries initial} if $\star \vdash_G m$, an \emph{\bfseries O-move} if $\lambda_G^\textsf{OP}(m) = \textsf{O}$, a \emph{\bfseries P-move} if $\lambda_G^\textsf{OP}(m) = \mathsf{P}$, a \emph{\bfseries question} if $\lambda_G^\textsf{QA}(m) = \textsf{Q}$, and an \emph{\bfseries answer} if $\lambda_G^\textsf{QA}(m) = \textsf{A}$.
The symbols $\mathsf{O, P, Q, A}$ are calle \emph{\bfseries labels}.
\end{definition}

\begin{definition}[J-sequences \cite{hyland2000full,abramsky1999game,mccusker1998games}]
A \emph{\bfseries justified (j-) sequence} in an arena $G$ is a finite sequence $\bm{s} \in M_G^*$, in which each non-initial move $m$ is associated with a move $\mathcal{J}_{\bm{s}}(m)$, or written $\mathcal{J}(m)$, called the \emph{\bfseries justifier} of $m$ in $\bm{s}$, that occurs previously in $\bm{s}$ and satisfies $\mathcal{J}_{\bm{s}}(m) \vdash_G m$. We also say that $m$ is \emph{\bfseries justified} by $\mathcal{J}_{\bm{s}}(m)$, and there is a \emph{\bfseries pointer} from $m$ to $\mathcal{J}_{\bm{s}}(m)$. 
\end{definition}

\begin{notation}
We write $\mathscr{J}_G$ for the set of all j-sequences in an arena $G$.
\end{notation}

\begin{remark}
Given arenas $A, B$ and j-sequences $\bm{s}m \in \mathscr{J}_A$, $\bm{t}n \in \mathscr{J}_B$, the equation $\bm{s}m = \bm{t}n$ means not only they are equal sequences but also their justification and labeling structures are identical, i.e., $\bm{s} m = \bm{t} n \Rightarrow \lambda_A(m) = \lambda_B(n) \wedge \mathcal{J}_{\bm{s}}(m) = \mathcal{J}_{\bm{t}}(n)$, where $\mathcal{J}_{\bm{s}}(m) = \mathcal{J}_{\bm{t}}(n)$ denotes not only the equality of moves but also the equality of the positions of their respective occurrences in $\bm{s}$ and $\bm{t}$, i.e., $\mathcal{J}_{\bm{s}}(m)$ and $\mathcal{J}_{\bm{t}}(n)$ are the $i^{\text{th}}$-elements of $\bm{s}$ and $\bm{t}$, respectively, for some $i \in \mathbb{N}$.
\end{remark}

\begin{definition}[Views \cite{hyland2000full,abramsky1999game,mccusker1998games}]
For a j-sequence $\bm{s}$ in an arena $G$, we define the \emph{\bfseries Player (P-) view} $\lceil \bm{s} \rceil_G$ and the \emph{\bfseries Opponent (O-) view} $\lfloor \bm{s} \rfloor_G$ by induction on the length of $\bm{s}$: 
\begin{itemize}

\item $\lceil \bm{\epsilon} \rceil_G \stackrel{\mathrm{df. }}{=} \bm{\epsilon}$

\item $\lceil \bm{s} m \rceil_G \stackrel{\mathrm{df. }}{=} \lceil \bm{s} \rceil_G . m$, if $m$ is a P-move

\item $\lceil \bm{s} m \rceil_G \stackrel{\mathrm{df. }}{=} m$, if $m$ is initial 

\item $\lceil \bm{s} m \bm{t} n \rceil_G \stackrel{\mathrm{df. }}{=} \lceil \bm{s} \rceil_G . m n$, if $n$ is an O-move with $\mathcal{J}_{\bm{s} m \bm{t} n}(n) = m$

\item $\lfloor \bm{\epsilon} \rfloor_G \stackrel{\mathrm{df. }}{=} \bm{\epsilon}$, $\lfloor \bm{s} m \rfloor_G \stackrel{\mathrm{df. }}{=} \lfloor \bm{s} \rfloor_G . m$, if $m$ is an O-move

\item $\lfloor \bm{s} m \bm{t} n \rfloor_G \stackrel{\mathrm{df. }}{=} \lfloor \bm{s} \rfloor_G . m n$, if $n$ is a P-move with $\mathcal{J}_{\bm{s} m \bm{t} n}(n) = m$

\end{itemize}
where justifiers of the remaining moves in $\lceil \bm{s} \rceil_G$ (resp. $\lfloor \bm{s} \rfloor_G$) are unchanged if they occur in $\lceil \bm{s} \rceil_G$ (resp. $\lfloor \bm{s} \rfloor_G$) and undefined otherwise.
\end{definition}

\begin{definition}[Legal positions \cite{abramsky1999game,mccusker1998games}]
\label{DefLegalPositions}
A \emph{\bfseries legal position} in an arena $G$ is a finite sequence $\bm{s} \in M_G^*$ (equipped with justifiers) that satisfies the following conditions:
\begin{itemize}
\item \textsf{(Justification)} $\bm{s}$ is a j-sequence in $G$

\item \textsf{(Alternation)} If $\bm{s} = \bm{s}_1 m n \bm{s}_2$, then $\lambda_G^\mathsf{OP} (m) \neq \lambda_G^\mathsf{OP} (n)$


\item \textsf{(Visibility)} If $\bm{s} = \bm{t} m \bm{u}$ with $m$ non-initial, then $\mathcal{J}_{\bm{s}}(m)$ occurs in $\lceil \bm{t} \rceil_{G}$ if $m$ is a P-move, and it occurs in $\lfloor \bm{t} \rfloor_{G}$ if $m$ is an O-move.

\end{itemize}

\end{definition}

\begin{notation}
We write $\mathscr{L}_G$ for the set of all legal positions in an arena $G$.
\end{notation}

\begin{definition}[Threads \cite{abramsky1999game,mccusker1998games}]
Let $G$ be an arena, and $\bm{s} \in \mathscr{L}_G$. Assume that $m$ is an occurrence of a move in $\bm{s}$. The \emph{\bfseries chain of justifiers} from $m$ is a sequence $m_0 m_1 \dots m_k m$ of justifiers, i.e., $m_0 m_1 \dots m_k m \in M_G^\ast$ that satisfies $\mathcal{J}(m) = m_k, \mathcal{J}(m_k) = m_{k-1}, \dots, \mathcal{J}(m_1) = m_0$, where $m_0$ is initial. In this case, we say that $m$ is \emph{\bfseries hereditarily justified} by $m_0$. The subsequence of $\bm{s}$ consisting of the chains of justifiers in which $m_0$ occurs is called the \emph{\bfseries thread} of $m_0$ in $\bm{s}$. 
An occurrence of an initial move is called an \emph{\bfseries initial occurrence}. 
\end{definition}

\begin{notation}
Let $G$ be an arena, and $\bm{s} \in \mathscr{L}_G$.
$\mathsf{InitOcc}(\bm{s})$ denotes the set of all initial occurrences in $\bm{s}$.
We write $\bm{s} \upharpoonright I$, where $I \subseteq \mathsf{InitOcc}(\bm{s})$, for the subsequence of $\bm{s}$ consisting of threads of initial occurrences in $I$, but we rather write $\bm{s} \upharpoonright m_0$ for $\bm{s} \upharpoonright \{ m_0 \}$.
\end{notation}

We are now ready to define the notion of \emph{games}:
\begin{definition}[Games \cite{abramsky1999game,mccusker1998games}]
\label{Games}
A \emph{\bfseries game} is a quadruple $G = (M_G, \lambda_G, \vdash_G, P_G)$ such that the triple $(M_G, \lambda_G, \vdash_G)$ forms an arena (also denoted by $G$), and $P_G$ is a subset of $\mathscr{L}_G$ whose elements are called \emph{\bfseries (valid) positions} in $G$ satisfying:
\begin{itemize}
\item \textsf{(V1)} $P_G$ is non-empty and ``\emph{prefix-closed}'': $\bm{s}m \in P_G \Rightarrow \bm{s} \in P_G$
\item \textsf{(V2)} If $\bm{s} \in P_G$ and $I \subseteq \mathsf{InitOcc}(\bm{s})$, then $\bm{s} \upharpoonright I \in P_G$.

\end{itemize}
A \emph{\bfseries play} in $G$ is a (finite or infinite) sequence $\bm{\epsilon}, m_1, m_1m_2, \dots$ of positions in $G$.
\end{definition}

\begin{convention}
For technical convenience, we assume, without loss of any important generality, that every game $G$ is \emph{\bfseries economical}: Every move $m \in M_G$ appears in a position in $G$, and every ``enabling pair'' $m \vdash_G n$ occurs as a non-initial move $n$ and its justifier $m$ in a position in $G$.
Consequently, a game $G$ is completely determined by the set $P_G$ of its positions. 
\end{convention}

The following is a natural ``substructure-relation'' between games:
\begin{definition}[Subgames \cite{yamada2016game}]
\label{Subgames}
A \emph{\bfseries subgame} of a game $G$ is a game $H$, written $H \trianglelefteq G$, that satisfies $M_H \subseteq M_G$, $\lambda_H = \lambda_G \upharpoonright M_H$, $\vdash_H \ \subseteq \ \vdash_G \cap \ ((\{ \star \} + M_H) \! \times \! M_H )$, and $P_H \subseteq P_G$.
\end{definition}

Later, we shall focus on games that satisfy the following two conditions:
\begin{definition}[Well-openness \cite{abramsky1999game,mccusker1998games}]
A game $G$ is \emph{\bfseries well-opened (wo)} if $\bm{s} m \in P_G$ with $m$ initial implies $\bm{s} = \bm{\epsilon}$.
\end{definition}

\begin{definition}[Well-foundness \cite{clairambault2010totality}]
A game $G$ is \emph{\bfseries well-founded (wf)} if so is the enabling relation $\vdash_G$, i.e., there is no infinite sequence $\star \vdash_G m_0 \vdash_G m_1 \vdash_G m_2 \dots$ of ``enabling pairs''.
\end{definition}

Next, let us recall the standard constructions on games.
\begin{notation}
For brevity, we usually omit the ``tags'' for disjoint union of sets of moves in the following constructions, e.g., we write $a \in A + B$, $b \in A + B$ if $a \in A$, $b \in B$, and given relations $R_A \subseteq A \times A$, $R_B \subseteq B \times B$, we write $R_A + R_B$ for the relation on $A + B$ such that $(x, y) \in R_A + R_B \stackrel{\mathrm{df. }}{\Leftrightarrow} (x, y) \in R_A \vee (x, y) \in R_B$, and so on. 
However, in contrast, we explicitly define tags for certain disjoint union in Section~\ref{PredicativeGames} because they are a part of our proposed structures and play an important role in this paper.
\end{notation}

\begin{definition}[Tensor product \cite{abramsky1999game,mccusker1998games}]
Given games $A, B$, we define their \emph{\bfseries tensor product} $A \otimes B$ as follows:  
\begin{itemize}

\item $M_{A \otimes B} \stackrel{\mathrm{df. }}{=} M_A + M_B$ 

\item $\lambda_{A \otimes B} \stackrel{\mathrm{df. }}{=} [\lambda_A, \lambda_B]$

\item $\vdash_{A \otimes B} \stackrel{\mathrm{df. }}{=} \ \! \vdash_A + \vdash_B$ 

\item $P_{A \otimes B} \stackrel{\mathrm{df. }}{=} \{ \bm{s} \in L_{A \otimes B} \ \! | \ \! \bm{s} \upharpoonright A \in P_A, \bm{s} \upharpoonright B \in P_B \}$
\end{itemize}
where $\bm{s} \upharpoonright A$ (resp. $\bm{s} \upharpoonright B$) denotes the subsequence of $\bm{s}$ that consists of moves of $A$ (resp. $B$) equipped with the justifiers in $\bm{s}$.
\end{definition}

\begin{definition}[Linear implication \cite{abramsky1999game,mccusker1998games}]
Given games $A, B$, we define their \emph{\bfseries linear implication} $A \multimap B$ as follows:
\begin{itemize}

\item $M_{A \multimap B} \stackrel{\mathrm{df. }}{=} M_{A} + M_B$

\item $\lambda_{A \multimap B} \stackrel{\mathrm{df. }}{=} [\overline{\lambda_A}, \lambda_B]$, where $\overline{\lambda_A} \stackrel{\mathrm{df. }}{=} \langle \overline{\lambda_A^\textsf{OP}}, \lambda_A^\textsf{QA} \rangle$ and $\overline{\lambda_A^\textsf{OP}} (m) \stackrel{\mathrm{df. }}{=} \begin{cases} \mathsf{P} \ &\text{if $\lambda_A^\textsf{OP} (m) = \textsf{O}$} \\ \textsf{O} \ &\text{otherwise} \end{cases}$

\item $\star \vdash_{A \multimap B} m \stackrel{\mathrm{df. }}{\Leftrightarrow} \star \vdash_B m$

\item $m \vdash_{A \multimap B} n \ (m \neq \star) \stackrel{\mathrm{df. }}{\Leftrightarrow} (m \vdash_A n) \vee (m \vdash_B n) \vee (\star \vdash_B m \wedge \star \vdash_A n)$ 

\item $P_{A \multimap B} \stackrel{\mathrm{df. }}{=} \{ \bm{s} \in L_{A \multimap B} \ \! | \! \ \bm{s} \upharpoonright A \in P_A, \bm{s} \upharpoonright B \in P_B \}$.

\end{itemize} 
\end{definition}

\begin{definition}[Product \cite{abramsky1999game,mccusker1998games}]
Given games $A, B$, we define their \emph{\bfseries product} $A \& B$ as follows:
\begin{itemize}

\item $M_{A \& B} \stackrel{\mathrm{df. }}{=} M_A + M_B$

\item $\lambda_{A \& B} \stackrel{\mathrm{df. }}{=} [\lambda_A, \lambda_B]$

\item $\vdash_{A \& B} \stackrel{\mathrm{df. }}{=} \ \! \vdash_A + \vdash_B$ 

\item $P_{A \& B} \stackrel{\mathrm{df. }}{=} \{ \bm{s} \in L_{A \& B} \ \! | \ \! \bm{s} \upharpoonright A \in P_A, \bm{s} \upharpoonright B = \bm{\epsilon} \ \! \} \cup \{ \bm{s} \in L_{A \& B} \ \! | \ \! \bm{s} \upharpoonright A = \bm{\epsilon}, \bm{s} \upharpoonright B \in P_B \}$.
\end{itemize}
\end{definition}

\begin{definition}[Exponential \cite{abramsky1999game,mccusker1998games}]
For any game $A$, we define its \emph{\bfseries exponential} $!A$ as follows: The arena $!A$ is just the arena $A$, and $P_{!A} \stackrel{\mathrm{df. }}{=} \{ \bm{s} \in L_{!A} \ \! | \ \! \forall m \in \mathsf{InitOcc}(\bm{s}) . \ \! \bm{s} \upharpoonright m \in P_A \ \! \}$.
\end{definition}

In addition, \cite{yamada2016game} has introduced ``composition of games'', which is a natural generalization of the composition of strategies \cite{abramsky1999game,mccusker1998games}:
\begin{definition}[Composition of games \cite{yamada2016game}]
\label{DefCompositionOfGames}
Given games $A, B, C$, $J \trianglelefteq A \multimap B^{[1]}$, $K \trianglelefteq B^{[2]} \multimap C$, where the superscripts $[1], [2]$ are to distinguish two copies of $B$, we define their \emph{\bfseries composition} $K \circ J$ (or written $J ; K$) by:
\begin{itemize}

\item $M_{K \circ J} \stackrel{\mathrm{df. }}{=} (M_J \upharpoonright A) + (M_K \upharpoonright C)$

\item $\lambda_{K \circ J} \stackrel{\mathrm{df. }}{=} [\lambda_J \upharpoonright M_A, \lambda_K \upharpoonright M_C]$ 

\item $\star \vdash_{K \circ J} m \stackrel{\mathrm{df. }}{\Leftrightarrow} \star \vdash_K m$

\item $m \vdash_{K \circ J} n \ (m \neq \star) \stackrel{\mathrm{df. }}{\Leftrightarrow} m \vdash_J n \vee m \vdash_K n \vee \exists b \in M_{B} . \ \! m \vdash_{K} b^{[2]} \wedge b^{[1]} \vdash_J n$

\item $P_{K \circ J} \stackrel{\mathrm{df. }}{=} \{ \bm{s} \upharpoonright A, C \ \! | \ \! \bm{s} \in P_J \ddagger P_K \}$ 

\end{itemize} 
where $P_J \ddagger P_K \stackrel{\mathrm{df. }}{=} \{ \bm{s} \in (M_J + M_K)^\ast \ \! | \ \! \bm{s} \upharpoonright J \in P_J, \bm{s} \upharpoonright K \in P_K, \bm{s} \upharpoonright B^{[1]}, B^{[2]} \in \mathsf{pr}_B \}$, $\mathsf{pr}_B \stackrel{\mathrm{df. }}{=} \{ \bm{s} \in P_{B^{[1]} \multimap B^{[2]}} \ \! | \ \! \forall \bm{t} \preceq{\bm{s}}. \ \mathsf{even}(\bm{t}) \Rightarrow \bm{t} \upharpoonright B^{[1]} = \bm{t} \upharpoonright B^{[2]} \}$, $b^{[1]} \in M_{B^{[1]}}, b^{[2]} \in M_{B^{[2]}}$ are $b \in M_B$ equipped with respective tags for $B^{[1]}, B^{[2]}$, $M_J \upharpoonright A \stackrel{\mathrm{df. }}{=} \{ m \in M_J \ \! | \ \! m \upharpoonright A \neq \bm{\epsilon} \ \! \}$, $M_K \upharpoonright C \stackrel{\mathrm{df. }}{=} \{ n \in M_K \ \! | \ \! n \upharpoonright C \neq \bm{\epsilon} \ \! \}$, and $\bm{s} \upharpoonright A, C$ is $\bm{s} \upharpoonright M_A + M_C$ equipped with the pointer defined by $m \leftarrow n$ in $\bm{s} \upharpoonright A, C$ if and only if $\exists m_1, m_2, \dots, m_k \in M_{((A \multimap B^{[1]}) \multimap B^{[2]}) \multimap C} \setminus M_{A \multimap C} . \ \! m \leftarrow m_1 \leftarrow m_2 \leftarrow \dots \leftarrow m_k \leftarrow n$ in $\bm{s}$ ($\bm{s} \upharpoonright B^{[1]}, B^{[2]}$ is defined analogously).
\end{definition}

It has been shown in \cite{abramsky1999game,mccusker1998games,yamada2016game} that these constructions are all well-defined.

\subsubsection{Strategies}
Let us turn our attention to strategies.
A \emph{strategy} on a game, roughly, is what tells Player which move she should make next at each of her turns in the game, and so in game semantics it interprets a \emph{term} of the type (resp. a \emph{proof} of the proposition) which the game interprets. 

In conventional game semantics, a strategy on a game $G$ is defined as a certain set of even-length positions in $G$ \cite{abramsky1999game,mccusker1998games}.
However, it prevents us from talking about strategies \emph{without underlying games}.
To overcome this point, \cite{yamada2016game} reformulates strategies as follows:

\begin{definition}[Strategies \cite{yamada2016game}]
\label{DefStrategiesAsWellOpenedDeterministicGames}
A \emph{\bfseries strategy} is a wo-game $\sigma$ that is ``\emph{deterministic}'': $\bm{s}mn, \bm{s}mn' \in P_\sigma^{\mathsf{even}}$ implies $\bm{s}mn = \bm{s}mn'$. 
Given a wo-game $G$, we say that $\sigma$ is \emph{on} $G$ and write $\sigma : G$ if $\sigma \trianglelefteq G$ and $P_\sigma$ is ``\emph{O-inclusive}'' with respect to $P_G$: $\bm{s} \in P_\sigma^{\mathsf{even}} \wedge \bm{s}m \in P_G$ implies $\bm{s}m \in P_\sigma$.
\end{definition} 

This notion of strategies corresponds to conventional strategies on wo-games by the \emph{\bfseries canonical bijection} $\Phi : \sigma \stackrel{\sim}{\mapsto} P_\sigma^{\mathsf{even}}$, and $\sigma : G$ in the sense of Definition~\ref{DefStrategiesAsWellOpenedDeterministicGames} if and only if $\Phi(\sigma) : G$ in the conventional sense \cite{abramsky1999game,mccusker1998games}, under the assumption that every game is economical (see \cite{yamada2016game} for the proof of this fact). 
We need to focus on \emph{wo}-games here since otherwise the inverse $\Phi^{-1}$ is not well-defined due to the axiom \textsf{V2} of games (though in any case objects of the cartesian closed category of games are usually wo \cite{abramsky1999game,mccusker1998games,abramsky2000full,hyland2000full}).

Moreover, we may reformulate standard constraints and constructions on strategies:
\begin{definition}[Constraints on strategies \cite{yamada2016game}]
A strategy $\sigma : G$ is:
\begin{itemize}

\item \emph{\bfseries innocent} if $\bm{s}mn, \bm{t} \in P_\sigma^{\mathsf{even}} \wedge \bm{t} m \in P_G \wedge \lceil \bm{t} m \rceil_G = \lceil \bm{s} m \rceil_G$ implies $\bm{t}mn \in P_\sigma$

\item \emph{\bfseries well-bracketed (wb)} if, whenever $\bm{s} q \bm{t} a \in P_\sigma^{\mathsf{even}}$, where $q$ is a question that justifies an answer $a$, every question in $\bm{\tilde{t}}$, where $\lceil \bm{s} q \bm{t} \rceil_G = \bm{\tilde{s}} q \bm{\tilde{t}}$, justifies an answer in $\bm{\tilde{t}}$

\item \emph{\bfseries total} if $\bm{s} \in P_\sigma^{\mathsf{even}} \wedge \bm{s} m \in P_G$ implies $\bm{s} m n \in P_\sigma$ for some (unique) $n \in M_G$

\item \emph{\bfseries noetherian} if $P_\sigma$ does not contain any strictly increasing (with respect to $\preceq$) infinite sequence of P-views of positions in $G$.

\end{itemize}
\end{definition}

\begin{definition}[Constructions on strategies \cite{yamada2016game}]
The \emph{\bfseries composition} $\circ$, \emph{\bfseries tensor (product)} $\otimes$, \emph{\bfseries pairing} $\langle \_, \_ \rangle$ and \emph{\bfseries promotion} $(\_)^\dagger$ of strategies are defined to be the composition $\circ$, tensor product $\otimes$, product $\&$ and exponential $!$ of games, respectively. 
\end{definition}

\begin{definition}[Copy-cats and derelictions \cite{yamada2016game}]
Given a game $A$, the \emph{\bfseries copy-cat} $\mathit{cp}_A$ is the strategy on $A \multimap A$ whose arena is $A \multimap A$ and positions are given by $P_{\mathit{cp}_A} \stackrel{\mathrm{df. }}{=} \mathsf{pr}_A$.
If $A$ is wo, the \emph{\bfseries dereliction} $\mathit{der}_A : \ !A \multimap A$ is $\mathit{cp}_A$ up to tags for the disjoint union for $!A$.
\end{definition}

These constrains and constructions on strategies coincide with the standard ones \cite{abramsky1999game,clairambault2010totality} with respect to the canonical bijection $\Phi$; see \cite{yamada2016game} for the proofs.

Note that \textsf{MLTT} is a \emph{total} type theory, i.e., its computation terminates in a finite period of time; thus it makes sense to focus on total strategies. 
However, it is well-known that totality of strategies is not preserved under composition due to the ``\emph{infinite chattering}'' between strategies \cite{clairambault2010totality}. For this problem, we further impose \emph{noetherianity} \cite{clairambault2010totality} on strategies.

Nevertheless, the dereliction $\mathit{der}_A$, the identity on each object $A$ in our category of games, is in general not noetherian.
This motivates us to focus on wf-games because:

\begin{lemma}[Well-defined derelictions \cite{yamada2016game}]
\label{LemWellDefinedDerelictions}
For any wo-game $A$, the dereliction $\mathit{der}_A$ is an innocent, wb and total strategy on $!A \multimap A$. It is noetherian if $A$ is additionally wf.
\end{lemma}
\begin{proof}
We just show that $\mathit{der}_A$ is noetherian if $A$ is wf, as the other statements are trivial. For any $\bm{s}mm \in \mathit{der}_A$, it is easy to see by induction on the length of $\bm{s}$ that  the P-view $\lceil \bm{s} m \rceil_{!A \multimap A}$ is of the form $m_1 m_1 m_2 m_2 \dots m_k m_k m$, and there is a sequence $\star \vdash_A m_1 \vdash_A m_2 \dots \vdash_A m_k \vdash_A m$. Therefore if $A$ is wf, then $\mathit{der}_A$ must be noetherian. 
\end{proof}

In addition to totality and noetherianity, we shall later impose \emph{innocence} and \emph{well-bracketing} on strategies since \textsf{MLTT} is a \emph{functional} programming language \cite{abramsky1999game}.

The following two lemmata are immediate consequences of our definitions and established facts on (conventional) strategies in the literature.

\begin{lemma}[Well-defined composition \cite{yamada2016game}]
\label{LemWellDefinedComposition}
Given strategies $\sigma : A \multimap B$, $\tau : B \multimap C$, their composition $\tau \circ \sigma$ (or written $\sigma ; \tau$) forms a strategy on the game $A \multimap C$. If $\sigma$ and $\tau$ are both innocent, wb, total and noetherian, then so is $\tau \circ \sigma$.
\end{lemma}

\begin{lemma}[Well-defined pairing, tensor and promotion \cite{yamada2016game}]
\label{LemWellDefinedPairingTensorPromotion}
Given strategies $\sigma : C \multimap A$, $\tau : C \multimap B$, $\lambda : A \multimap C$, $\gamma : B \multimap D$, $\phi : \ !A \multimap B$, the pairing $\langle \sigma, \tau \rangle$, tensor $\lambda \otimes \gamma$ and promotion $\phi^\dagger$ form strategies on $C \multimap A \& B$, $A \otimes B \multimap C \otimes D$ and $!A \multimap B$, respectively. 
They are innocent (resp. wb, total, noetherian) if so are the respective component strategies.
\end{lemma}

\subsection{Review: Predicative Games}
\label{PredicativeGames}
Now, we are ready to recall the variant of games and strategies in \cite{yamada2016game} that has established a game semantics for \textsf{MLTT} with the (cumulative hierarchy of) universes. 

\subsubsection{Motivating Observation}
Let us first give a characterization of wo-games as \emph{sets of strategies with some constraint} as it motivates and justifies the notion of predicative games introduced in the next section.

Then what constraint do we need? First, strategies on the same game must be \emph{consistent}:
\begin{definition}[Consistency \cite{yamada2016game}]
A set $\mathcal{S}$ of strategies is \emph{\bfseries consistent} if, for any $\sigma, \tau \in \mathcal{S}$, (i) $\lambda_\sigma (m) = \lambda_\tau (m)$ for all $m \in M_\sigma \cap M_\tau$; (ii) $\star \vdash_\sigma m \Leftrightarrow \star \vdash_\tau m$ and $m \vdash_\sigma n \Leftrightarrow m \vdash_\tau n$ for all $m, n \in M_\sigma \cap M_\tau$; and (iii) $\bm{s} m \in P_\sigma \Leftrightarrow \bm{s} m \in P_\tau$ for all $\bm{s} \in (P_\sigma \cap P_\tau)^{\mathsf{even}}$, $\bm{s} m \in P_\sigma \cup P_\tau$.
\end{definition}

Such a set $\mathcal{S}$ induces a game $\bigcup \mathcal{S} \stackrel{\mathrm{df. }}{=} (\bigcup_{\sigma \in \mathcal{S}}M_\sigma, \bigcup_{\sigma \in \mathcal{S}}\lambda_\sigma, \bigcup_{\sigma \in \mathcal{S}} \! \vdash_{\sigma}, \bigcup_{\sigma \in \mathcal{S}}P_\sigma)$ thanks to the first two conditions. 
Conversely, $\bigcup \mathcal{S}$ is not well-defined if $\mathcal{S}$ does not satisfy either. 
The third condition ensures the ``consistency of possible O-positions'' among strategies in $\mathcal{S}$.

However, some strategies on $\bigcup \mathcal{S}$ may not exist in $\mathcal{S}$. For this, we further require:
\begin{definition}[Completeness \cite{yamada2016game}]
A consistent set $\mathcal{S}$ of strategies is \emph{\bfseries complete} if, for any subset $\mathcal{A} \subseteq \bigcup_{\sigma \in \mathcal{S}}P_\sigma$ that satisfies $\mathcal{A} \neq \emptyset \wedge \forall \bm{s} \in \mathcal{A} . \ \! \bm{t} \preceq \bm{s} \Rightarrow \bm{t} \in \mathcal{A}$, $\forall \bm{s}mn, \bm{s}mn' \in \mathcal{A}^{\mathsf{even}} . \ \! smn = smn'$, and $\forall \bm{s} \in \mathcal{A}^{\mathsf{even}} . \ \! \bm{s}m \in \bigcup_{\sigma \in \mathcal{S}}P_\sigma \Rightarrow \bm{s}m \in \mathcal{A}$, the strategy $\Phi^{-1}(\mathcal{A}^{\mathsf{even}})$ exists in $\mathcal{S}$.
\end{definition}

Intuitively, the completeness of a set $\mathcal{S}$ of strategies means its closure under the ``patchwork of strategies'' $\mathcal{A} \subseteq \bigcup_{\sigma \in \mathcal{S}}P_\sigma$.
For example, consider a consistent set $\mathcal{S} = \{ \sigma, \tau \}$ of strategies $\sigma = \mathsf{pref}(\{ ac, bc \})$, $\tau = \mathsf{pref}(\{ ad, bd \})$ (for brevity, here we specify strategies by their positions).
Then we may take a subset $\phi = \mathsf{pref}(\{ ac, bd \}) \subseteq \sigma \cup \tau$, but it does not exist in $\mathcal{S}$, showing that $\mathcal{S}$ is not complete. 
Note that there are total nine strategies on the game $\bigcup \mathcal{S}$, and so we need to add $\phi$ and the remaining six to $\mathcal{S}$ to make it complete.

Now, we have arrived at the desired characterization:
\begin{theorem}[Games as collections of strategies \cite{yamada2016game}]
\label{ThmGamesAsCollectionsOfStrategies}
There is a one-to-one correspondence between wo-games and complete sets of strategies:
\begin{enumerate}

\item For any wo-game $G$, the set $\{ \sigma \ \! | \ \! \sigma : G \ \! \}$ is complete, and $G = \bigcup \{ \sigma \ \! | \ \! \sigma : G \ \! \}$.

\item For any complete set $\mathcal{S}$ of strategies, the strategies on $\bigcup \mathcal{S}$ are precisely the ones in $\mathcal{S}$.

\end{enumerate}
\end{theorem}

Thus, any wo-game is of the form $\bigcup \mathcal{S}$ with its strategies in $\mathcal{S}$, where $\mathcal{S}$ is complete. 
Importantly, there is no essential difference between $\bigcup \mathcal{S}$ and $\Sigma \mathcal{S}$ defined by:
\begin{itemize}

\item $M_{\Sigma \mathcal{S}} \stackrel{\mathrm{df. }}{=} \{ q_{\mathcal{S}} \} \cup \mathcal{S} \cup \{ (m, \sigma) \ \! | \ \! \sigma \in \mathcal{S}, m \in M_\sigma \}$, where $q_{\mathcal{S}}$ is any element

\item $\lambda_{\Sigma \mathcal{S}} : q_{\mathcal{S}} \mapsto \mathsf{OQ}, (\sigma \in \mathcal{S}) \mapsto \mathsf{PA}, (m, \sigma) \mapsto \lambda_\sigma(m)$

\item $\vdash_{\Sigma \mathcal{S}} \ \stackrel{\mathrm{df. }}{=} \{ (\star, q_{\mathcal{S}}) \} \cup \{ (q_{\mathcal{S}}, \sigma) \ \! | \ \! \sigma \in \mathcal{S} \ \!  \} \cup \{ (\sigma, (m, \sigma)) \ \! | \ \! \sigma \in \mathcal{S}, \star \vdash_{\sigma} m \ \! \} \\ \cup \{ ((m, \sigma), (n, \sigma)) \ \! | \ \! \sigma \in \mathcal{S}, m \vdash_\sigma n \ \! \}$

\item $P_{\Sigma \mathcal{S}} \stackrel{\mathrm{df. }}{=} \mathsf{pref} (\{ q_{\mathcal{S}} . \sigma . (m_1, \sigma) . (m_2, \sigma) \dots (m_k, \sigma) \ \! | \! \ \sigma \in \mathcal{S}, m_1 . m_2 \dots m_k \in P_\sigma \})$.

\end{itemize}

A position in $\Sigma \mathcal{S}$ is essentially a position $\bm{s}$ in $\bigcup \mathcal{S}$ equipped with the initial two moves $q_{\mathcal{S}} . \sigma$ and the tag $(\_, \sigma)$ on subsequent moves, where $\sigma$ is any (not unique) $\sigma \in \mathcal{S}$ such that $\bm{s} \in P_\sigma$; the difference between $\bigcup \mathcal{S}$ and $\Sigma \mathcal{S}$ is whether to specify such $\sigma$ for each $\bm{s} \in P_{\bigcup \mathcal{S}}$.

Intuitively, a play in the game $\Sigma \mathcal{S}$ proceeds as follows. \emph{Judge} of the game first asks Player about her strategy in mind, and Player answers a strategy $\sigma \in \mathcal{S}$; and then an actual play between Opponent and Player begins as in $\bigcup \mathcal{S}$ except that Player must follow $\sigma$. 

To be fair, the declared strategy should be ``invisible'' to Opponent, and he has to declare an ``\emph{anti-strategy}'' beforehand, which is ``invisible'' to Player, and plays by following it as well.
To be precise, an \emph{\bfseries anti-strategy} $\tau$ on a game $G$ is a subgame $\tau \trianglelefteq G$ that is ``\emph{P-inclusive}'' (dual to ``O-inclusive'') with respect to $P_G$ and ``\emph{deterministic}'' on odd-length positions. 
Clearly, we may achieve the ``invisibility'' of Player's strategy to Opponent by requiring that \emph{anti-strategies cannot ``depend on the tags''}, i.e., any anti-strategy $\tau$ on $\Sigma \mathcal{S}$ must satisfy 
\begin{equation*}
q_{\mathcal{S}} . \sigma_1 . (m_1, \sigma_1) . (m_2, \sigma_1) \dots (m_{2k+1}, \sigma_1) \in P_\tau \Leftrightarrow q_{\mathcal{S}} . \sigma_2 . (m_1, \sigma_2) . (m_2, \sigma_2) \dots (m_{2k+1}, \sigma_2) \in P_\tau 
\end{equation*}
for any $\sigma_1, \sigma_2 \in \mathcal{S}$, $q_{\mathcal{S}} . \sigma_i . (m_1, \sigma_i) . (m_2, \sigma_i) \dots (m_{2k}, \sigma_i) \in P_\tau$ for $i = 1, 2$, $m_{2k+1} \in M_{\sigma_1} \cup M_{\sigma_2}$.
However, since the ``spirit'' of game semantics is not to restrict Opponent's computational power at all, we choose not to incorporate the declaration of anti-strategies or their ``invisibility condition'' into games. 
Consequently, Player cannot see Opponent's ``declaration'' either. 

Therefore we may reformulate any wo-game in the form $\Sigma \mathcal{S}$, where $\mathcal{S}$ is a complete set of strategies.
Moreover, $\Sigma \mathcal{S}$ is a generalization of a wo-game if we do not require the completeness of $\mathcal{S}$; such $\mathcal{S}$ may not have all the strategies on $\bigcup \mathcal{S}$. Furthermore, since we take a disjoint union of sets of moves for $\Sigma \mathcal{S}$, \emph{it is trivially a well-defined game even if we drop the consistency of $\mathcal{S}$}; then $\Sigma \mathcal{S}$ can be thought of as a ``family of games'' as its strategies may have different underlying games, in which Player has an additional opportunity to declare a strategy that simultaneously specifies a component game to play (see Example~\ref{ExFamilyOfGames} below). 

This is the idea behind the notion of \emph{predicative games}: A predicative game, roughly, is a game of the form $\Sigma \mathcal{S}$, where $\mathcal{S}$ is a (not necessarily complete or consistent) set of strategies.

However, assuming there is a \emph{name} $|G|$ of each game $G$ and a strategy $\underline{G}$ given by $P_{\underline{G}} \stackrel{\mathrm{df. }}{=} \mathsf{pref}(\{ q . |G| \})$, we may form a set $\mathcal{P}$ of strategies by $\mathcal{P} \stackrel{\mathrm{df. }}{=} \{ \underline{G} \ \! | \ \! \text{$G$ is a game}, (|G|, \underline{G}) \not \in M_G \}$. Then the induced game $\mathscr{P} \stackrel{\mathrm{df. }}{=} \Sigma \mathcal{P}$ gives rise to a \emph{Russell-like paradox}: If $(|\mathscr{P}|, \underline{\mathscr{P}}) \in M_{\mathscr{P}}$, then $(|\mathscr{P}|, \underline{\mathscr{P}}) \not \in M_{\mathscr{P}}$, and vice versa.
Our solution is the \emph{ranks} of games in the next section.

\subsubsection{Predicative Games}
Now, let us proceed to define the central notion of \emph{predicative games}.

\begin{definition}[Ranked moves \cite{yamada2016game}]
A move of a game is \emph{\bfseries ranked} if it is a pair $(m, r)$ of some object $m$ and a natural number $r \in \mathbb{N}$, which is usually written $[m]_r$.
A ranked move $[m]_r$ is more specifically called an \emph{\bfseries $\bm{r}^{\textit{th}}$-rank move}, and $r$ is said to be the \emph{\bfseries rank} of the move.
\end{definition}

\begin{notation}
For a sequence $[\bm{s}]_{\bm{r}} = [m_1]_{r_1} . [m_2]_{r_2} \dots [m_k]_{r_k}$ of ranked moves and an element $\square$, we define $[\bm{s}]_{\bm{r}}^\square \stackrel{\mathrm{df. }}{=} [m_1]_{r_1}^\square . [m_2]_{r_2}^\square \dots [m_k]_{r_k}^\square \stackrel{\mathrm{df. }}{=} [(m_1, \square)]_{r_1} . [(m_2, \square)]_{r_2} \dots [(m_k, \square)]_{r_k}$. 
\end{notation}

Our intention is as follows: A $0^{\text{th}}$-rank move is just a move of a game in the conventional sense, and an $(r+1)^{\text{st}}$-rank move is the \emph{name} of another game (whose moves are all ranked) such that the supremum of the ranks of its moves is $r$:
\begin{definition}[Ranked games \cite{yamada2016game}]
\label{RankedGames}
A \emph{\bfseries ranked game} is a game whose moves are all ranked.  
The \emph{\bfseries rank} $\mathcal{R}(G)$ of a ranked game $G$ is defined by $\mathcal{R}(G) \stackrel{\mathrm{df. }}{=}\mathsf{sup}(\{r \ \! | \! \ [m]_r \! \in M_G \}) + 1$ if $M_G \neq \emptyset$, and $\mathcal{R}(G) \stackrel{\mathrm{df. }}{=} 1$ otherwise.
$G$ is particularly called an \emph{\bfseries $\bm{\mathcal{R}(G)}^{\textit{th}}$-rank game}.
\end{definition}

\begin{definition}[Name of games \cite{yamada2016game}]
The \emph{\bfseries name} of a ranked game $G$, written $\mathcal{N}(G)$, is the pair $[G] _{\mathcal{R}(G)}$ of $G$ (as a set) itself and its rank $\mathcal{R}(G)$ if $\mathcal{R}(G) \in \mathbb{N}$, and undefined otherwise.
\end{definition}

\begin{remark}
As we shall see shortly, the rank of each predicative game is finite.
\end{remark}

The name of a ranked game can be a move of a ranked game, but that name cannot be a move of the game itself by its rank, which prevents the paradox described above.

\begin{definition}[Predicative games \cite{yamada2016game}]
\label{DefPredicativeGames}
For each integer $k \geqslant 1$, a \emph{\bfseries $\bm{k}$-predicative game} is a $k^{\text{th}}$-rank game $G$ equipped with a set $\mathsf{st}(G)$ of ranked strategies $\sigma$ with $M_\sigma \subseteq (\mathbb{N} \times \{ 0 \}) \cup \{ \mathcal{N}(H) \ \! | \ \! \textit{$H$ is an $l$-predicative game}, l < k \ \! \}$, that satisfies:
\begin{itemize}

\item $M_G = \Sigma_{\sigma \in \mathsf{st}(G)}M_\sigma \stackrel{\mathrm{df. }}{=} \{ [m]_r^\sigma \ \! | \ \! \sigma \in \mathsf{st}(G), [m]_r \in M_\sigma \}$; $\lambda_G : [m]_r^\sigma \mapsto \lambda_\sigma([m]_r)$

\item $\vdash_G \ = \{ (\star, [m]_r^\sigma) \ \! | \ \! \sigma \in \mathsf{st}(G), \star \vdash_\sigma [m]_r \} \cup \{ ([m]_r^\sigma, [m']_{r'}^\sigma) \ \! | \ \! \sigma \in \mathsf{st}(G), [m]_r \vdash_\sigma [m']_{r'} \}$

\item $P_G = \mathsf{pref} (\{ q_G . \mathcal{N}(\sigma) . [\bm{s}]_{\bm{r}}^\sigma \ \! | \! \ \sigma \in \mathsf{st}(G), [\bm{s}]_{\bm{r}} \in P_\sigma \})$, where $q_G \stackrel{\mathrm{df. }}{=} [0]_0$. 

\end{itemize}
A \emph{\bfseries predicative game} is a $k$-predicative game for some $k \geqslant 1$. A \emph{\bfseries strategy on a predicative game} $G$ is any element in $\mathsf{st}(G)$, and $\sigma : G$ denotes $\sigma \in \mathsf{st}(G)$. 
\end{definition}

\begin{notation}
We write $\mathcal{PG}_k$ (resp. $\mathcal{PG}_{\leqslant k}$) for the set of all $k$-predicative games (resp. $i$-predicative games with $1 \leqslant i \leqslant k$). Similar notations $\mathcal{ST}_k, \mathcal{ST}_{\leqslant k}$ are used for strategies.
\end{notation}

That is, predicative games $G$ are essentially the games $\Sigma \mathsf{st}(G)$ inductively defined along with their ranks except that the elements $g_G, \mathcal{N}(\sigma)$ are not included in $M_G$.

\begin{remark}
Strictly speaking, a predicative game $G$ is not a game because the elements $q_G, \mathcal{N}(\sigma)$ are not counted as moves of $G$, they do not have labels, and $\mathcal{N}(\sigma)$ occurs in a position without a justifier. We have defined $G$ as above, however, for an ``\emph{initial protocol}'' $q_G . \mathcal{N}(\sigma)$ between Judge and Player is not an actual play between Opponent and Player, and it should not appear in O-views.
Also, it prevents the name $\mathcal{N}(\sigma)$ of a strategy $\sigma : G$ from affecting the rank $\mathcal{R}(G)$.
Except these points, $G$ is a particular type of a ranked wo-game.
\end{remark}

Intuitively, a play in a predicative game $G$ proceeds as follows. At the beginning, Player has an opportunity to ``declare'' a strategy $\sigma : G$ to Judge, and then a play between the participants follows, where Player is forced to play by $\sigma$. 
The point is that $\sigma : G$ may range over strategies on different games in the conventional game semantics sense, and so Player may \emph{choose an underlying game} when she selects a strategy.
As a consequence, a predicative game may be a ``family of games'' (see Ex.~\ref{ExFamilyOfGames} below) and interpret type dependency.

In light of Theorem.~\ref{ThmGamesAsCollectionsOfStrategies}, we may generalize the subgame relation as follows:
\begin{definition}[Subgames of predicative games \cite{yamada2016game}]
A \emph{\bfseries subgame} of a predicative game $G$ is a predicative game $H$ that satisfies $\mathsf{st}(H) \subseteq \mathsf{st}(G)$. In this case, we write $H \trianglelefteq G$.
\end{definition}

\begin{example}
The \emph{\bfseries flat game} $\mathit{flat}(A)$ on a set $A$ is given by:
\begin{itemize} 

\item $M_{N} \stackrel{\mathrm{df. }}{=} \{ q \} \cup A$, $q \stackrel{\mathrm{df. }}{=} 0$

\item $\lambda_{N} : q \mapsto \mathsf{OQ}, (a \in A) \mapsto \mathsf{PA}$

\item $\vdash_{N} \stackrel{\mathrm{df. }}{=} \{ (\star, q) \} \cup \{ (q, a) \ \! | \! \ a \in A \ \! \}$

\item $P_{N} \stackrel{\mathrm{df. }}{=} \mathsf{pref} (\{ q a \ \! | \! \ a \in A\ \! \})$.

\end{itemize}

Then the \emph{\bfseries natural number game} $N$ is defined by $N  \stackrel{\mathrm{df. }}{=} \mathit{flat}(\mathbb{N})$, where $\mathbb{N}$ is the set of all natural numbers.
For each $n \in \mathbb{N}$, let $\underline{n}$ denote the strategy on $N$ such that $P_{\underline{n}} = \mathsf{pref} (\{ q n \})$. 
A maximal position of the corresponding 1-predicative game $N_1$ is of the form 
\begin{equation*}
q_{N_1} . \mathcal{N}(\underline{n}_0) . [q]_0^{\underline{n}_0} . [n]_0^{\underline{n}_0}
\end{equation*}
where $\underline{n}_0$ is obtained from $\underline{n}$ by changing each move $m$ to the $0^{\text{th}}$-rank move $[m]_0$. 
For readability, we usually abbreviate the play as 
\begin{equation*}
q_{N_1} . \mathcal{N}(\underline{n}) . q . n
\end{equation*}
which is essentially the play $q n$ in $N$. 
Below, we usually abbreviate $N_1$ as $N$.

Also, there are the \emph{\bfseries unit game} $\bm{1} \stackrel{\mathrm{df. }}{=} \mathit{flat}(\{ \checkmark \})$, the \emph{\bfseries terminal game} $I \stackrel{\mathrm{df. }}{=} (\emptyset, \emptyset, \emptyset, \{ \bm{\epsilon} \})$, and the \emph{\bfseries empty game} $\bm{0} \stackrel{\mathrm{df. }}{=} \mathit{flat}(\emptyset)$.
We then define the obvious strategies $\underline{\checkmark} : \bm{1}$, $\_ : I$ $\bot : \bm{0}$, where $\underline{\checkmark}$ and $\_$ are total, while $\bot$ is not.
Again, abusing notation, we write $\bm{1}, I, \bm{0}$ for the corresponding 1-predicative games $\bm{1}_1, I_1, \bm{0}_1$, respectively. 
\end{example}

\begin{definition}[Parallel union \cite{yamada2016game}]
For $k \geqslant 1$, $\mathcal{S} \subseteq \mathcal{PG}_{\leqslant k}$, the \emph{\bfseries parallel union} $\int \! \mathcal{S}$ is given by: 
\begin{itemize}

\item $\mathsf{st}(\int \! \mathcal{S}) \stackrel{\mathrm{df. }}{=} \bigcup_{G \in \mathcal{S}} \mathsf{st}(G)$; $M_{\int \! \mathcal{S}} \! \stackrel{\mathrm{df. }}{=} \bigcup_{G \in \mathcal{S}} M_G$; $\lambda_{\int \! \mathcal{S}} : ([m]_r \in M_G) \mapsto \lambda_G([m]_r)$

\item $\vdash_{\int \! \mathcal{S}} \stackrel{\mathrm{df. }}{=} \{ (\star, [m]_r) \ \! | \ \! \exists G \in \mathcal{S}, \star \vdash_G [m]_r \} \cup \{ ([m]_r, [m']_{r'}) \ \! | \ \! \exists G \in \mathcal{S}, [m]_r \vdash_G [m']_{r'} \}$

\item $P_{\int \! \mathcal{S}} \! \stackrel{\mathrm{df. }}{=} \! \mathsf{pref} (\{ q_{\int \! \mathcal{S}} . \mathcal{N}(\sigma) . \bm{s} \ \! | \! \ \exists G \in \mathcal{S} . \ \! q_G . \mathcal{N}(\sigma) . \bm{s} \in P_G \})$, where $q_{\int \! \mathcal{S}} \stackrel{\mathrm{df. }}{=} [0]_0$.

\end{itemize}
\end{definition}

\begin{definition}[Predicative union \cite{yamada2016game}]
For $k \geqslant 1$, $\mathcal{S} \subseteq \mathcal{ST}_{\leqslant k}$, the \emph{\bfseries predicative union} $\oint \! \mathcal{S}$ is given by:
\begin{itemize}

\item $\mathsf{st}(\oint \! \mathcal{S}) \stackrel{\mathrm{df. }}{=} \mathcal{S}$; $M_{\oint \! \mathcal{S}} \! \stackrel{\mathrm{df. }}{=} \Sigma_{\sigma \in \mathcal{S}} M_\sigma$; $\lambda_{\oint \! \mathcal{S}} : [m]_r^\sigma \mapsto \lambda_\sigma([m]_r)$

\item $\vdash_{\oint \! \mathcal{S}} \stackrel{\mathrm{df. }}{=} \{ (\star, [m]_r^\sigma) \ \! | \ \! \sigma \in \mathcal{S}, \star \vdash_\sigma [m]_r \} \cup \{ ([m]_r^\sigma, [m']_{r'}^\sigma) \ \! | \ \! \sigma \in \mathcal{S}, [m]_r \vdash_\sigma [m']_{r'} \}$

\item $P_{\oint \! \mathcal{S}} \stackrel{\mathrm{df. }}{=} \mathsf{pref} (\{ q_{\oint \! \mathcal{S}} .  \mathcal{N}(\sigma) . [\bm{s}]_{\bm{r}}^\sigma \ \! | \! \ \sigma \in \mathcal{S} . \ \! [\bm{s}]_{\bm{r}} \in P_{\sigma} \})$, where $q_{\oint \! \mathcal{S}} \stackrel{\mathrm{df. }}{=} [0]_0$.

\end{itemize}
\end{definition}

Clearly, parallel and predicative unions are well-defined predicative games.

\begin{example}
\label{ExFamilyOfGames}
In the 1-predicative game $\oint \{ \underline{100}, \underline{\checkmark}, \_, \bot \}$, a play proceeds as either of:
\begin{center}
\begin{tabular}{cccccccccc}
$\oint \{ \underline{100}, \underline{\checkmark}, \_, \bot \}$ &&& $\oint \{ \underline{100}, \underline{\checkmark}, \_, \bot \}$ &&& $\oint \{ \underline{100}, \underline{\checkmark}, \_, \bot \}$ &&& $\oint \{ \underline{100}, \underline{\checkmark}, \_, \bot \}$ \\ \cline{1-1} \cline{4-4} \cline{7-7} \cline{10-10}
$q_{\oint \{ \underline{100}, \underline{\checkmark}, \_, \bot \}}$&&&$q_{\oint \{ \underline{100}, \underline{\checkmark}, \_, \bot \}}$&&&$q_{\oint \{ \underline{100}, \underline{\checkmark}, \_, \bot \}}$&&&$q_{\oint \{ \underline{100}, \underline{\checkmark}, \_, \bot \}}$ \\
$\mathcal{N}(\underline{100})$&&&$\mathcal{N}(\underline{\checkmark})$&&&$\mathcal{N}(\_)$&&&$\mathcal{N}(\bot)$ \\
$q$&&&$q$&&&&&&$q$ \\
$100$&&&$\checkmark$&&&&&&
\end{tabular}
\end{center}
This game illustrates the point that a predicative game can be a ``family of games''.
\end{example}

Let us now define a particular kind of predicative games to interpret universes:
\begin{definition}[Universe games \cite{yamada2016game}]
For each $k \in \mathbb{N}$, we define the \emph{\bfseries $\bm{k}^\textit{th}$-universe game} $\mathcal{U}_k$ by $\mathcal{U}_k \stackrel{\mathrm{df. }}{=} \oint \{ \underline{G} \ \! | \ \! G \in \mathcal{PG}_{\leqslant k+1} \}$, where $\underline{G} \stackrel{\mathrm{df. }}{=} \mathit{flat}(\{ \mathcal{N}(G) \})$.
A \emph{\bfseries universe game} is the $k^\text{th}$-universe game $\mathcal{U}_k$ for some $k \in \mathbb{N}$, and we often write it by $\mathcal{U}$ when $k \in \mathbb{N}$ is not very important.
\end{definition}

\begin{notation}
Given total $\mu : \mathcal{U}$, we write $\mathit{El}(\mu)$ for the predicative game such that $\underline{\mathit{El}(\mu)} = \mu$
\end{notation}

As $G \in \mathcal{PG}_{\leqslant k+1} \Leftrightarrow \underline{G} : \mathcal{U}_k$, $\mathcal{U}_k$ is a ``universe'' of all $i$-predicative games with $1 \leqslant i \leqslant k+1$.
Also, we obtain a cumulative hierarchy: $\underline{\mathcal{U}_i} : \mathcal{U}_j$ if $i < j$.
An ``actual play'' in $\mathcal{U}$ starts with the question $q$ \emph{``What is your game?''}, followed by an answer $\mathcal{N}(G)$, meaning \emph{``It is $G$!''}.


\subsubsection{The Category of Well-founded Predicative Games}
This section generalizes the existing constructions on games  (in Section~\ref{McGamesAndStrategies}) so that they preserve \emph{predicativity} of games, based on which we define the category $\mathcal{WPG}$ of wf-predicative games.

However, there is a technical problem for linear implication:
The interpretation of a $\Pi$-type $\mathsf{\Pi_{a : A}B(a)}$ must be a generalization of the implication $A \to B = \ !A \multimap B$ of games, where $B$ may depend on a strategy on $A$ which Opponent chooses. 
Naively, it seems that we may interpret it by the subgame of $A \to \int \! \{ B(\sigma) \ \! | \ \! \sigma : A \ \! \}$ whose strategies $f$ satisfy $f \circ \sigma_0^\dagger : B(\sigma_0)$ for all $\sigma_0 : A$. Then the ``initial protocol'' bocomes $q_B . q_A . \mathcal{N}(\sigma_0) . \mathcal{N}(f \circ \sigma_0^\dagger)$, and a play in the subgame $\sigma_0 \to f \circ \sigma_0^\dagger \trianglelefteq \sigma_0 \to B(\sigma_0)$ follows. 
This nicely captures the phenomenon of $\Pi$-types, but imposes another challenge: The implication $A \to \int \! \{ B(\sigma) \ \! | \ \! \sigma : A \}$ no longer has a protocol since the second move $q_A$ is not the name of the strategy to follow. 

Our solution, which is one of the main achievements of the paper, is the following:
\begin{definition}[Products of PLIs \cite{yamada2016game}]
A \emph{\bfseries product of point-wise linear implications (PLIs)} between predicative games $A, B$ is a strategy of the form $\phi = \&_{\sigma : A} \phi_\sigma$, where $(\phi_\sigma)_{\sigma : A}$ is a family of strategies $\phi_\sigma : \sigma \multimap \pi_\phi (\sigma)$ with $\pi_\phi \in \mathsf{st}(B)^{\mathsf{st}(A)}$ that is ``\emph{uniform}'': $\bm{s}mn \in P_{\phi_{\sigma}} \Leftrightarrow \bm{s}mn \in P_{\phi_{\sigma'}}$ for all $\sigma, \sigma' : A, \bm{s}m \in P_{\phi_\sigma}^{\mathsf{odd}} \cap P_{\phi_{\sigma'}}^{\mathsf{odd}}, smn \in P_{\phi_\sigma} \cup P_{\phi_{\sigma'}}$, which is defined by:
\begin{itemize}

\item $M_{\&_{\sigma : A} \phi_\sigma} \stackrel{\mathrm{df. }}{=} \{ [m]_r^{\sigma} \ \! | \ \! \sigma : A, [m]_r \in M_{\phi_\sigma} \cap M_{\sigma} \} \cup \{ [m]_{r}^{\pi_\phi(\sigma)} \ \! | \ \! \sigma : A, [m]_{r} \in M_{\phi_\sigma} \cap M_{\pi_\phi(\sigma)} \}$ 

\item $\lambda_{\&_{\sigma : A} \phi_\sigma} : [m]_r^{\sigma} \mapsto \overline{\lambda_{\sigma}}([m]_r),  [m]_{r}^{\pi_\phi(\sigma)} \mapsto \lambda_{\pi_\phi(\sigma)}([m]_{r})$

\item $\vdash_{\&_{\sigma : A} \phi_\sigma} \stackrel{\mathrm{df. }}{=} \{ (\star, [m]_{r}^{\pi_\phi(\sigma)}) \ \! | \ \! \sigma : A, \star \vdash_{\pi_\phi(\sigma)} \! [m]_{r} \} \cup \{ ([m]_{r}^{\pi_\phi(\sigma)}, [n]_{l}^{\pi_\phi(\sigma)}) \ \! | \ \! \sigma : A, [m]_r \vdash_{\pi_\phi(\sigma)} [n]_{l} \} \cup \{ ([m]_r^{\sigma}, [n]_{l}^{\sigma}) \ \! | \ \! \sigma : A, [m]_r \vdash_{\sigma} [n]_{l} \} \cup \{ ([m]_{r}^{\pi_\phi(\sigma)}, [n]_{l}^{\sigma}) \ \! | \ \! \sigma : A, \star \vdash_{\pi_\phi(\sigma)} [m]_{r}, \star \vdash_{\sigma} [n]_{l} \} $ 

\item $P_{\&_{\sigma : A} \phi_\sigma} \stackrel{\mathrm{df. }}{=} \bigcup_{\sigma : A} \{ [m_1]_{r_1}^{\phi_\sigma^{(1)}} [m_2]_{r_2}^{\phi_\sigma^{(2)}} \dots [m_k]_{r_k}^{\phi_\sigma^{(k)}} \ \! | \ \! [m_1]_{r_1} [m_2]_{r_2} \dots [m_k]_{r_k} \in P_{\phi_\sigma} \}$, where $\phi_\sigma^{(i)} \stackrel{\mathrm{df. }}{=} \sigma$ if $[m_i]_{r_i} \in M_\sigma$, and $\phi_\sigma^{(i)} \stackrel{\mathrm{df. }}{=} \pi_\phi(\sigma)$ otherwise, for $i = 1, 2, \dots, k$.

\end{itemize}

\end{definition}

\begin{notation}
The set of all products of PLIs from $A$ to $B$ is written $\mathsf{PLI}(A, B)$.
\end{notation}

Clearly, products $\phi \in \mathsf{PLI}(A, B)$ of PLIs are well-defined strategies.
They are strategies on the linear implication $A \multimap B$ defined in Definition~\ref{DefConstructionsOnPredicativeGames} below. The basic idea is as follows. When Opponent makes an initial move in $A \multimap B$, he needs to determine a strategy $\sigma$ on $A$, which together with Player's declared strategy $\phi_0 : A \multimap B$ in turn determines her strategy $\pi_{\phi_0}(\sigma) $ on $B$. 
Note that $\phi_0$ has to be \emph{uniform} because she should not be able to see Opponent's choice $\sigma$.
In fact, by the uniformity, $\phi_0$ is a natural generalization of strategies on linear implication in the conventional game semantics: If $A, B$ are complete, then there is a bijection $\phi \in \mathsf{PLI}(A, B) \stackrel{\sim}{\mapsto} \bigcup \{ \phi_\sigma \ \! | \ \! \sigma : A \ \! \} : \bigcup \mathsf{st}(A) \multimap \bigcup \mathsf{st}(B)$, where note that $\bigcup \mathsf{st}(A) \multimap \bigcup \mathsf{st}(B)$ is the MC-game corresponding to $A \multimap B$ in Definition~\ref{DefConstructionsOnPredicativeGames} below (see \cite{yamada2016game} for the details).

In this manner, the new linear implication overcomes the problem mentioned above:
\begin{definition}[Constructions on predicative games \cite{yamada2016game}]
\label{DefConstructionsOnPredicativeGames}
Given a family $(G_i)_{i \in I}$ of predicative games, where $I$ is $\{ 1 \}$ or $\{ 1, 2 \}$, we define $G_1 \multimap G_2 \stackrel{\mathrm{df. }}{=} \oint \mathsf{PLI}(G_1, G_2)$ and $\clubsuit_{i \in I}G_i \stackrel{\mathrm{df. }}{=} \oint \{ \clubsuit_{i \in I} \sigma_i \ \! | \! \ \forall i \in I . \ \! \sigma_i : G_i \}$ if $\clubsuit$ is product $\&$, tensor $\otimes$, exponential $!$ or composition $\circ$.
\end{definition}


\begin{theorem}[Well-defined constructions \cite{yamada2016game}]
\label{WellDefinedConstructionsOnPredicativeGames}
Predicative games are closed under all the constructions in Definition~\ref{DefConstructionsOnPredicativeGames} except that tensor and exponential do not preserve well-openness.
\end{theorem}
\begin{proof}
Since constructions on predicative games are defined in terms of the corresponding ones on strategies, and predicative unions are well-defined, the theorem immediately follows, where uniformity of strategies on linear implication is clearly preserved under composition.
\end{proof}

\begin{definition}[Copy-casts and derelictions \cite{yamada2016game}]
The \emph{\bfseries copy-cat} $\mathit{cp}_G : G \multimap G$ (resp. \emph{\bfseries dereliction} $\mathit{der}_G : \ !G \multimap G$) on a predicative game $G$ is the product $\&_{\sigma : G} \mathit{cp}_\sigma$ (resp. $\&_{\tau : !G} \mathit{der}_{\tau^\star}$) of PLIs, where $\mathit{der}_{\tau^\star} : \tau \multimap \tau^\star$ and $\tau^\star : G$ is obtained from $\tau : \ !G$ by deleting positions with more than one initial move.
\end{definition}

It is not hard to see that if predicative games $G_i$ correspond to MC-games, i.e., the sets $\mathsf{st}(G_i)$ are complete, then the constructions defined above correspond to the usual constructions on MC-games \cite{abramsky1999game,mccusker1998games} given in Seciton~\ref{McGamesAndStrategies} (see \cite{yamada2016game} for the proof). 

\begin{example}
Consider the strategies $\underline{10} : N$, $\mathit{succ}, \mathit{double} : N \multimap N$ whose plays are:
\begin{center}
\begin{tabular}{ccccccccccc}
$N$ && $N$ & $\stackrel{\mathit{succ}}{\multimap}$ & $N$ && $N$ & $\stackrel{\mathit{double}}{\multimap}$ & $N$ \\ \cline{1-1} \cline{3-5} \cline{7-9} 
$q_N$ && &$q_{N \multimap N}$& && &$q_{N \multimap N}$& \\
$\mathcal{N}(\underline{10})$ && &$\mathcal{N}(\mathit{succ})$& && &$\mathcal{N}(\mathit{double})$& \\
$[q]_0^{\underline{10}}$ && &&$[q]_0^{\underline{n+1}, \mathit{succ}}$& && &$[q]_0^{\underline{2m}, \mathit{double}}$ \\
$[0]_0^{\underline{10}}$ && $[q]_0^{\underline{n}, \mathit{succ}}$&& && $[q]_0^{\underline{m}, \mathit{double}}$&& \\
&& $[n]_0^{\underline{n}, \mathit{succ}}$&& && $[m]_0^{\underline{m}, \mathit{double}}$&& \\
&&&&$[n+1]_0^{\underline{n+1}, \mathit{succ}}$&&&&$[2m]_0^{\underline{2m}, \mathit{double}}$
\end{tabular}
\end{center}
The tensor product $\underline{0} \otimes \underline{1} : N \otimes N$ and the composition $\mathit{succ} ; \mathit{double} : N \multimap N$ play as follows:
\begin{center}
\begin{tabular}{ccccccc}
$N$ & $\otimes$ & $N$ && $N$ & $\stackrel{\mathit{succ} ; \mathit{double}}{\multimap}$ & $N$ \\ \cline{1-3} \cline{5-7}
&$q_{N \otimes N}$& && &$q_{N \multimap N}$& \\
&$\mathcal{N}(\underline{0} \otimes \underline{1})$& && &$\mathcal{N}(\mathit{succ} ; \mathit{double})$& \\
$[q]_0^{\underline{0} \otimes \underline{1}}$&&&&&&$[q]_0^{\underline{2(n+1)}, \mathit{succ}; \mathit{double}}$ \\
$[0]_0^{\underline{0} \otimes \underline{1}}$&&&&$[q]_0^{\underline{n}, \mathit{succ}; \mathit{double}}$&& \\
&&$[q]_0^{\underline{0} \otimes \underline{1}}$&&$[n]_0^{\underline{n}, \mathit{succ} ; \mathit{double}}$&& \\
&&$[1]_0^{\underline{0} \otimes \underline{1}}$&&&&$[2(n+1)]_0^{\underline{2(n+1)}, \mathit{succ}; \mathit{double}}$
\end{tabular}
\end{center}
\end{example}

\begin{definition}[The category $\mathcal{WPG}$ \cite{yamada2016game}]
\label{DefCategoryWPG}
The category $\mathcal{WPG}$ is defined as follows:
\begin{itemize}

\item Objects are wf-predicative games

\item Morphisms $A \to B$ are innocent, wb, total and noetherian strategies on $A \to B \stackrel{\mathrm{df. }}{=} \ \! ! A \multimap B$

\item The composition of morphisms $\phi : A \to B$, $\psi : B \to C$ is $\psi \bullet \phi \stackrel{\mathrm{df. }}{=} \psi \circ \phi^\dagger : A \to C$

\item The identity $\mathit{id}_A$ on each object $A$ is the dereliction $\mathit{der}_A : A \to A$.

\end{itemize}
\end{definition}

\begin{remark}
Anti-strategies do not have to be innocent, wb, total or noetherian (we have not formulated these notions, but it should be clear what it means). Thus, in particular, for any morphism $\phi = \&_{\sigma : !A}\phi_\sigma : A \to B$ in $\mathcal{WPG}$, $\sigma$ ranges over \emph{any} strategies on $!A$.
\end{remark}

\begin{corollary}[Well-defined $\mathcal{WPG}$ \cite{yamada2016game}]
The structure $\mathcal{WPG}$ forms a well-defined category.
\end{corollary}
\begin{proof}
By Lemmata~\ref{LemWellDefinedDerelictions}, \ref{LemWellDefinedComposition}, \ref{LemWellDefinedPairingTensorPromotion} and Theorem~\ref{WellDefinedConstructionsOnPredicativeGames}.
\end{proof}

\subsection{Games with Equalities}
\label{GamesWithEqualities}
We have reviewed all the preliminary concepts, and the main content of the paper starts from the present section.
From now on, let \emph{\bfseries games} and \emph{\bfseries strategies} refer to predicative games and strategies on them by default.

\begin{notation}
We write $A \Rightarrow B$ for $!A \multimap B$, as well as $A \stackrel{\sim}{\multimap} B$ and $A \stackrel{\sim}{\Rightarrow} B$ for their respective subgames whose strategies are \emph{invertible}, and $\psi \bullet \phi$ for the composition $\psi \circ \phi^\dagger : A \Rightarrow C$ of strategies $\phi : A \Rightarrow B$, $\psi : B \Rightarrow C$.
We often present a strategy $\sigma$ by just specifying the set $P_\sigma$ of its positions whenever the other components are unambiguous, and write $\bm{s} \in \sigma$ for $\bm{s} \in P_\sigma$.
\end{notation}

Let us begin with a key observation (which is applied not only to predicative games but also to any conventional games):
\begin{theorem}[Isom theorem]
\label{ThmIsomThm}
There is an invertible strategy $\phi : A \stackrel{\sim}{\multimap} B$ or $\psi : A \stackrel{\sim}{\Rightarrow} B$ (with respect to copy-cats or derelictions) if and only if there is a bijection $f : P_A \stackrel{\sim}{\to} P_B$ such that $f(\bm{\epsilon}) = \bm{\epsilon}$ and $f(\bm{s}m) = \bm{t}n \Rightarrow f(\bm{s}) = \bm{t}$ for all $\bm{s}m \in P_A$, $\bm{t}n \in P_B$.
\end{theorem}
\begin{proof}
Assume $\phi : A \stackrel{\sim}{\multimap} B$; the case $\psi : A \stackrel{\sim}{\Rightarrow} B$ is analogous, and so we omit it.
It is easy to see that $\phi$ and the inverse $\phi^{-1} : B \multimap A$ both ``behave like copy-cats'' in the sense that $\bm{s} b \in \phi^{\mathsf{odd}}$ (resp. $\bm{s} b \in (\phi^{-1})^{\mathsf{odd}}$) with $b \in M_B$ implies $\bm{s} b a \in \phi$ (resp. $\bm{s} b a \in \phi^{-1}$) for some $a \in M_A$, and $\bm{t} a' \in \phi^{\mathsf{odd}}$ (resp. $\bm{t} a' \in (\phi^{-1})^{\mathsf{odd}}$) with $a' \in M_A$ implies $\bm{t} a' b' \in \phi$ (resp. $\bm{t} a' b' \in \phi^{-1}$) for some $b' \in M_B$ (since otherwise $\phi \circ \phi^{-1}$ or $\phi^{-1} \circ \phi$ would not be a copy-cat).
Hence, we may define the function $\mathscr{F}(\phi) : P_A \to P_B$ that maps:
\begin{enumerate}

\item $\bm{\epsilon} \mapsto \bm{\epsilon}$

\item If $a_1 a_2 \dots a_{2n} \mapsto b_1 b_2 \dots b_{2n}$ and $a_1 b_1 b_2 a_2 \dots a_{2n-1} b_{2n-1} b_{2n} a_{2n} a_{2n+1} b_{2n+1} \in \phi^{-1}$, then
\begin{equation*}
a_1 a_2 \dots a_{2n+1} \mapsto b_1 b_2 \dots b_{2n+1}
\end{equation*} 

\item If $a_1 a_2 \dots a_{2n+1} \mapsto b_1 b_2 \dots b_{2n+1}$ and $b_1 a_1 a_2 b_2 \dots b_{2n+1} a_{2n+1} a_{2n+2} b_{2n+2} \in \phi$, then
\begin{equation*}
a_1 a_2 \dots a_{2n+2} \mapsto b_1 b_2 \dots b_{2n+2}.
\end{equation*}

\end{enumerate}
 
Analogously and symmetrically, we may define another function $\mathscr{G}(\phi) : P_B \to P_A$.
By induction on the length of input, it is easy to see that $\mathscr{F}(\phi)$ and $\mathscr{G}(\phi)$ are mutually inverses, and they both satisfy the required two conditions.

Conversely, if there is a bijection $f : P_A \stackrel{\sim}{\to} P_B$ satisfying the two conditions, then by ``reversing'' the above procedure, we may construct an invertible strategy $\mathscr{S}(f) : A \stackrel{\sim}{\multimap} B$ and its inverse $\mathscr{S}(f)^{-1} : B \stackrel{\sim}{\multimap} A$ from $f$ and $f^{-1}$, completing the proof.
\end{proof} 

Note that what essentially identifies a given game $G$ is the set $P_G$ of its positions (since games are assumed to be \emph{economical} and j-sequences contain information for labeling); however, it is not an essential point what each of these positions really is as long as it is distinguished from other positions in $G$.
Hence, Theorem~\ref{ThmIsomThm} can be read as: 
\begin{quote}
Isomorphic games are essentially the same ``up to implementation of positions''.
\end{quote}
In other words, the category-theoretic point of view that identifies isomorphic objects makes sense in game semantics as well, where note that a category of games and strategies usually consists of games as objects and strategies between them as morphisms.

Now, let us consider strategies on a fixed game $G$.
Which strategies on $G$ should be considered to be essentially the same or \emph{equivalent}?
Contrary to the case of games, ``implementation of positions'' in strategies \emph{matters} as the underlying game $G$ is already given.
For instance, if we identify any isomorphic strategies on the natural number game $N$, then there would be just one total strategy on $N$, which clearly should not be the case if we want $N$ to represent the set $\mathbb{N}$ of natural numbers. 
On the other hand, e.g., we may \emph{choose} to identify strategies $\underline{n_1}, \underline{n_2} : N$ exactly when $n_1 \equiv n_2 \mod 2$, so that the resulting game represents the set of natural numbers modulo $2$.
Note that it is reasonable to require equivalent strategies to be isomorphic as a minimal requirement since it guarantees as in the case of games that they are graph-theoretically isomorphic (i.e., \emph{isomorphic rooted forests}).

Thus, it seems that we may \emph{define} an equivalence between strategies $\sigma_1, \sigma_2 : G$ by equipping it with a set of \emph{selected} invertible strategies\footnote{These strategies are not necessarily between $\sigma_1, \sigma_2$ themselves as explained below. Also, they are \emph{invertible} not necessarily with respect to the composition of strategies but the composition of the underlying category.}; the set must contain the identity strategies and be closed under composition and inverses as it represents an equivalence relation.
This simple idea leads to the following central notion of the present paper:

\begin{definition}[GwEs]
A \emph{\bfseries game with equality (GwE)} is a groupoid whose objects are strategies on a fixed game and morphisms are invertible strategies.
\end{definition}

\begin{notation}
We usually specify a GwE by a pair $G = (G, =_G)$ of an underlying game $G$ and an assignment $=_G$ of a game $\sigma_1 =_G \sigma_2 \stackrel{\mathrm{df. }}{=} \oint G(\sigma_1, \sigma_2)$ to each pair $\sigma_1, \sigma_2 : G$, i.e., $\mathsf{ob}(G) = \mathsf{st}(G)$, $G(\sigma_1, \sigma_2) = \mathsf{st}(\sigma_1 =_G \sigma_2)$.
We often write $\rho : \sigma_1 =_G \sigma_2$ rather than $\rho \in G(\sigma_1, \sigma_2)$, and call it a \emph{proof of the equality} between $\sigma_1$ and $\sigma_2$. 
Moreover, let the assignment $=_G$ turn into the game by $=_G \ \stackrel{\mathrm{df. }}{=} \oint \bigcup_{\sigma_1, \sigma_2 : G} G(\sigma_1, \sigma_2) = \int \{ \sigma_1 =_G \sigma_2 \ \! | \ \! \sigma_1, \sigma_2 : G \ \! \}$.
Note, however, that the set $\mathsf{st}(=_G)$ is in general \emph{not} the set of all morphisms in the GwE $G$ since hom-sets of $G$ may not be pairwise disjoint, and so the domain or codomain of some morphism may not be recovered from $\mathsf{st}(=_G)$\footnote{In other words, we must identify $=_G$ as an assignment of $\sigma_1 =_G \sigma_2$ or $G(\sigma_1, \sigma_2)$ to each $\sigma_1, \sigma_2 : G$; the game $=_G$ or the set $\mathsf{st}(=_G)$ may lose the information for domain or codomain of some morphisms.}.
\end{notation}

When we say \emph{games} $A, B$ or \emph{strategies} $\sigma : A$, $\phi : A \Rightarrow B$, etc., where $A, B$ are GwEs, we refer to the underlying games $A, B$. A GwE $A$ is defined to be \emph{\bfseries wf} if so are the games $A, =_A$.

\begin{remark}
One may wonder if the relation $\sigma_1 =_G \sigma_2 \trianglelefteq \oint \{ \sigma_1 \} \stackrel{\sim}{\Rightarrow} \oint \{ \sigma_2 \}$ or at least $\sigma_1 =_G \sigma_2 \trianglelefteq G \stackrel{\sim}{\Rightarrow} G$ should hold; however, neither is general enough for Definition~\ref{DefDependentPairSpace}.
Conceptually, this is because a proof $\rho : \sigma_1 =_G \sigma_2$ may ``look at'' relevant information for $\sigma_1, \sigma_2$ but not necessarily $\sigma_1, \sigma_2$ themselves.
Nevertheless, in most games $G$, strategies $\rho : \sigma_1 =_G \sigma_2$ satisfy $\rho : \sigma_1 \stackrel{\sim}{\Rightarrow} \sigma_2$.
\end{remark}

\begin{definition}[Ep-strategies]
A strategy $\phi : A \Rightarrow B$, where $A, B$ are GwEs, is \emph{\bfseries equality-preserving (ep)} if it is equipped with another strategy $\phi^= : \ =_A \ \Rightarrow \ =_B$, called its \emph{\bfseries equality-preservation}, such that the maps $(\sigma : A) \mapsto \phi \bullet \sigma$, $(\rho : \ =_A) \mapsto \phi^= \bullet \rho$ respectively form the object- and arrow-maps of the \emph{\bfseries extensional functor} $\mathsf{fun}(\phi) : A \to B$ induced by $\phi$.
\end{definition}

Explicitly, an ep-strategy $\phi : A \Rightarrow B$ is a pair $\phi = (\phi, \phi^=)$ of strategies 
\begin{align*}
\phi &: A \Rightarrow B \\
\phi^= &: \ =_A \ \Rightarrow \ =_B
\end{align*}
that satisfies, for all $\sigma_1, \sigma_2, \sigma_3 : A$, $\rho_1 : \sigma_1 =_A \sigma_2$, $\rho_2 : \sigma_2 =_A \sigma_3$, the following three conditions: 
\begin{enumerate}

\item $\phi^= \bullet \rho_1 : \phi \bullet \sigma_1 =_B \phi \bullet \sigma_2$

\item $\phi^= \bullet (\rho_2 \bullet \rho_1) = (\phi^= \bullet \rho_2) \bullet (\phi^= \bullet \rho_1)$

\item $\phi^= \bullet \mathit{id}_{\sigma_1} = \mathit{id}_{\phi \bullet \sigma_1}$.

\end{enumerate}
We define $\phi = (\phi, \phi^=)$ to be \emph{\bfseries innocent} (resp. \emph{\bfseries total}, \emph{\bfseries wb}, \emph{\bfseries noetherian}) if so are both $\phi$ and $\phi^=$.

\begin{remark}
For any ep-strategy $\phi : A \Rightarrow B$, $\phi^=$ is not a family $(\phi^=_{\sigma_1, \sigma_2})_{\sigma_1, \sigma_2 : A}$ of strategies $\phi^=_{\sigma_1, \sigma_2} : \sigma_1 =_A \sigma_2 \Rightarrow \phi \bullet \sigma_2 =_B \phi \bullet \sigma_2$ but a single strategy $\phi^= : \ =_A \ \Rightarrow \ =_B$.
This is for $\phi$ to be accordance with the \emph{intensional} and \emph{computational} nature of game semantics in the sense that $\phi^=$ cannot \emph{extensionally} access to the information for the domain and codomain of a given input strategies on $=_A$.
Notice that this formulation is not possible for conventional games and strategies; it is possible due to our formulation of the relation ``$\sigma : G$'', i.e., a strategy $\sigma$ is defined \emph{independently} of games, and we may determine if the relation $\sigma : G$ holds for a given game $G$.
\end{remark}

Conceptually, a GwE $G$ is a game $G$ equipped with the set $\sigma_1 =_G \sigma_2$ of ``(computational) proofs of the (intensional) equality'' between $\sigma_1$ and $\sigma_2$ for all $\sigma_1, \sigma_2 : G$, and an ep-strategy $\phi : A \Rightarrow B$ is a strategy equipped with another one $\phi^=$ that computes on proofs of equalities in $A$ and $B$.
Therefore one may say that GwEs and ep-strategies are groupoids and functors between them equipped with the \emph{intensional} structure of games and strategies.

As expected, GwEs and ep-strategies form a category:
\begin{definition}[The category $\mathcal{PGE}$]
\label{DefPGE}
The category $\mathcal{PGE}$ is defined by: 
\begin{itemize}

\item Objects are wf-GwEs

\item Morphisms $A \to B$ are total, innocent, wb and noetherian ep-strategies $\phi : A \Rightarrow B$

\item The composition $\psi \bullet \phi : A \to C$ of morphisms $\phi : A \to B$, $\psi : B \to C$ is given by the compositions of strategies $\psi \bullet \phi$, $(\psi \bullet \phi)^= \stackrel{\mathrm{df. }}{=} \psi^= \bullet \phi^=$

\item The identity $\mathit{id}_A$ is the dereliction $\mathit{der}_A$ equipped with $(\mathit{der}_A)^= \stackrel{\mathrm{df. }}{=} \mathit{der}_{=_A}$.

\end{itemize}
\end{definition}

The category $\mathcal{PGE}$ is basically the category $\mathcal{WPG}$ of wf-games (Definition~\ref{DefCategoryWPG}) equipped with an ``intensional groupoid structure'' in the sense that $\mathcal{PGE}$ forms a subcategory of the category of groupoids and functors \cite{hofmann1998groupoid}.
Accordingly, it is straightforward to establish:

\begin{theorem}[Well-defined $\mathcal{PGE}$]
The structure $\mathcal{PGE}$ forms a well-defined category.
\end{theorem}
\begin{proof}
We first show that the composition is well-defined.
Let $\phi : A \to B$, $\psi : B \to C$ be any morphisms in $\mathcal{PGE}$.
By Lemma~\ref{LemWellDefinedComposition}, $\psi \bullet \phi : A \Rightarrow C$ and $\psi^= \bullet \phi^= : \ =_A \ \Rightarrow \ =_C$ both form total, innocent, wb and noetherian strategies.
Moreover, for all $\sigma_1, \sigma_2, \sigma_3 : A$, $\rho_1 : \sigma_1 =_A \sigma_2$, $\rho_2 : \sigma_2 =_A \sigma_3$, they satisfy: 
\begin{enumerate}

\item $\psi^= \bullet (\phi^= \bullet \rho_1) : \psi \bullet (\phi \bullet \sigma_1) =_B \psi \bullet (\phi \bullet \sigma_2) \Leftrightarrow (\psi^= \bullet \phi^=) \bullet \rho_1 : (\psi \bullet \phi) \bullet \sigma_1 =_B (\psi \bullet \phi) \bullet \sigma_2$

\item $(\psi^= \bullet \phi^=) \bullet (\rho_2 \bullet \rho_1) = \psi^= \bullet (\phi^= \bullet (\rho_2 \bullet \rho_1)) = \psi^= \bullet ((\phi^= \bullet \rho_2) \bullet (\phi^= \bullet \rho_1)) = (\psi^= \bullet (\phi^= \bullet \rho_2)) \bullet (\psi^= \bullet (\phi^= \bullet \rho_1)) = ((\psi^= \bullet \phi^=) \bullet \rho_2) \bullet ((\psi^= \bullet \phi^=) \bullet \rho_1)$

\item $(\psi^= \bullet \phi^=) \bullet \mathit{id}_{\sigma_1} = \psi^= \bullet (\phi^= \bullet \mathit{id}_{\sigma_1}) = \psi^= \bullet \mathit{id}_{\phi \bullet \sigma_1} = \mathit{id}_{\psi \bullet (\phi \bullet \sigma_1)} = \mathit{id}_{(\psi \bullet \phi) \bullet \sigma_1}$.

\end{enumerate}
Therefore the composition $\psi \bullet \phi = (\psi \bullet \phi, \psi^= \bullet \phi^=)$ is in fact a morphism $A \to C$ in $\mathcal{PGE}$.
Note that the associativity of composition in $\mathcal{PGE}$ immediately follows from the associativity of composition of strategies.

Next, by Lemma~\ref{LemWellDefinedDerelictions}, $\mathit{der}_A$ (resp. $\mathit{der}_{=_A}$) is a total, innocent, wb and noetherian strategy on $A \stackrel{\sim}{\Rightarrow} A$ (resp. $=_A \ \stackrel{\sim}{\Rightarrow} \ =_A$) for each wf-GwE $A$.
It is also easy to see that the pair $\mathit{der}_A = (\mathit{der}_A, \mathit{der}_{=_A})$ satisfies the required functoriality, forming a morphism $A \to A$ in $\mathcal{PGE}$.
Finally, for any morphism $\phi : A \to B$ in $\mathcal{PGE}$, we clearly have $\mathit{der}_B \bullet \phi = \phi$, $\mathit{der}_B^= \bullet \phi^= = \phi^=$, $\phi \bullet \mathit{der}_A = \phi$ and $\phi^= \bullet \mathit{der}_A^= = \phi^=$.
Hence, these pairs $\mathit{der}_A = (\mathit{der}_A, \mathit{der}_{=_A})$ of derelictions satisfy the unit law, completing the proof.
\end{proof}

As explained in \cite{yamada2016game}, wf-games can be seen as ``propositions'' and total, innocent, wb and noetherian strategies on them as ``(constructive) proofs''.
Thus, we say that an object $A \in \mathcal{PGE}$ is \emph{\bfseries true} if there is some $\sigma \in \mathcal{PGE}(I, A)$, called a \emph{\bfseries proof} of $A$, and it is \emph{\bfseries false} otherwise.

\subsection{Dependent Games with Equalities}
\label{DependentGamesWithEqualities}
This section gives constructions on GwEs to interpret $\Pi$-, $\Sigma$- and \textsf{Id}-types ``\emph{partially}'' in the sense that they are applied only to \emph{closed} terms.
The ``full interpretation'' of these types will be given in Section~\ref{GameTheoreticTypeFormers}.

We begin with our game semantics for \emph{dependent types}:
\begin{definition}[DGwEs]
A \emph{\bfseries dependent game with equality (DGwE)} over $A \in \mathcal{PGE}$ is a functor $B : A \to \mathcal{PGE}$ that is ``\emph{uniform}'': 
\begin{align*}
\bm{s} ab \in B(\rho)_\tau &\Leftrightarrow \bm{s} ab \in B(\tilde{\rho})_{\tilde{\tau}} \\
\bm{t} mn \in B(\rho)^=_\varrho &\Leftrightarrow \bm{t} mn \in B(\tilde{\rho})^=_{\tilde{\varrho}} 
\end{align*}
for all $\gamma, \gamma', \tilde{\gamma}, \tilde{\gamma}' : \Gamma$, $\rho : \gamma =_\Gamma \gamma'$, $\tilde{\rho} : \tilde{\gamma} =_\Gamma \tilde{\gamma}'$, $\tau, \tau_1, \tau_2 : B(\gamma)$, $\tilde{\tau}, \tilde{\tau}_1, \tilde{\tau}_2 : B(\tilde{\gamma})$, $\varrho : \tau_1 =_{B(\gamma)} \tau_2$, $\tilde{\varrho} : \tilde{\tau}_1 =_{B(\tilde{\gamma})} \tilde{\tau}_2$, $\bm{s} a \in B(\rho)_\tau^{\mathsf{odd}} \cap B(\tilde{\rho})_{\tilde{\tau}}^{\mathsf{odd}}$, $\bm{s} ab \in B(\rho)_\tau \cup B(\tilde{\rho})_{\tilde{\tau}}$, $\bm{t} m \in (B(\rho)^=_{\varrho})^{\mathsf{odd}} \cap (B(\tilde{\rho})^=_{\tilde{\varrho}})^{\mathsf{odd}}$, $\bm{t} mn \in B(\rho)^=_\varrho \cup B(\tilde{\rho})^=_{\tilde{\varrho}}$.
\end{definition}

\begin{notation}
We write $\mathscr{D}(A)$ for the set of all DGwEs over $A \in \mathcal{PGE}$.
For each $B \in \mathscr{D}(A)$, recall that \cite{yamada2016game} defined the \emph{\bfseries dependent union} $\uplus B$ by $\uplus B  \stackrel{\mathrm{df. }}{=} \int \{ B(\sigma) \ \! | \ \! \sigma : A \ \! \}$.
\end{notation}

The uniformity of DGwEs $B : A \to \mathcal{PGE}$ ensures that the respective strategies $B(\rho)$, $B(\rho)^=$, where $\rho$ ranges over morphisms in $\Gamma$, behave in the ``uniform manner''.
We need it in the proof of Theorems~\ref{ThmGameTheoreticPiTypes}, \ref{ThmGameTheoreticSumTypes}.
Also, DGwEs are a generalization of wf-GwEs since a wf-GwE can be equivalently presented as a DGwE $B : I \to \mathcal{PGE}$, where $I$ is the \emph{discrete} (i.e., morphisms are only identities) wf-GwE on the terminal game $I = (\emptyset, \emptyset, \emptyset, \{ \bm{\epsilon}, q_I, q_I \mathcal{N}(\{ \bm{\epsilon} \}) \})$.

We are now ready to give a partial interpretation of $\Pi$-types:
\begin{definition}[Dependent function space]
\label{DefDependentFunctionSpace}
Given a DGwE $B : A \to \mathcal{PGE}$, the \emph{\bfseries dependent function space} $\widehat{\Pi}(A, B)$ from $A$ to $B$ is defined as follows:
\begin{itemize}

\item The game $\widehat{\Pi}(A, B)$ is the subgame of $A \Rightarrow \uplus B$ whose strategies $\phi$ are equipped with an equality-preservation $\phi^= : \ =_A \ \Rightarrow \uplus =_B$, where $\uplus =_B \ \stackrel{\mathrm{df. }}{=} \int  \{ \ \! =_{B(\sigma)} \! | \ \! \sigma : A \ \! \}$, that satisfy:
\begin{align*}
\phi \bullet \sigma &: B(\sigma) \\
\phi^= \bullet \rho &: B(\rho) \bullet \phi \bullet \sigma =_{B(\sigma')} \phi \bullet \sigma' \\
\phi^= \bullet (\rho' \bullet \rho) &= (\phi^= \bullet \rho') \bullet (B(\rho')^= \bullet \phi^= \bullet \rho) \\
\phi^= \bullet \mathit{id}_\sigma &= \mathit{id}_{\phi \bullet \sigma}
\end{align*}
for all $\sigma, \sigma', \sigma'' : A$, $\rho : \sigma =_A \sigma'$, $\rho' : \sigma' =_A \sigma''$

\item For any $\phi_1, \phi_2 : \widehat{\Pi}(A, B)$, the game $\phi_1 =_{\widehat{\Pi}(A, B)} \phi_2$ consists of strategies $\mu : \phi_1 \stackrel{\sim}{\Rightarrow} \phi_2$, where we write $\phi_i : \widehat{\Pi}(A^{[i]}, B^{[i]})$ for $i = 1, 2$ to distinguish different copies of $A$, $\uplus B$, that satisfy:
\begin{enumerate}

\item $\mathsf{even}(\bm{s})$ implies $\mathsf{even}(\bm{s} \upharpoonright A^{[1]}, A^{[2]})$ for all $\bm{s} \in \mu$

\item $\mu \upharpoonright A^{[2]} \in \sigma$ implies $\mu \upharpoonright A^{[1]} \in \sigma$ for all $\sigma : A$

\item $\mu_\sigma \stackrel{\mathrm{df. }}{=} \{ \bm{s} \upharpoonright \uplus B^{[1]}, \uplus B^{[2]} \ | \ \bm{s} \in \mu, \mu \upharpoonright A^{[2]} \in \sigma \} : \phi_1 \bullet \sigma =_{B(\sigma)} \phi_2 \bullet \sigma$ for all $\sigma : A$

\item $\mathsf{nat}(\mu) \stackrel{\mathrm{df. }}{=} \{ \mu_\sigma \ \! | \! \ \sigma : A \ \! \}$ forms a natural transformation from $\mathsf{fun}(\phi_1)$ to $\mathsf{fun}(\phi_2)$

\end{enumerate}

\item The composition, identities and inverses of morphisms are the ones for strategies.

\end{itemize}
\end{definition}

In the game $\phi_1 =_{\widehat{\Pi}(A, B)} \phi_2$, Player (resp. Opponent) can control only P-moves in $\phi_1$ (resp. O-moves in $\phi_2$); thus, in $\mu : \phi_1 =_{\widehat{\Pi}(A, B)} \phi_2$, a play is completely determined by O-moves in $\phi_2$.  

The intuition behind the four axioms for $\mu : \phi_1 =_{\widehat{\Pi}(A, B)} \phi_2$ is as follows.
The condition 1 ensures that $\mu$ ``witnesses'' that $\phi_1$ and $\phi_2$ go back and forth between $A$ and $\uplus B$ ``in the same timing''.
The conditions 2, 3 guarantee that the extensional ``input/output behaviors'' of $\phi_1$ and $\phi_2$ are shown to be equal by $\mu$. 
Finally, the condition 4, just as naturality in general, corresponds to the ``uniformity'' of $\mu_\sigma$, where $\sigma$ ranges over strategies on $A$.

\begin{lemma}[Well-defined $\widehat{\Pi}$]
\label{LemWellDefinedDependentFunctionSpace}
For any DGwE $B : A \to \mathcal{PGE}$, we have $\widehat{\Pi}(A, B) \in \mathcal{PGE}$.
\end{lemma}
\begin{proof}
First, it is easy to see that the games $\widehat{\Pi}(A, B), =_{\widehat{\Pi}(A, B)}$ are wf thanks to ``initial protocols''.
For the composition $\bullet$, let $\mu : \phi_1 =_{\widehat{\Pi}(A, B)} \phi_2$, $\nu : \phi_2 =_{\widehat{\Pi}(A, B)} \phi_3$.
It is easy to see that the axioms 1, 2 are satisfied by the composition $\nu \bullet \mu$.
Also, it is not hard to see that these two axioms imply $(\nu \bullet \mu)_\sigma = \nu_\sigma \bullet \mu_\sigma : \phi_1 \bullet \sigma =_{B(\sigma)} \phi_3 \bullet \sigma$ for all $\sigma : A$ (by induction on four consecutive positions with a case analysis), and so $\mathsf{nat}(\nu \bullet \mu) = \{ \mu_\sigma \bullet q_\sigma \ \! | \! \ \sigma : A \ \! \}$ satisfies the naturality condition as $\mathsf{nat}(\nu \bullet \mu)$ is just the vertical composition $\mathsf{nat}(\nu) \circ \mathsf{nat}(\mu)$ of natural transformations.
Thus, $\nu \bullet \mu : \phi_1 =_{\widehat{\Pi}(A, B)} \phi_3$, and so the composition $\bullet$ is well-defined. 
Also, the identity $\mathit{id}_\phi$ on each $\phi : \widehat{\Pi}(A, B)$ clearly satisfies the four axioms, and so $\mathit{id}_\phi : \phi =_{\widehat{\Pi}(A, B)} \phi$.
Note that the associativity of the composition and the unit law of the identities are just the corresponding properties of the composition and identities of strategies.

Next, in light of Theorem~\ref{ThmIsomThm}, it is clear that the inverse $\mu^{-1} : \phi_2 \stackrel{\sim}{\Rightarrow} \phi_1$ satisfies the first two axioms, and $(\mu^{-1})_\sigma = (\mu_\sigma)^{-1} : \phi_2 \bullet \sigma =_{B(\sigma)} \phi_1 \bullet \sigma$ for all $\sigma : A$. 
Thus, $\mathsf{nat}(\mu^{-1}) = \mathsf{nat}(\mu)^{-1}$, and so $\mu^{-1}$ satisfies the naturality condition as the inverse of any natural transformation does. Explicitly, given any $\sigma, \sigma' : A$, $\rho : \sigma =_A \sigma'$, we have:
\begin{equation*}
\phi_2^= \bullet \rho \bullet \mu_\sigma = \mu_{\sigma'} \bullet \phi_1^= \bullet \rho 
\end{equation*}
by the nautrality of $\mathsf{nat}(\mu)$, whence $\mu_{\sigma'}^{-1} \bullet \phi_2^= \bullet \rho \bullet \mu_\sigma \bullet \mu_\sigma^{-1} = \mu_{\sigma'}^{-1} \bullet \mu_{\sigma'} \bullet \phi_1^= \bullet \rho \bullet \mu_\sigma^{-1}$, i.e., 
\begin{equation*}
\mu_{\sigma'}^{-1} \bullet \phi_2^= \bullet \rho = \phi_1^= \bullet \rho \bullet \mu_\sigma^{-1}
\end{equation*}
which completes the proof.
\end{proof}

The idea is best described by a set-theoretic analogy: $\widehat{\Pi} (A, B)$ represents the space of functions $f : A \to \bigcup_{x \in A} B(x)$ that satisfies $f(a) \in B(a)$ for all $a \in A$. 
Again, the conditions on morphisms are an ``intensional refinement'' of those in the groupoid interpretation \cite{hofmann1998groupoid}.
However, we have to handle the case where $A$ is a DGwE; so $\widehat{\Pi}$ is not general enough. 
In terms of the syntax, we can interpret $(\textsc{$\Pi$-Form}) \ \mathsf{\Gamma, x:A \vdash B(x) \ \mathsf{type} \Rightarrow \Gamma \vdash \textstyle \Pi_{x:A}B(x) \ \mathsf{type}}$ only when $\mathsf{\Gamma = \diamondsuit}$ (the \emph{empty context}) at the moment.
We shall define a more general $\Pi$ shortly.
This applies to the interpretation of $\Sigma$- and \textsf{Id}-types given below as well.

Note that a function maps equal inputs to equal outputs, but it is not obvious if it is the case for our computational equalities, i.e., morphisms in GwEs, since there are non-trivial ones. 
However, equality-preservations ensure the desired property: 
\begin{proposition}[Dependent functionality]
Let $A \in \mathcal{PGE}$, $B \in \mathscr{D}(A)$, $\phi_1, \phi_2 : \widehat{\Pi}(A, B)$.
If there are morphisms $\varrho : \phi_1 =_{\widehat{\Pi}(A, B)} \phi_2$, $\rho : \sigma_1 =_A \sigma_2$ in $\mathcal{PGE}$, then we have at least two morphisms 
\begin{equation*}
(\phi_2^= \bullet \rho) \bullet (B(\rho)^= \bullet \varrho_{\sigma_1}), \varrho_{\sigma_2} \bullet (\phi_1^= \bullet \rho) : B(\rho) \bullet \phi_1 \bullet \sigma_1 =_{B(\sigma_2)} \phi_2 \bullet \sigma_2
\end{equation*}
in $\mathcal{PGE}$.
\end{proposition}

On the other hand, for any $\phi_1, \phi_2 : \widehat{\Pi}(A, B)$, each morphism on $\phi_1 =_{\widehat{\Pi}(A, B)} \phi_2$ tracks the ``dynamics'' of $\phi_1$ and $\phi_2$; in particular, there is such a morphism only if $\phi_1$ and $\phi_2$ go back and forth between $A$ and $\uplus B$ ``in the same timing''.
Therefore our game semantics \emph{refutes} the axiom of \emph{\bfseries function extensionality (FunExt)} by the same argument as \cite{abramsky2015games,yamada2016game}.
This \emph{intensional} nature of our interpretation presents a sharp contrast to the groupoid model \cite{hofmann1998groupoid}.

We proceed to (partially) interpret $\Sigma$-types: 
\begin{definition}[Dependent pair space]
\label{DefDependentPairSpace}
Given a DGwE $B : A \to \mathcal{PGE}$, the \emph{\bfseries dependent pair space} $\widehat{\Sigma}(A, B)$ of $A$ and $B$ is defined as follows:
\begin{itemize}

\item The game $\widehat{\Sigma}(A, B)$ is the subgame of $A \& (\uplus B)$ whose strategies $\langle \sigma, \tau \rangle$ satisfy $\tau : B(\sigma)$

\item For any objects $\langle \sigma, \tau \rangle, \langle \sigma', \tau' \rangle : \widehat{\Sigma}(A, B)$, the game $\langle \sigma, \tau \rangle =_{\widehat{\Sigma}(A, B)} \langle \sigma', \tau' \rangle$ is given by:
\begin{equation*}
\mathsf{st}(\langle \sigma, \tau \rangle =_{\widehat{\Sigma}(A, B)} \langle \sigma', \tau' \rangle) \stackrel{\mathrm{df. }}{=} \{ \langle \rho, \varrho \rangle \ \! | \ \! \rho : \sigma =_A \sigma', \varrho : B(\rho) \bullet \tau =_{B(\sigma')} \tau' \ \! \} 
\end{equation*}

\item The composition $\langle \rho_2, \varrho_2 \rangle \bullet \langle \rho_1, \varrho_1 \rangle : \langle \sigma_1, \tau_1 \rangle =_{\widehat{\Sigma}(A, B)} \langle \sigma_3, \tau_3 \rangle$ of morphisms $\langle \rho_1, \varrho_1 \rangle : \langle \sigma_1, \tau_1 \rangle =_{\widehat{\Sigma}(A, B)} \langle \sigma_2, \tau_2 \rangle$, $\langle \rho_2, \varrho_2 \rangle : \langle \sigma_2, \tau_2 \rangle =_{\widehat{\Sigma}(A, B)} \langle \sigma_3, \tau_3 \rangle$ is given by: 
\begin{equation*}
\langle \rho_2, \varrho_2 \rangle \bullet \langle \rho_1, \varrho_1 \rangle \stackrel{\mathrm{df. }}{=} \langle \rho_2 \bullet \rho_1, \varrho_2 \odot \varrho_1 \rangle
\end{equation*}
where $\varrho_2 \odot \varrho_1$ is defined by: 
\begin{equation*}
\varrho_2 \odot \varrho_1 \stackrel{\mathrm{df. }}{=} \varrho_2 \bullet B(\rho_2)^= \bullet \varrho_1
\end{equation*} 

\item The identity $\mathit{id}_{\langle \sigma, \tau \rangle}$ on each object $\langle \sigma, \tau \rangle : \widehat{\Sigma}(A, B)$ is the pairing $\langle \mathit{id}_\sigma, \mathit{id}_\tau \rangle$

\item The inverse $\langle \rho, \varrho \rangle^{-1}$ of each morphism $\langle \rho, \varrho \rangle$ is given by:
\begin{equation*}
\langle \rho, \varrho \rangle^{-1} \stackrel{\mathrm{df. }}{=} \langle \rho^{-1}, \varrho^\circleddash \rangle
\end{equation*}
where $\varrho^{\circleddash}$ is given by:
\begin{equation*}
\varrho^{\circleddash} \stackrel{\mathrm{df. }}{=} B(\rho^{-1})^= \bullet \varrho^{-1}.
\end{equation*}

\end{itemize}
\end{definition}

\begin{remark}
A morphism between $\langle \sigma, \tau \rangle, \langle \sigma', \tau' \rangle : \widehat{\Sigma}(A, B)$ is not a strategy on $\langle \sigma, \tau \rangle \stackrel{\sim}{\Rightarrow} \langle \sigma', \tau' \rangle$ but a pairing $\langle \rho, \varrho \rangle$ of morphisms $\rho : \sigma =_A \sigma'$, $\varrho : B(\rho) \bullet \tau =_{B(\sigma')} \tau'$ because it seems impossible to organize $\rho$ and $\varrho$ into a strategy on $\langle \sigma, \tau \rangle \stackrel{\sim}{\Rightarrow} \langle \sigma', \tau' \rangle$.
This is the main motivation to define a GwE $G = (G, =_G)$ such that $\sigma_1 =_G \sigma_2 \trianglelefteq G \stackrel{\sim}{\Rightarrow} G$ for all $\sigma_1, \sigma_2 : G$ does not necessarily hold.
\end{remark}

\begin{lemma}[Well-defined $\widehat{\Sigma}$]
\label{LemWellDefinedDependentPairSpace}
For any DGwE $B : A \to \mathcal{PGE}$, we have $\widehat{\Sigma}(A, B) \in \mathcal{PGE}$.
\end{lemma}
\begin{proof}
First, the games $\widehat{\Sigma}(A, B), =_{\widehat{\Sigma}(A, B)}$ are clearly wf again by ``initial protocols''.
Also, it is straightforward to see that the composition and the identities are well-defined, where note that 
\begin{align*}
B(\rho_2)^= \bullet \varrho_1 &: B(\rho_2) \bullet B(\rho_1) \bullet \tau_1 =_{B(\sigma_3)} B(\rho_2) \bullet \tau_2 \wedge B(\rho_2) \bullet B(\rho_1) = B(\rho_2 \bullet \rho_1) \\
\mathit{id}_\tau &: \mathit{id}_{B(\sigma)} \bullet \tau =_{B(\sigma)} \tau \wedge B(\mathit{id}_\sigma) = \mathit{id}_{B(\sigma)}
\end{align*}
for any $\tau : B(\sigma)$, $\tau_1 : B(\sigma_1)$, $\tau_2 : B(\sigma_2)$, $\sigma, \sigma_1, \sigma_2, \sigma_3 : A$, $\rho_1 : \sigma_1 =_A \sigma_2$, $\rho_2 : \sigma_2 =_A \sigma_3$, $\varrho_1 : B(\rho_1) \bullet \tau_1 =_{B(\sigma_2)} \tau_2$.

For the associativity of the composition, additionally let $\tau_3 : B(\sigma_3)$, $\tau_4 : B(\sigma_4)$, $\sigma_4 : A$, $\varrho_2 : B(\rho_2) \bullet \tau_2 =_{B(\sigma_3)} \tau_3$, $\varrho_3 : B(\rho_3) \bullet \tau_3 =_{B(\sigma_4)} \tau_4$; it suffices to show $\varrho_3 \odot (\varrho_2 \odot \varrho_1)  = (\varrho_3 \odot \varrho_2) \odot \varrho_1$.
Then observe that:
\begin{align*}
\varrho_3 \odot (\varrho_2 \odot \varrho_1) &= \varrho_3 \odot (\varrho_2 \bullet B(\rho_2)^= \bullet \varrho_1) \\
&= \varrho_3 \bullet B(\rho_3)^= \bullet (\varrho_2 \bullet B(\rho_2)^= \bullet \varrho_1) \\
&= \varrho_3 \bullet (B(\rho_3)^= \bullet \varrho_2) \bullet (B(\rho_3)^= \bullet B(\rho_2)^= \bullet \varrho_1) \ \text{(by the functoriality of $B(\rho_3)$)} \\
&= (\varrho_3 \bullet B(\rho_3)^= \bullet \varrho_2) \bullet (B(\rho_3)^= \bullet B(\rho_2)^=) \bullet \varrho_1 \\
&= (\varrho_3 \odot \varrho_2) \bullet B(\rho_3 \bullet \rho_2)^= \bullet \varrho_1 \ \text{(by the functoriality of $B$)} \\
&= (\varrho_3 \odot \varrho_2) \odot \varrho_1.
\end{align*}

For the unit law of the identities, it suffices to show $\varrho_1 \odot \mathit{id}_{\tau_1} = \varrho_1$ and $\mathit{id}_{\tau_2} \odot \varrho_1 = \varrho_1$.
Then observe that:
\begin{equation*}
\varrho_1 \odot \mathit{id}_{\tau_1} = \varrho_1 \bullet B(\rho_1)^= \bullet \mathit{id}_{\tau_1} = \varrho_1 \bullet \mathit{id}_{B(\rho_1) \bullet \tau_1}^= = \varrho_1 
\end{equation*}
as well as:
\begin{equation*}
\mathit{id}_{\tau_2} \odot \varrho_1 = \mathit{id}_{\tau_2} \bullet B(\mathit{id}_{\sigma_2})^= \bullet \varrho_1 = \mathit{id}_{B(\sigma_2)}^= \bullet \varrho_1 = \varrho_1.
\end{equation*}

Finally, $\varrho^\circleddash$ for any morphism $\varrho : B(\rho) \bullet \tau =_{B(\sigma')} \tau'$, where $\sigma, \sigma' : A$, $\tau : B(\sigma)$, $\tau' : B(\sigma')$, $\rho : \sigma =_A \sigma'$, satisfies:
\begin{equation*}
\varrho^\circleddash = B(\rho^{-1})^= \bullet \varrho^{-1} : B(\rho^{-1}) \bullet \tau' =_{B(\sigma)} \tau
\end{equation*}
where note that $\varrho^{-1}$ is the inverse of $\varrho$ in $B(\sigma')$.
In fact, $\varrho^\circleddash$ is the inverse of $\varrho$ with respect to $\odot$: 
\begin{align*}
\varrho^\circleddash \odot \varrho &= B(\rho^{-1})^= \bullet \varrho^{-1} \bullet B(\rho^{-1})^= \bullet \varrho \\
&= B(\rho^{-1})^= \bullet (\varrho^{-1} \bullet \varrho) \ \text{(by the functoriality of $B(\rho^{-1})$)} \\
&= B(\rho^{-1})^= \bullet (\mathit{id}_{B(\rho) \bullet \tau}) \\
&= \mathit{id}_{B(\rho^{-1}) \bullet B(\rho) \bullet \tau} \ \text{(by the functoriality of $B(\rho^{-1})$)} \\
&= \mathit{id}_{B(\rho^{-1} \bullet \rho) \bullet \tau} \ \text{(by the functoriality of $B$)} \\
&= \mathit{id}_{B(\mathit{id}_\sigma) \bullet \tau} \\
&= \mathit{id}_{\mathit{id}_{B(\sigma)} \bullet \tau} \ \text{(by the functoriality of $B$)} \\
&= \mathit{id}_\tau
\end{align*}
and 
\begin{align*}
\varrho \odot \varrho^\circleddash &= \varrho \bullet B(\rho)^= \bullet B(\rho^{-1})^= \bullet \varrho^{-1} \\
&= \varrho \bullet B(\rho \bullet \rho^{-1})^= \bullet \varrho^{-1} \ \text{(by the functoriality of $B$)} \\
&= \varrho \bullet B(\mathit{id}_{\sigma'})^= \bullet \varrho^{-1} \\
&= \varrho \bullet \mathit{id}_{B(\sigma')}^= \bullet \varrho^{-1} \ \text{(by the functoriality of $B$)} \\
&= \varrho \bullet \varrho^{-1} \\
&= \mathit{id}_{\tau'}.
\end{align*}
From this, it follows that the pairing $\langle \rho^{-1}, \varrho^\circleddash \rangle$ is a well-defined morphism from $\langle \sigma, \tau \rangle$ to $\langle \sigma', \tau' \rangle$, and it is in fact the inverse of $\langle \rho, \varrho \rangle$, completing the proof.
\end{proof}

By a set-theoretic analogy, $\widehat{\Sigma} (A, B)$ represents the space of pairs $(a, b)$, where $a \in A$, $b \in B(a)$.
As in the case of $\Pi$-types, $\widehat{\Sigma}$ is an ``intensional refinement'' of the groupoid interpretation of $\Sigma$-types \cite{hofmann1998groupoid}.

At the end of the present section, we give a (partial) interpretation of \textsf{Id}-types:
\begin{definition}[Identity space]
Given $G \in \mathcal{PGE}$, $\sigma_1, \sigma_2 : G$, the \emph{\bfseries identity space} $\widehat{\mathsf{Id}}_G(\sigma_1, \sigma_2)$ between $\sigma_1$ and $\sigma_2$ is the discrete groupoid with the underlying game $\sigma_1 =_G \sigma_2$.
\end{definition}

\begin{lemma}[Well-defined $\widehat{\mathsf{Id}}$]
\label{LemWellDefinedIdentitySpace}
For any $G \in \mathcal{PGE}$, $\sigma_1, \sigma_2 : G$, we have $\widehat{\mathsf{Id}}_G(\sigma_1, \sigma_2) \in \mathcal{PGE}$.
\end{lemma}
\begin{proof}
Straightforward.
\end{proof}

Note that we simply ``truncate'' higher morphisms as the groupoid model \cite{hofmann1998groupoid}.
To overcome this point, we may generalize the extensional structure of GwEs to \emph{$\omega$-groupoids}, and define the construction $\widehat{\textsf{Id}}$ to ``ascend by one-step'' the infinite hierarchy of the $\omega$-groupoid structure. 
Nevertheless, we leave this point as future work.

\section{Game-semantic Groupoid Interpretation of MLTT}
\label{InterpretationOfMLTT}
This section is the climax of the paper: It gives an interpretation of \textsf{MLTT} by GwEs.
Specifically, we equip the category $\mathcal{PGE}$ with the structure of a \emph{category with families (CwF)} in Section~\ref{CwFPGE}, and further game-theoretic $\Pi$-, $\Sigma$- and \textsf{Id}-types in Section~\ref{GameTheoreticTypeFormers}.

\subsection{The Game-semantic Category with Families}
\label{CwFPGE}
CwFs \cite{dybjer1996internal,hofmann1997syntax} are an abstract semantics for \textsf{MLTT}.
Roughly, a CwF is a category $\mathcal{C}$ with additional structures to interpret judgements common to all types as: 
\begin{align*}
\mathsf{\vdash \Gamma \ ctx} &\mapsto \Gamma \in \mathcal{C} \\
\mathsf{\Gamma \vdash A \ type} &\mapsto A \in \mathit{Ty}(\Gamma) \\
\mathsf{\Gamma \vdash a : A} &\mapsto \textstyle a \in \mathit{Tm}(\Gamma, A)
\end{align*}
where $\mathit{Ty}(\Gamma)$, $\mathit{Tm}(\Gamma, A)$ are assigned sets indexed by $\Gamma \in \mathcal{C}$, $A \in \mathit{Ty}(\Gamma)$.
To interpret specific types such as $\Pi$-, $\Sigma$- and \textsf{Id}-types, we need to equip $\mathcal{C}$ with the corresponding \emph{semantic type formers} \cite{hofmann1997syntax}.
By the \emph{soundness} of CwFs \cite{hofmann1997syntax}, this suffices to give a model of \textsf{MLTT} in $\mathcal{C}$, i.e., each context, type and term is interpreted in $\mathcal{C}$, and each judgmental equality is reflected by the corresponding semantic equality; see \cite{hofmann1997syntax} for the details.
We employ this framework because it is in general easier to prove that a structure forms a CwF than to directly show that it is a model of \textsf{MLTT}.

We first recall the definition of CwFs; our presentation follows that of \cite{hofmann1997syntax}.
\begin{definition}[\textsf{CwF}s \cite{dybjer1996internal,hofmann1997syntax}]
A \emph{\bfseries category with families (CwF)} is a structure $\mathcal{C} = (\mathcal{C}, \mathit{Ty}, \mathit{Tm}, \_\{\_\}, T, \_.\_, \mathit{p}, \mathit{v}, \langle\_,\_\rangle_\_)$,
where:
\begin{itemize}

\item $\mathcal{C}$ is a category of \emph{\bfseries contexts} and \emph{\bfseries context morphisms}

\item $\mathit{Ty}$ assigns, to each object $\Gamma \in \mathcal{C}$, a set $\mathit{Ty}(\Gamma)$ of \emph{\bfseries types} in the context $\Gamma$

\item $\mathit{Tm}$ assigns, to each pair of a context $\Gamma \in \mathcal{C}$ and a type $A \in \mathit{Ty}(\Gamma)$, a set $\mathit{Tm}(\Gamma, A)$ of \emph{\bfseries terms} of type $A$ in the context $\Gamma$

\item For each morphism $\phi : \Delta \to \Gamma$ in $\mathcal{C}$, $\_\{\_\}$ induces a function $\_\{\phi\} : \mathit{Ty}(\Gamma) \to \mathit{Ty}(\Delta)$
and a family $(\_\{\phi\}_A : \mathit{Tm}(\Gamma, A) \to \mathit{Tm}(\Delta, A\{\phi\}))_{A \in \mathit{Ty}(\Gamma)}$ of functions, called the \emph{\bfseries substitutions}

\item $T \in \mathcal{C}$ is a terminal object

\item $\_ . \_$ assigns, to each pair of a context $\Gamma \in \mathcal{C}$ and a type $A \in \mathit{Ty}(\Gamma)$, a context $\Gamma . A \in \mathcal{C}$, called the \emph{\bfseries comprehension} of $A$

\item $\mathit{p}$ associates each pair of a context $\Gamma \in \mathcal{C}$ and a type $A \in \mathit{Ty}(\Gamma)$ with a morphism $\mathit{p}(A) : \Gamma . A \to \Gamma$
in $\mathcal{C}$, called the \emph{\bfseries first projection} associated to $A$

\item $\mathit{v}$ associates each pair of a context $\Gamma \in \mathcal{C}$ and a type $A \in \mathit{Ty}(\Gamma)$ with a term $\mathit{v}_A \in \mathit{Tm}(\Gamma . A, A\{\mathit{p}(A)\})$, called the \emph{\bfseries second projection} associated to $A$

\item $\langle \_, \_ \rangle_\_$ assigns, to each triple of a morphism $\phi : \Delta \to \Gamma$ in $\mathcal{C}$, a type $A \in \mathit{Ty}(\Gamma)$ and a term $\tau \in \mathit{Tm}(\Delta, A\{\phi\})$, a morphism $\langle \phi, \tau \rangle_A : \Delta \to \Gamma . A$
in $\mathcal{C}$, called the \emph{\bfseries extension} of $\phi$ by $\tau$
\end{itemize}
that satisfies the following axioms:
\begin{itemize}

\item {\bfseries \sffamily Ty-Id.} $A \{ \textit{id}_\Gamma \} = A$

\item {\bfseries \sffamily Ty-Comp.} $A \{ \phi \circ \psi \} = A \{ \phi \} \{ \psi \}$

\item {\bfseries \sffamily Tm-Id.} $\varphi \{ \textit{id}_\Gamma \} = \varphi$

\item {\bfseries \sffamily Tm-Comp.} $\varphi \{ \phi \circ \psi \} = \varphi \{ \phi \} \{ \psi \}$

\item {\bfseries \sffamily Cons-L.} $\mathit{p}(A) \circ \langle \phi, \tau \rangle_A = \phi$

\item {\bfseries \sffamily Cons-R.} $\mathit{v}_A \{ \langle \phi, \tau \rangle_A \} = \tau$

\item {\bfseries \sffamily Cons-Nat.} $\langle \phi, \tau \rangle_A \circ \psi = \langle \phi \circ \psi, \tau \{ \psi \} \rangle_A$

\item {\bfseries \sffamily Cons-Id.} $\langle \mathit{p}(A), \mathit{v}_A \rangle_A = \textit{id}_{\Gamma . A}$

\end{itemize}
for all $\Gamma, \Delta, \Theta \in \mathcal{C}$, $A \in \mathit{Ty}(\Gamma)$, $\phi : \Delta \to \Gamma$, $\psi : \Theta \to \Delta$, $\varphi \in \mathit{Tm}(\Gamma, A)$, $\tau \in \mathit{Tm}(\Delta, A\{ \phi \})$.
\end{definition}

We now give our CwF of game-semantic groupoids and functors:
\begin{definition}
The CwF $\mathcal{PGE} = (\mathcal{PGE}, \mathit{Ty}, \mathit{Tm}, \_\{ \_ \}, I, \mathit{p}, \mathit{v}, \langle \_, \_ \rangle_{\_})$ is defined by:
\begin{itemize}

\item The underlying category $\mathcal{PGE}$ has been defined in Definition~\ref{DefPGE}.

\item Given $\Gamma \in \mathcal{PGE}$, $\mathit{Ty}(\Gamma) \stackrel{\mathrm{df. }}{=} \mathscr{D}(\Gamma)$, and given $A \in \mathscr{D}(\Gamma)$, $\mathit{Tm}(\Gamma, A) \stackrel{\mathrm{df. }}{=} \mathsf{st}(\widehat{\Pi}(\Gamma, A))$.

\item For each $\phi : \Delta \to \Gamma$ in $\mathcal{PGE}$, the function $\_\{ \phi \} : \mathscr{D}(\Gamma) \to \mathscr{D}(\Delta)$ is defined by: 
\begin{equation*}
A \{ \phi \} \stackrel{\mathrm{df. }}{=} A \circ \mathsf{fun}(\phi)
\end{equation*} 
i.e., the composition of functors for all $A \in \mathscr{D}(\Gamma)$, and the function $\_\{ \phi \}_A : \mathsf{st}(\widehat{\Pi}(\Gamma, A)) \to \mathsf{st}(\widehat{\Pi}(\Delta, A \{ \phi \}))$ for each $A \in \mathscr{D}(\Gamma)$ is defined by: 
\begin{align*}
\varphi \{ \phi \} &\stackrel{\mathrm{df. }}{=} \varphi \bullet \phi \\
\varphi \{ \phi \}^= &\stackrel{\mathrm{df. }}{=} \varphi^= \bullet \phi^=
\end{align*}
for all $\varphi : \widehat{\Pi}(\Gamma, A)$.

\item $I$ is the discrete GwE on the terminal game $I = (\emptyset, \emptyset, \emptyset, \{ \bm{\epsilon}, q_I, q_I \mathcal{N}(\{ \bm{\epsilon} \}) \})$.

\item $\Gamma . A \stackrel{\mathrm{df. }}{=} \widehat{\Sigma}(\Gamma, A)$, and $\mathit{p}(A) : \widehat{\Sigma}(\Gamma, A) \to \Gamma$, $\mathit{v}_A : \widehat{\Pi}(\widehat{\Sigma}(\Gamma, A), A\{ \mathit{p}(A) \})$ are defined by:
\begin{align*}
\mathit{p}(A) &\stackrel{\mathrm{df. }}{=} \&_{\langle \gamma, \sigma \rangle : \widehat{\Sigma}(\Gamma, A)} \mathit{der}_\gamma \\
\mathit{p}(A)^= &\stackrel{\mathrm{df. }}{=} \&_{\langle \rho, \varrho \rangle : \ =_{\widehat{\Sigma}(\Gamma, A)}} \mathit{der}_\rho \\
\mathit{v}_A &\stackrel{\mathrm{df. }}{=} \&_{\langle \gamma, \sigma \rangle : \widehat{\Sigma}(\Gamma, A)} \mathit{der}_\sigma \\
\mathit{v}_A^= &\stackrel{\mathrm{df. }}{=} \&_{\langle \rho, \varrho \rangle : \ =_{\widehat{\Sigma}(\Gamma, A)}} \mathit{der}_\varrho
\end{align*}
up to tags for disjoint union.

\item Given $\tau : \widehat{\Pi}(\Delta, A \{ \phi \})$, the extension $\langle \phi, \tau \rangle_A : \Delta \to \widehat{\Sigma}(\Gamma, A)$ is the pairing $\langle \phi, \tau \rangle$ equipped with the equality preservation $\langle \phi, \tau \rangle^= \stackrel{\mathrm{df. }}{=} \langle \phi^=, \tau^= \rangle$.
\end{itemize}
\end{definition}

\begin{theorem}[Well-defined $\mathcal{PGE}$]
\label{ThmPGE}
The structure $\mathcal{PGE}$ forms a well-defined CwF.
\end{theorem}

\begin{proof}
By lemmata~\ref{LemWellDefinedDependentFunctionSpace}, \ref{LemWellDefinedDependentPairSpace}, it is immediate to see that each component of $\mathcal{PGE}$ is well-defined except the substitution of terms and the extension.
Let $\Gamma, \Delta \in \mathcal{PGE}$, $A \in \mathscr{D}(\Gamma)$, $\phi : \Delta \to \Gamma$, $\varphi : \widehat{\Pi}(\Gamma, A)$.
It has been shown in \cite{yamada2016game} that $\varphi \{ \phi \} = \varphi \bullet \phi$ forms a strategy on the game $\widehat{\Pi}(\Delta, A \{ \phi \})$.
The equality-preservation $\varphi \{ \phi \}^= = \varphi^= \bullet \phi^= : \ =_\Delta \ \Rightarrow \ \uplus =_A$ satisfies:
\begin{equation*}
\varphi^= \bullet (\phi^= \bullet \vartheta) : A(\phi^= \bullet \vartheta) \bullet (\varphi \bullet (\phi \bullet \delta)) =_{A(\phi \bullet \delta')} \varphi \bullet (\phi \bullet \delta')
\end{equation*}
i.e.,
\begin{equation*}
(\varphi \bullet \phi)^= \bullet \vartheta : A \{ \phi \} (\vartheta) \bullet ((\varphi \bullet \phi) \bullet \delta) =_{A(\phi \bullet \delta')} (\varphi \bullet \phi) \bullet \delta'
\end{equation*}
for all $\delta, \delta' : \Delta$, $\vartheta : \delta =_\Delta \delta'$.
Therefore we may conclude that $\varphi \{ \phi \} = (\varphi \bullet \phi, \varphi^= \bullet \phi^=) : \widehat{\Pi}(\Delta, A \{ \phi \})$, showing that the substitution of terms is well-defined.

Next, for the context extension, let $\tau : \widehat{\Pi}(\Delta, A \{ \phi \})$.
Again, it has been shown in \cite{yamada2016game} that the pairing $\langle \phi, \tau \rangle$ forms a strategy on the game $\Delta \Rightarrow \widehat{\Sigma}(\Gamma, A)$; thus, it remains to show that it is ep. Then for any $\delta, \delta' : \Delta$, $\vartheta : \delta =_\Delta \delta'$, we have:
\begin{equation*}
\langle \phi, \tau \rangle^= \bullet \vartheta = \langle \phi^=, \tau^= \rangle \bullet \vartheta = \langle \tau^= \bullet \vartheta, \tau^= \bullet \vartheta \rangle
\end{equation*}
where $\phi^= \bullet \vartheta : \phi \bullet \delta =_\Gamma \phi \bullet \delta'$, $\tau^= \bullet \vartheta : A(\tau^= \bullet \vartheta) \bullet (\tau \bullet \delta) =_{A(\tau \bullet \delta')} \tau \bullet \delta'$, whence
\begin{equation*}
\langle \phi, \psi \rangle^= \bullet \vartheta : \langle \phi, \tau \rangle \bullet \delta =_{\widehat{\Sigma}(\Gamma, A)} \langle \phi, \tau \rangle \bullet \delta'
\end{equation*}
which shows that $\langle \phi, \tau \rangle^=$ preserves domain and codomain.
It also preserves composition:
\begin{align*}
\langle \phi, \tau \rangle^= \bullet (\vartheta' \bullet \vartheta) &= \langle \phi^= \bullet (\vartheta' \bullet \vartheta), \tau^= \bullet (\vartheta' \bullet \vartheta) \rangle \\
&= \langle (\phi^= \bullet \vartheta') \bullet (\phi^= \bullet \vartheta), (\tau^= \bullet \vartheta') \bullet (A\{\phi\}(\vartheta')^= \bullet \tau^= \bullet \vartheta) \rangle \\
&= \langle (\phi^= \bullet \vartheta') \bullet (\phi^= \bullet \vartheta), (\tau^= \bullet \vartheta') \bullet (A(\phi^= \bullet \vartheta')^= \bullet \tau^= \bullet \vartheta) \rangle \\
&= \langle (\phi^= \bullet \vartheta') \bullet (\phi^= \bullet \vartheta), (\tau^= \bullet \vartheta') \odot (\tau^= \bullet \vartheta) \rangle \\
&= \langle \phi^= \bullet \vartheta', \tau^= \bullet \vartheta' \rangle \bullet \langle \phi^= \bullet \vartheta, \tau^= \bullet \vartheta \rangle \\
&= (\langle \phi, \tau \rangle^= \bullet \vartheta') \bullet (\langle \phi, \tau \rangle^= \bullet \vartheta)
\end{align*}
for any $\delta'' : \Delta$, $\vartheta' : \delta' =_\Delta \delta''$.
Of course, it preserves identities as well:
\begin{align*}
\langle \phi, \tau \rangle^= \bullet \mathit{id}_\delta &= \langle \phi^= \bullet \mathit{id}_\delta, \tau^= \bullet \mathit{id}_\delta \rangle \\
&= \langle \mathit{id}_{\phi \bullet \delta}, \mathit{id}_{\tau \bullet \delta} \rangle \\
&= \mathit{id}_{\langle \phi \bullet \delta, \tau \bullet \delta \rangle} \\
&= \mathit{id}_{\langle \phi, \tau \rangle \bullet \delta}
\end{align*}
for all $\delta : \Delta$.
Therefore the pair $\langle \phi, \tau \rangle = (\langle \phi, \tau \rangle, \langle \phi^=, \tau^= \rangle)$ is a morphism $\Delta \to \widehat{\Sigma}(\Gamma, A)$ in $\mathcal{PGE}$, showing that the extension is well-defined.

Finally, we verify the required equations. Let $\Gamma, \Delta, \Theta \in \mathcal{PGE}$, $A \in \mathscr{D}(\Gamma)$, $\phi : \Delta \to \Gamma$, $\psi : \Theta \to \Delta$, $\varphi : \widehat{\Pi}(\Gamma, A)$, $\tau : \widehat{\Pi}(\Delta, A \{ \phi \})$. 
\begin{itemize}

\item \textsf{\bfseries Ty-Id.} $A\{ \mathit{id}_\Gamma \} = A \circ \mathsf{fun}(\mathit{der}_\Gamma) = A$

\item \textsf{\bfseries Ty-Comp.} $A\{ \phi \bullet \psi \} = A \circ \mathsf{fun}(\phi \bullet \psi) = A \circ (\mathsf{fun}(\phi) \circ \mathsf{fun}(\psi)) = (A \circ \mathsf{fun}(\phi)) \circ \mathsf{fun}(\psi) = A \{ \phi \} \circ \mathsf{fun}(\psi) = A \{ \phi \} \{ \psi \}$

\item \textsf{\bfseries Tm-Id.} $\varphi \{ \mathit{id}_\Gamma \} = \varphi \bullet \mathit{der}_\Gamma = \varphi \wedge \varphi \{ \mathit{id}_\Gamma \}^= = \varphi^= \bullet \mathit{der}_\Gamma^= = \varphi^= \bullet \mathit{der}_{=_\Gamma} = \varphi^=$

\item \textsf{\bfseries Tm-Comp.} $\varphi \{ \phi \bullet \psi \} = \varphi \bullet (\phi \bullet \psi) = (\varphi \bullet \phi) \bullet \psi = \varphi \{ \phi \} \{ \psi \} \wedge \varphi \{ \phi \bullet \psi \}^= = \varphi^= \bullet (\phi \bullet \psi)^= = \varphi^= \bullet (\phi^= \bullet \psi^=) = (\varphi^= \bullet \phi^=) \bullet \psi^= = \varphi \{ \phi \}^= \bullet \psi^= = \varphi \{ \phi \} \{ \psi \}^=$

\item \textsf{\bfseries Cons-L.} $\mathit{p}(A) \bullet \langle \phi, \tau \rangle = \phi \wedge \mathit{p}(A)^= \bullet \langle \phi, \tau \rangle^= = \mathit{p}(A)^= \bullet \langle \phi^=, \tau^= \rangle = \phi^=$

\item \textsf{\bfseries Cons-R.} $\mathit{v}_A \{ \langle \phi, \tau \rangle \} = \mathit{v}_A \bullet \langle \phi, \tau \rangle = \tau \wedge \mathit{v}_A \{ \langle \phi, \tau \rangle \}^= = \mathit{v}_A^= \bullet \langle \phi, \tau \rangle^= = \mathit{v}_A^= \bullet \langle \phi^=, \tau^= \rangle = \tau^=$

\item \textsf{\bfseries Cons-Nat.} $\langle \phi, \tau \rangle \bullet \psi = \langle \phi \bullet \psi, \tau \bullet \psi \rangle \wedge \langle \phi, \tau \rangle^= \bullet \psi^= = \langle \phi^=, \tau^= \rangle \bullet \psi^= = \langle \phi^= \bullet \psi^=, \tau^= \bullet \psi^= \rangle = \langle (\phi \bullet \psi)^=, \tau \{ \psi \}^= \rangle = \langle \phi \bullet \psi, \tau \{ \psi \} \rangle^=$

\item \textsf{\bfseries Cons-Id.} $\langle \mathit{p}(A), \mathit{v}_A \rangle = \mathit{der}_{\widehat{\Sigma}(\Gamma, A)} \wedge \langle \mathit{p}(A), \mathit{v}_A \rangle^= = \langle \mathit{p}(A)^=, \mathit{v}_A^= \rangle = \mathit{der}_{=_{\widehat{\Sigma}(\Gamma, A)}}$

\end{itemize}
which completes the proof.
\end{proof}

\subsection{Game-semantic Type Formers}
\label{GameTheoreticTypeFormers}
As stated before, a CwF gives only an interpretation of the syntax common to all types.
Thus, for a ``full interpretation'' of \textsf{MLTT}, we need to equip $\mathcal{PGE}$ with \emph{semantic type formers}. 
We address this point in the present section.

\subsubsection{Game-semantic Dependent Function Types}
\label{Prod}
We begin with $\Pi$-types. First, we recall the general, categorical interpretation of $\Pi$-types.
\begin{definition}[\textsf{CwF}s with $\Pi$-types \cite{hofmann1997syntax}]
A \textsf{CwF} $\mathcal{C}$ \emph{\bfseries supports $\bm{\Pi}$-types} if:
\begin{itemize}

\item {\bfseries \sffamily $\bm{\Pi}$-Form.} For any $\Gamma \in \mathcal{C}$, $A \in \mathit{Ty}(\Gamma)$, $B \in \mathit{Ty}(\Gamma . A)$, there is a type
\begin{equation*}
\textstyle \Pi (A, B) \in \mathit{Ty}(\Gamma).
\end{equation*}

\item {\bfseries \sffamily $\bm{\Pi}$-Intro.} If $\varphi \in \mathit{Tm}(\Gamma . A, B)$, then there is a term
\begin{equation*}
\textstyle \lambda_{A, B} (\varphi) \in \mathit{Tm}(\Gamma, \Pi (A, B)). 
\end{equation*}

\item {\bfseries \sffamily $\bm{\Pi}$-Elim.} If $\kappa \in \mathit{Tm}(\Gamma, \Pi (A, B))$, $\tau \in \mathit{Tm}(\Gamma, A)$, then there is a term 
\begin{equation*}
\textstyle \mathit{App}_{A, B} (\kappa, \tau) \in \mathit{Tm}(\Gamma, B\{ \overline{\tau} \}) 
\end{equation*}
where $\overline{\tau} \stackrel{\mathrm{df. }}{=} \langle \textit{id}_\Gamma, \tau \rangle_A : \Gamma \to \Gamma . A$.

\item {\bfseries \sffamily $\bm{\Pi}$-Comp.} For all $\varphi \in \mathit{Tm}(\Gamma . A, B)$, $\tau \in \mathit{Tm}(\Gamma, A)$, 
\begin{equation*}
\mathit{App}_{A, B} (\lambda_{A, B} (\varphi) , \tau) = \varphi \{ \overline{\tau} \}.
\end{equation*}

\item {\bfseries \sffamily $\bm{\Pi}$-Subst.} For any $\Delta \in \mathcal{C}$, $\phi : \Delta \to \Gamma$ in $\mathcal{C}$, 
\begin{equation*}
\textstyle \Pi (A, B) \{ \phi \} = \Pi (A\{\phi\}, B\{\phi^+\})
\end{equation*}
where $\phi^+ \stackrel{\mathrm{df. }}{=} \langle \phi \circ \mathit{p}(A\{ \phi \}), \mathit{v}_{A\{\phi\}} \rangle_A : \Delta . A\{\phi\} \to \Gamma . A$.

\item {\bfseries \sffamily $\bm{\lambda}$-Subst.} For all $\varphi \in \mathit{Tm} (\Gamma. A, B)$, 
\begin{equation*}
\lambda_{A, B} (\varphi) \{ \phi \} = \lambda_{A\{\phi\}, B\{\phi^+\}} (\varphi \{ \phi^+ \}) \in \mathit{Tm} (\Delta, \textstyle \Pi (A\{\phi\}, B\{ \phi^+\}))
\end{equation*}
where note that $\Pi (A\{\phi\}, B\{ \phi^+ \}) \in \mathit{Ty}(\Delta)$.

\item {\bfseries \sffamily App-Subst.} Under the same assumption, 
\begin{equation*}
\mathit{App}_{A, B} (\kappa, \tau) \{ \phi \} = \mathit{App}_{A\{\phi\}, B\{\phi^+\}} (\kappa \{ \phi \}, \tau \{ \phi \}) \in \mathit{Tm} (\Delta, B\{ \overline{\tau} \circ \phi \})
\end{equation*}
where note that $\kappa\{\phi\} \in \mathit{Tm} (\Delta, \Pi (A\{\phi\}, B\{\phi^+\}))$, $\tau\{\phi\} \in \mathit{Tm}(\Delta, A\{\phi\})$, and $\phi^+ \circ \overline{\tau\{\phi\}} =  \langle \phi \circ \mathit{p}(A\{\phi\}), \mathit{v}_{A\{\phi\}} \rangle_A \circ \langle \textit{id}_\Delta, \tau\{\phi\} \rangle_{A\{\phi\}} = \langle \phi, \tau \{\phi\} \rangle_A = \langle \textit{id}_\Gamma, \tau \rangle_A\circ \phi = \overline{\tau} \circ \phi $.
\end{itemize}

Furthermore, $\mathcal{C}$ \emph{\bfseries supports $\bm{\Pi}$-types in the strict sense} if it additionally satisfies the following:
\begin{itemize}
\item {\bfseries \sffamily $\bm{\lambda}$-Uniq.} For all $\mu : \widehat{\Pi}(\widehat{\Sigma}(\Gamma, A), \Pi(A, B)\{\mathit{p}(A)\})$, 
\begin{equation*}
\lambda_{A\{\mathit{p}(A)\}, B\{\mathit{p}(A)^+\}} (\mathit{App}_{A\{\mathit{p}(A)\}, B\{\mathit{p}(A)^+\}}(\mu, \mathit{v}_A)) = \mu.
\end{equation*}
Note that it corresponds to the rule \textsf{$\Pi$-Uniq} or \emph{$\eta$-rule} in \textsf{MLTT}.
\end{itemize}

\end{definition}

Let us now give our game-semantic interpretation of $\Pi$-types:
\begin{theorem}[Game-semantic $\Pi$-types]
\label{ThmGameTheoreticPiTypes}
The CwF $\mathcal{PGE}$ supports $\Pi$-types.
\end{theorem}
\begin{proof}
Let $\Gamma \in \mathcal{PGE}$, and $A : \Gamma \to \mathcal{PGE}$, $B : \widehat{\Sigma}(\Gamma, A) \to \mathcal{PGE}$ be DGwEs.

\paragraph{\textsc{$\Pi$-Form.}}
For each $\gamma : \Gamma$, we define a DGwE $B_\gamma :A(\gamma) \to \mathcal{PGE}$ by: 
\begin{align*}
B_\gamma(\sigma) &\stackrel{\mathrm{df. }}{=} B(\langle \gamma, \sigma \rangle) \in \mathcal{PGE} \\
B_\gamma(\varrho) &\stackrel{\mathrm{df. }}{=} B(\langle \mathit{id}_\gamma, \varrho \rangle) : B(\langle \gamma, \sigma \rangle) \to B(\langle \gamma, \tilde{\sigma} \rangle)
\end{align*}
for all $\sigma, \tilde{\sigma} : A(\gamma)$, $\varrho : \sigma =_{A(\gamma)} \tilde{\sigma}$.
$B_\gamma$ is clearly well-defined since the functoriality and uniformity of $B_\gamma$ follow from those of $B$. 
We usually write $B(\gamma, \sigma)$, $B(\mathit{id}, \varrho)$ for $B(\langle \gamma, \sigma \rangle)$, $B(\langle \mathit{id}_\gamma, \varrho \rangle)$, respectively.

For each $\rho : \gamma =_\Gamma \gamma'$, we may define a natural transformation $\mathsf{nat}(B_\rho) : B_\gamma \Rightarrow B_{\gamma'} \{ A(\rho) \} : A(\gamma) \to \mathcal{PGE}$ whose components are defined by:
\begin{align*}
\mathsf{nat}(B_\rho)_\sigma \stackrel{\mathrm{df. }}{=} B(\rho, \mathit{id}_{A(\rho) \bullet \sigma}) : B(\gamma, \sigma) \to B(\gamma', A(\rho) \bullet \sigma)
\end{align*}
for all $\sigma : A(\gamma)$.
In fact, $\mathsf{nat}(B_\rho)$ is natural in $A(\gamma)$: Given $\sigma_1, \sigma_2 : A(\gamma)$, $\varrho : \sigma_1 =_{A(\gamma)} \sigma_2$, the diagram
\begin{diagram}
B_\gamma(\sigma_1) & \rTo^{\mathsf{nat}(B_\rho)_{\sigma_1}} & B_{\gamma'}\{A(\rho)\}(\sigma_1) \\
\dTo^{B_\gamma(\varrho)} && \dTo_{B_{\gamma'}\{A(\rho)\}(\varrho)} \\
B_\gamma(\sigma_2) & \rTo_{\mathsf{nat}(B_\rho)_{\sigma_2}} & B_{\gamma'}\{ A(\rho) \}(\sigma_2)
\end{diagram}
commutes because 
\begin{align*}
B_{\gamma'} \{ A(\rho) \} (\varrho) \bullet \mathsf{nat}(B_\rho)_{\sigma_1} &= B(\mathit{id}_{\gamma'}, A(\rho)^= \bullet \varrho) \bullet B(\rho, \mathit{id}_{A(\rho) \bullet \sigma_1}) \\
&= B(\rho, \mathit{id}_{A(\rho) \bullet \sigma_2}) \bullet B(\mathit{id}_\gamma, \varrho) \\
&= \mathsf{nat}(B_\rho)_{\sigma_2} \bullet B_\gamma(\varrho).
\end{align*}
By the uniformity of $B$, we may organize strategies $\{ \mathsf{nat}(B_\rho)_\sigma \ \! | \ \! \sigma : A(\gamma) \ \! \}$ into a single strategy $B_\rho : \uplus B_\gamma \Rightarrow \uplus B_{\gamma'}\{ A(\rho) \}$, and similarly $\{ \mathsf{nat}(B_\rho)_\sigma^= \ \! | \ \! \sigma : A(\gamma) \ \! \}$ into $B_\rho^= : \uplus =_{B_\gamma} \Rightarrow \uplus =_{B_{\gamma'} \{ A(\gamma) \}}$.

We then define a DGwE $\Pi(A, B) : \Gamma \to \mathcal{PGE}$ by:
\begin{align*}
\textstyle \Pi(A, B)(\gamma) &\stackrel{\mathrm{df. }}{=} \textstyle \widehat{\Pi}(A(\gamma), B_\gamma) \\ 
\textstyle \Pi(A, B)(\rho) &\stackrel{\mathrm{df. }}{=} \textstyle \rho_{\Pi(A, B)} : \widehat{\Pi}(A(\gamma), B_\gamma) \to \widehat{\Pi}(A(\gamma'), B_{\gamma'})
\end{align*} 
for all $\gamma, \gamma' : \Gamma$, $\rho : \gamma =_{\Gamma} \gamma'$, where $\rho_{\Pi(A, B)} : \widehat{\Pi}(A(\gamma), B_\gamma) \to \widehat{\Pi}(A(\gamma'), B_{\gamma'})$ is the ep-strategy 
\begin{equation*}
\rho_{\Pi(A, B)} \stackrel{\mathrm{df. }}{=} \&_{\phi : \widehat{\Pi}(A(\gamma), B_\gamma)} \phi \leftrightarrows A(\rho^{-1})^\dagger ; \phi^\dagger ; B_\rho : \textstyle \widehat{\Pi}(A(\gamma), B_\gamma) \Rightarrow \widehat{\Pi}(A(\gamma'), B_{\gamma'})
\end{equation*}
for which we define:
\begin{align*}
\phi \leftrightarrows A(\rho^{-1})^\dagger ; \phi^\dagger ; B_\rho &\stackrel{\mathrm{df. }}{=} \{ \bm{s} \upharpoonright \phi^{[1]}, A(\gamma'), \uplus B_{\gamma'} \ \! | \ \! \bm{s} \in \phi^{[1]} \stackrel{\sim}{\Rightarrow} A(\rho^{-1})^\dagger \ddagger (\phi^{[2]})^\dagger \ddagger B_\rho, \\ & \ \ \ \ \ \ \ \ \forall \bm{t}mn \preceq \bm{s} . \ \! \mathsf{even}(\bm{t}) \wedge m \in M_\phi \Rightarrow (\bm{t} \upharpoonright \phi^{[1]}, \phi^{[2]}) . mn \in \mathit{der}_\phi \}
\end{align*}
where clearly $\phi^\dagger ; B_\rho : \widehat{\Pi}(A(\gamma), B_{\gamma'}\{ A(\rho) \})$, whence $A(\rho^{-1})^\dagger ; \phi^\dagger ; B_\rho : \widehat{\Pi}(A(\gamma'), B_{\gamma'})$. 
Also, $\rho_{\Pi(A, B)}$ is uniform since its components $\phi \leftrightarrows A(\rho^{-1})^\dagger ; \phi^\dagger ; B_\rho$ solely depend on the ``behavior'' of $\phi$. 
Therefore it follows that $\rho_{\widehat{\Pi}(A, B)}$ is a well-defined strategy on $\widehat{\Pi}(A(\gamma), B_\gamma) \Rightarrow \widehat{\Pi}(A(\gamma'), B_{\gamma'})$.
For brevity, from now on, let us write $\phi \leftrightarrows B_\rho \bullet \phi \bullet A(\rho^{-1})$ for $\phi \leftrightarrows A(\rho^{-1})^\dagger ; \phi^\dagger ; B_\rho$.

The equality-preservation 
\begin{equation*}
\rho_{\Pi(A, B)}^= : \textstyle (\widehat{\Pi}(A(\gamma)^{[1]}, B_\gamma^{[1]}) \stackrel{\sim}{\Rightarrow} \widehat{\Pi}(A(\gamma)^{[2]}, B_\gamma^{[2]})) \stackrel{\sim}{\Rightarrow} (\widehat{\Pi}(A(\gamma')^{[1]}, B_{\gamma'}^{[1]}) \stackrel{\sim}{\Rightarrow} \widehat{\Pi}(A(\gamma')^{[2]}, B_{\gamma'}^{[2]}))
\end{equation*}
is defined by:
\begin{equation*}
\rho_{\Pi(A, B)}^= \stackrel{\mathrm{df. }}{=} \&_{\phi_1, \phi_2 : \widehat{\Pi}(A(\gamma), B_\gamma), \nu : \phi_1 =_{\widehat{\Pi}(A(\gamma), B_\gamma)} \phi_2} \nu \leftrightarrows B(\rho)^= \bullet \nu \bullet A(\rho^{-1})^=.
\end{equation*}
As an illustration, given $\phi_1, \phi_2 : \widehat{\Pi}(A(\gamma), B_\gamma)$, $\nu : \phi_1 =_{\widehat{\Pi}(A(\gamma), B_\gamma)} \phi_2$, the strategy 
\begin{equation*}
\nu \leftrightarrows B(\rho)^= \bullet \nu \bullet A(\rho^{-1})^= : \phi_1 =_{\widehat{\Pi}(A(\gamma), B_\gamma)} \phi_2 \stackrel{\sim}{\Rightarrow} \rho_{\Pi(A, B)} \bullet \phi_1 =_{\widehat{\Pi}(A(\gamma'), B_{\gamma'})} \rho_{\Pi(A, B)} \bullet \phi_2
\end{equation*}
may be represented as the strategy between $\nu$ and $\rho_{\Pi(A, B)}^= \bullet \nu$ in the following diagram:
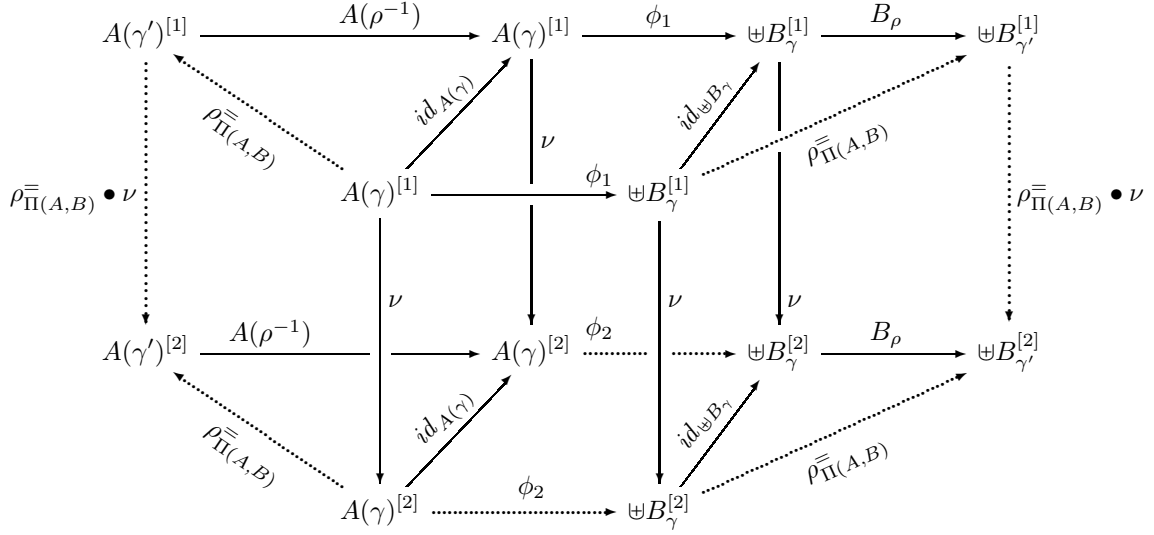
\begin{diagram}
A(\gamma')^{[1]} && \rTo^{A(\rho^{-1})} && A(\gamma)^{[1]} && \rTo^{\phi_1} && \uplus B_{\gamma}^{[1]} &  &\rTo^{B_\rho}&& \uplus B_{\gamma'}^{[1]} \\
&\luDotsto(2, 3)_{\rho_{\Pi(A, B)}^=}&&\ruTo(2, 3)^{\mathit{id}_{A(\gamma)}}& &&&\ruTo(2, 3)^{\mathit{id}_{\uplus B_\gamma}} &\vLine&&&\ruDotsto(6, 3)_{\rho_{\Pi(A, B)}^=}& \\
&&&&\vLine_\nu&&&&\HonV&&&& \\
\dDotsto^{\rho_{\Pi(A, B)}^= \bullet \nu}&&A(\gamma)^{[1]} && \HonV & \rTo^{\phi_1} &\uplus B_\gamma^{[1]}&&&&&& \dDotsto_{\rho_{\Pi(A, B)}^= \bullet \nu} \\
&&&&&&&&&&&& \\
&&\dTo_\nu&&\dTo&&\dTo_\nu&&\dTo_\nu&&&& \\
A(\gamma')^{[2]} &\hLine^{A(\rho^{-1})}& \VonH & \rTo &A(\gamma)^{[2]} &\hDots^{\phi_2} &\VonH&\rDotsto& \uplus B_{\gamma}^{[2]} & &\rTo^{B_\rho} && \uplus B_{\gamma'}^{[2]} \\
&\luDotsto(2, 3)_{\rho_{\Pi(A, B)}^=}&&\ruTo(2, 3)^{\mathit{id}_{A(\gamma)}}&&&&\ruTo(2, 3)^{\mathit{id}_{\uplus B_\gamma}}&&&&\ruDotsto(6,3)_{\rho_{\Pi(A, B)}^=}& \\
&&&&&&&&&&&& \\
&&A(\gamma)^{[2]}&&\rDotsto^{\phi_2}&&\uplus B_\gamma^{[2]}&&&&&&
\end{diagram}
where the strategies on dotted arrows indicate that they do not ``control'' the play in the diagram, but rather they ``occur'' as the result of the play.
It is immediate from the diagram that $\rho_{\Pi(A, B)}^=$ preserves composition and identities.
Also, it is obvious that $\rho_{\Pi(A, B)}^= \bullet \nu$ satisfies the first two axioms for morphisms between dependent functions (see Definition~\ref{DefDependentFunctionSpace}).

For the third axiom, let $\sigma' : A(\gamma')$, $\phi_1, \phi_2 : \widehat{\Pi}(A(\gamma), B_\gamma)$, $\nu : \phi_1 =_{\widehat{\Pi}(A(\gamma), B_\gamma)} \phi_2$ be fixed; we have to show $(\rho_{\Pi(A, B)}^= \bullet \nu)_{\sigma'} : (\rho_{\Pi(A, B)} \bullet \phi_1) \bullet \sigma' =_{B_{\gamma'}(\sigma')} (\rho_{\Pi(A, B)} \bullet \phi_2) \bullet \sigma'$.
By the definition, we have $\nu_{A(\rho^{-1}) \bullet \sigma'} : \phi_1 \bullet A(\rho^{-1}) \bullet \sigma' =_{B_\gamma(A(\rho^{-1}) \bullet \sigma')} \phi_2 \bullet A(\rho^{-1}) \bullet \sigma'$, whence the definition of $\rho_{\Pi(A, B)}^=$ implies the desired property:
\begin{equation*}
(\rho_{\Pi(A, B)}^= \bullet \nu)_{\sigma'} = B_\rho^= \bullet \nu_{A(\rho^{-1}) \bullet \sigma'} : B_\rho \bullet \phi_1 \bullet A(\rho^{-1}) \bullet \sigma' =_{B_{\gamma'}(\sigma')} B_\rho \bullet \phi_2 \bullet A(\rho^{-1}) \bullet \sigma'.
\end{equation*}

Now, we show the fourth axiom or the naturality of $\rho_{\Pi(A, B)}^= \bullet \nu$. Let $\sigma'_1, \sigma'_2 : A(\gamma')$, $\varrho' : \sigma'_1 =_{A(\gamma')} \sigma'_2$; we have to show that the following diagram commutes:
\begin{diagram}
(\rho_{\Pi(A, B)} \bullet \phi_1) \bullet \sigma'_1 & \rTo^{(\rho_{\Pi(A, B)}^= \bullet \nu)_{\sigma'_1}} & (\rho_{\Pi(A, B)} \bullet \phi_2) \bullet \sigma'_1 \\
\dTo^{(\rho_{\Pi(A, B)} \bullet \phi_1)^= \bullet \varrho'} && \dTo_{(\rho_{\Pi(A, B)} \bullet \phi_2)^= \bullet \varrho'} \\
(\rho_{\Pi(A, B)} \bullet \phi_1) \bullet \sigma'_2 & \rTo_{(\rho_{\Pi(A, B)}^= \bullet \nu)_{\sigma'_2}} & (\rho_{\Pi(A, B)} \bullet \phi_2) \bullet \sigma'_2
\end{diagram}
However, it is immediate:
\begin{align*}
((\rho_{\Pi(A, B)} \bullet \phi_2)^= \bullet \varrho') \bullet (\rho_{\Pi(A, B)}^= \bullet \nu)_{\sigma'_1} &= (B_\rho \bullet \phi_2 \bullet A(\rho^{-1}))^= \bullet \varrho' \bullet B_\rho^= \bullet \nu_{A(\rho^{-1}) \bullet \sigma'_1} \\
&= B_\rho^= \bullet (\phi_2^= \bullet A(\rho^{-1})^= \bullet \varrho') \bullet B_\rho^= \bullet \nu_{A(\rho^{-1}) \bullet \sigma'_1} \\
&= B_\rho^= \bullet (\phi_2^= \bullet A(\rho^{-1})^= \bullet \varrho' \bullet \nu_{A(\rho^{-1}) \bullet \sigma'_1}) \\
& \ \ \text{(by the functoriality of $B_\rho$)} \\
&= B_\rho^= \bullet (\nu_{A(\rho^{-1}) \bullet \sigma'_2} \bullet \phi_1^= \bullet A(\rho^{-1})^= \bullet \varrho') \\
& \ \ \text{(by the naturality of $\nu$)} \\
&= (B_\rho^= \bullet \nu_{A(\rho^{-1}) \bullet \sigma'_2}) \bullet (B_\rho^= \bullet \phi_1^= \bullet A(\rho^{-1})^= \bullet \varrho') \\
& \ \ \text{(by the functoriality of $B_\rho$)} \\
&= (\rho_{\Pi(A, B)}^= \bullet \nu)_{\sigma'_2} \bullet ((\rho_{\Pi(A, B)} \bullet \phi_1)^= \bullet \varrho').
\end{align*} 

By the definition, the strategies $\rho_{\Pi(A, B)}, \rho_{\Pi(A, B)}^=$ are clearly both total and wb.
Also, since $B_\rho, A(\rho^{-1})$ are both invertible and ``copy-cat-like'', so are $\rho_{\Pi(A, B)}, \rho_{\Pi(A, B)}^=$.
Therefore they are innocent and noetherian as well just by the same reason as copy-cats (and derelictions).

Hence, we have shown that the pair $\rho_{\Pi(A, B)} = (\rho_{\Pi(A, B)}, \rho_{\Pi(A, B)}^=)$ is a morphism $\widehat{\Pi}(A(\gamma), B_\gamma) \to \widehat{\Pi}(A(\gamma'), B_{\gamma'})$ in the category $\mathcal{PGE}$.
It remains to establish that $\Pi(A, B)$ preserves composition and identities. 
For composition, let $\rho : \gamma =_\Gamma \gamma'$, $\rho : \gamma' =_\Gamma \gamma''$.
Then observe that:
\begin{align*}
(\rho' \bullet \rho)_{\Pi(A, B)} &= \&_{\phi : \widehat{\Pi}(A(\gamma), B_\gamma)} \phi \leftrightarrows B_{\rho' \bullet \rho}
 \bullet \phi \bullet A((\rho' \bullet \rho)^{-1}) \\
&= \&_{\phi : \widehat{\Pi}(A(\gamma), B_\gamma)} \phi \leftrightarrows B_{\rho'} \bullet B_{\rho}
 \bullet \phi \bullet A(\rho^{-1}) \bullet A(\rho'^{-1}) \\
&= (\&_{\phi' : \widehat{\Pi}(A(\gamma'), B_{\gamma'})} \phi' \leftrightarrows B_{\rho'}
 \bullet \phi' \bullet A(\rho'^{-1})) \bullet (\&_{\phi : \widehat{\Pi}(A(\gamma), B_\gamma)} \phi \leftrightarrows B_{\rho}
 \bullet \phi \bullet A(\rho^{-1})) \\
&= \rho'_{\Pi(A, B)} \bullet \rho_{\Pi(A, B)} 
\end{align*}
where note that:
\begin{align*}
B_{\rho' \bullet \rho} &= \&_{\sigma : A(\gamma), \tau : B_\gamma(\sigma)} B(\rho' \bullet \rho, \mathit{id}_{A(\rho' \bullet \rho) \bullet \sigma})_\tau \\
&= \&_{\sigma : A(\gamma), \tau : B_\gamma(\sigma)} B(\rho' \bullet \rho, \mathit{id}_{A(\rho' \bullet \rho) \bullet \sigma} \bullet A(\rho')^= \bullet \mathit{id}_{A(\rho) \bullet \sigma})_\tau \\
&= \&_{\sigma : A(\gamma), \tau : B_\gamma(\sigma)} B(\rho' \bullet \rho, \mathit{id}_{A(\rho' \bullet \rho) \bullet \sigma} \odot \mathit{id}_{A(\rho) \bullet \sigma})_\tau \\
&= \&_{\sigma : A(\gamma), \tau : B_\gamma(\sigma)} B((\rho', \mathit{id}_{A(\rho') \bullet A(\rho) \bullet \sigma}) \bullet (\rho, \mathit{id}_{A(\rho) \bullet \sigma}))_\tau \\
&= \&_{\sigma : A(\gamma), \tau : B_\gamma(\sigma)} (B(\rho', \mathit{id}_{A(\rho') \bullet A(\rho) \bullet \sigma}) \bullet B(\rho, \mathit{id}_{A(\rho) \bullet \sigma}))_\tau \\
&= (\&_{\sigma' : A(\gamma'), \tau' : B_{\gamma'}(\sigma')} B(\rho', \mathit{id}_{A(\rho') \bullet \sigma'})_{\tau'}) \bullet (\&_{\sigma : A(\gamma), \tau : B_\gamma(\sigma)} B(\rho, \mathit{id}_{A(\rho) \bullet \sigma})_\tau) \\
&= B_{\rho'} \bullet B_\rho.
\end{align*}
To see $(\rho' \bullet \rho)_{\Pi(A, B)}^= = \rho'^=_{\Pi(A, B)} \bullet \rho_{\Pi(A, B)}^=$, it suffices to observe that the outer part of the digram 
\begin{diagram}
A(\gamma'')^{[1]} & \rTo^{A((\rho' \bullet \rho)^{-1})} & A(\gamma)^{[1]} & \rDotsto^{\phi_1} & \uplus B_{\gamma}^{[1]} & \rTo^{B_{\rho' \bullet \rho}} & \uplus B_{\gamma''}^{[1]} \\
\dDotsto_{(\rho' \bullet \rho)_{\Pi(A, B)}^= \bullet \nu} &&\dTo_{\nu} && \dTo^{\nu} & & \dDotsto^{(\rho' \bullet \rho)_{\Pi(A, B)}^= \bullet \nu} \\
A(\gamma'')^{[2]} &  \rTo_{A((\rho' \bullet \rho)^{-1})} & A(\gamma)^{[2]} & \rDotsto^{\phi_2} & \uplus B_{\gamma}^{[2]} &\rTo_{B_{\rho' \bullet \rho}}  & \uplus B_{\gamma''}^{[2]}
\end{diagram}
is equal to that of the diagram
\begin{diagram}
A(\gamma'')^{[1]} & \rTo^{A(\rho'^{-1})}  & A(\gamma')^{[1]} & \rTo^{A(\rho^{-1})} & A(\gamma)^{[1]} & \rDotsto^{\phi_1} & \uplus B_{\gamma}^{[1]} & \rTo^{B_{\rho}} & \uplus B_{\gamma'}^{[1]} & \rTo^{B_{\rho'}} & \uplus B_{\gamma''}^{[1]} \\
\dDotsto_{(\rho' \bullet \rho)_{\Pi(A, B)}^= \bullet \nu} &&\dDotsto_{\rho_{\Pi(A, B)}^= \bullet \nu} &&\dTo_{\nu} && \dTo^{\nu} & & \dDotsto^{\rho_{\Pi(A, B)}^= \bullet \nu}&& \dDotsto^{(\rho' \bullet \rho)_{\Pi(A, B)}^= \bullet \nu} \\
A(\gamma'')^{[2]} &  \rTo_{A(\rho'^{-1})} & A(\gamma')^{[2]} &  \rTo_{A(\rho^{-1})} & A(\gamma)^{[2]} & \rDotsto_{\phi_2} & \uplus B_{\gamma}^{[2]} &\rTo_{B_{\rho}}  & \uplus B_{\gamma'}^{[2]} &\rTo_{B_{\rho'}}  & \uplus B_{\gamma''}^{[2]}
\end{diagram}
for all $\phi_1, \phi_2 : \widehat{\Pi}(A(\gamma), B_\gamma)$, $\nu : \phi_1 =_{\widehat{\Pi}(A(\gamma), B_\gamma)} \phi_2$.

For identities, let $\gamma : \Gamma$. Then, we have:
\begin{align*}
\textstyle \Pi(A, B)(\mathit{id}_\gamma) &= (\mathit{id}_\gamma)_{\Pi(A, B)} \\
&= \&_{\phi : \widehat{\Pi}(A(\gamma), B_\gamma)} \phi \leftrightarrows B_{\mathit{id}_\gamma} \bullet \phi \bullet A(\mathit{id}_\gamma^{-1}) \\
&= \&_{\phi : \widehat{\Pi}(A(\gamma), B_\gamma)} \phi \leftrightarrows \mathit{id}_{B_\gamma} \bullet \phi \bullet \mathit{id}_{A(\gamma)} \\
&= \&_{\phi : \widehat{\Pi}(A(\gamma), B_\gamma)} \phi \leftrightarrows \phi \\
&= \mathit{id}_{\widehat{\Pi}(A(\gamma), B_\gamma)} \\
&= \mathit{id}_{\Pi(A, B)(\gamma)}
\end{align*}
as well as:
\begin{align*}
\textstyle \Pi(A, B)(\mathit{id}_\gamma)^= &= (\mathit{id}_\gamma)_{\Pi(A, B)}^= \\
&= \&_{\phi_1, \phi_2 : \widehat{\Pi}(A(\gamma), B_\gamma), \nu : \phi_1 =_{\widehat{\Pi}(A(\gamma), B_\gamma)} \phi_2} \nu \leftrightarrows B_{\mathit{id}_\gamma}^= \bullet \nu \bullet A(\mathit{id}_\gamma^{-1})^= \\
&= \&_{\phi_1, \phi_2 : \widehat{\Pi}(A(\gamma), B_\gamma), \nu : \phi_1 =_{\widehat{\Pi}(A(\gamma), B_\gamma)}} \nu \leftrightarrows \mathit{id}_{B_\gamma}^= \bullet \nu \bullet \mathit{id}_{A(\gamma)}^= \\
&= \&_{\phi_1, \phi_2 : \widehat{\Pi}(A(\gamma), B_\gamma), \nu : \phi_1 =_{\widehat{\Pi}(A(\gamma), B_\gamma)}} \nu \leftrightarrows \nu \\
&= \mathit{id}_{=_{\widehat{\Pi}(A(\gamma), B_\gamma)}} \\
&= \mathit{id}_{\Pi(A, B)(\gamma)}^=
\end{align*}
Therefore $\Pi(A, B)$ in fact preserves composition and identities.

\paragraph{\textsc{$\Pi$-Intro.}}
As in the previous work \cite{yamada2016game}, we have the obvious correspondence 
\begin{equation*}
\textstyle \mathsf{ob}(\widehat{\Pi}(\widehat{\Sigma}(\Gamma, A), B)) \cong \mathsf{ob}(\widehat{\Pi}(\Gamma, \Pi(A, B)))
\end{equation*}
between objects.
Moreover, we may extend this correspondence to morphisms in the obvious way.
Given an ep-strategy $\psi : \widehat{\Pi}(\widehat{\Sigma}(\Gamma, A), B)$, let us write $\lambda_{A, B}(\psi) : \widehat{\Pi}(\Gamma, \Pi(A, B))$ for the corresponding ep-strategy.
In fact, it is straightforward to see that:
\begin{align*}
\lambda_{A, B}(\psi) \bullet \gamma &: \textstyle \widehat{\Pi}(A(\gamma), B_\gamma) \\
\lambda_{A, B}(\psi)^= \bullet \rho &: \textstyle \rho_{\Pi(A, B)} \bullet \lambda_{A, B}(\psi) \bullet \gamma =_{\widehat{\Pi}(A(\gamma'), B_{\gamma'})} \lambda_{A, B}(\psi) \bullet \gamma' \\
\lambda_{A, B}(\psi)^= \bullet (\rho' \bullet \rho) &= (\lambda_{A, B}(\psi)^= \bullet \rho') \bullet ((\rho')_{\Pi(A, B)}^= \bullet \lambda_{A, B}(\psi)^= \bullet \rho) \\
\lambda_{A, B}(\psi)^= \bullet (\mathit{id}_\gamma) &= \mathit{id}_{\lambda_{A, B}(\psi) \bullet \gamma}
\end{align*}
for all $\gamma, \gamma', \gamma'' : \Gamma$, $\rho : \gamma =_\Gamma \gamma'$, $\rho' : \gamma' =_\Gamma \gamma''$.

\paragraph{\textsc{$\Pi$-Elim.}}
Given $\varphi : \widehat{\Pi}(\Gamma, \Pi(A, B))$, $\alpha : \widehat{\Pi}(\Gamma, A)$, we define $\mathit{App}_{A, B}(\varphi, \alpha) : \widehat{\Pi}(\Gamma, B\{\overline{\alpha}\})$ by:
\begin{align*}
\mathit{App}_{A, B}(\varphi, \alpha) &\stackrel{\mathrm{df. }}{=} \lambda_{A, B}^{-1}(\varphi) \bullet \overline{\alpha} \\
\mathit{App}_{A, B}(\varphi, \alpha)^= &\stackrel{\mathrm{df. }}{=} \lambda_{A, B}^{-1}(\varphi)^= \bullet \overline{\alpha}^=
\end{align*}
where $\overline{\alpha} \stackrel{\mathrm{df. }}{=} \langle \mathit{der}_\Gamma, \alpha \rangle : \Gamma \to \widehat{\Sigma}(\Gamma, A)$.
I.e., we define $\mathit{App}_{A, B}(\varphi, \alpha) \stackrel{\mathrm{df. }}{=} \lambda_{A, B}^{-1}(\varphi) \{ \overline{\alpha} \}$, and so $\mathit{App}_{A, B}(\varphi, \alpha) : \widehat{\Pi}(\Gamma, B\{\overline{\alpha}\})$ by Theorem~\ref{ThmPGE}.

\paragraph{\textsc{$\Pi$-Comp.}}
By a simple calculation, we have:
\begin{align*}
\mathit{App}_{A, B}(\lambda_{A, B}(\psi), \alpha) &= \lambda_{A, B}^{-1}(\lambda_{A, B}(\psi)) \{ \overline{\alpha} \} \\
&= \psi \{ \overline{\alpha} \}.
\end{align*}

\paragraph{\textsc{$\Pi$-Subst.}} Given $\Delta \in \mathcal{PGE}$, $\phi : \Delta \to \Gamma$ in $\mathcal{PGE}$, we have: 
\begin{align*}
\textstyle \Pi (A, B) \{ \phi \}(\delta) &= \textstyle \widehat{\Pi} (A(\phi \bullet \delta), B_{\phi \bullet \delta}) \\
&= \textstyle \widehat{\Pi} (A\{\phi\}(\delta), B\{\phi^+\}_\delta) \\
&= \textstyle \Pi (A\{\phi\}, B\{\phi^+\})(\delta)
\end{align*}
for all $\delta : \Delta$, where $\phi^+ \stackrel{\mathrm{df. }}{=} \langle \phi \bullet \mathit{p}(A \{ \phi \}), \mathit{v}_{A\{\phi\}} \rangle : \widehat{\Sigma} (\Delta,  A\{\phi\}) \to \widehat{\Sigma} (\Gamma, A)$ and $B\{\phi^+\} \in \mathscr{D}(\widehat{\Sigma} (\Delta, A\{\phi\}))$. 
Note that $B_{\phi \bullet \delta} = B \{ \phi^+ \}_\delta : A\{\phi\}(\delta) \to \mathcal{PGE}$ because
\begin{align*}
B_{\phi \bullet \delta}(\sigma) &= B(\phi \bullet \delta, \sigma) \\
&= B\{ \phi^+ \}(\delta, \sigma) \\
&= B\{ \phi^+ \}_\delta (\sigma)
\end{align*}
for all $\sigma : A\{ \phi \}(\delta)$, and similarly
\begin{align*}
B_{\phi \bullet \delta}(\varrho) &= B(\mathit{id}_{\phi \bullet \delta}, \varrho) \\
&= B\{ \phi^+ \}(\mathit{id}_\delta, \varrho) \\
&= B\{ \phi^+ \}_\delta (\varrho)
\end{align*}
for all $\sigma, \sigma' : A\{ \phi \}(\delta)$, $\varrho : \sigma =_{A\{ \phi \}(\delta)} \sigma'$.
Since $\delta$ was arbitrary, we have shown that the object-maps of $\Pi (A, B) \{ \phi \}$ and $\Pi (A\{\phi\}, B\{\phi^+\})$ coincide.

Furthermore, we have:
\begin{align*}
\textstyle \Pi (A, B) \{ \phi \}(\vartheta) &= (\phi^= \bullet \vartheta)_{\Pi(A, B)} : \textstyle \widehat{\Pi}(A(\phi \bullet \delta), B_{\phi \bullet \delta}) \to \widehat{\Pi}(A(\phi \bullet \delta'), B_{\phi \bullet \delta'})  \\
&= \&_{\varphi : \widehat{\Pi}(A(\phi \bullet \delta), B_{\phi \bullet \delta})} \varphi \leftrightarrows B_{\phi^= \bullet \vartheta} \bullet \varphi \bullet A((\phi^= \bullet \vartheta)^{-1}) \\
&= \&_{\varphi : \widehat{\Pi}(A\{\phi\}, B\{\phi^+\})(\delta)} \varphi \leftrightarrows B\{\phi^+\}_\vartheta \bullet \varphi \bullet A\{\phi\}(\vartheta^{-1}) \\
&= \vartheta_{\Pi (A\{\phi\}, B\{\phi^+\})} \\
&= \textstyle \Pi (A\{\phi\}, B\{\phi^+\})(\vartheta)
\end{align*}
for all $\delta, \delta' : \Delta$, $\vartheta: \delta =_\Delta \delta'$.
Also, it is completely analogous to establish:
\begin{align*}
\textstyle \Pi (A, B) \{ \phi \}(\vartheta)^= = \textstyle \Pi (A\{\phi\}, B\{\phi^+\})(\vartheta)^=.
\end{align*}
I.e., $\Pi (A, B) \{ \phi \}(\vartheta)$ and $\Pi (A\{\phi\}, B\{\phi^+\})(\vartheta)$ are the same ep-strategy.
Since $\delta, \delta' : \Delta$, $\vartheta: \delta =_\Delta \delta'$ were arbitrarily chosen, it implies that the arrow-maps of $\Pi (A, B) \{ \phi \}$ and $\Pi (A\{\phi\}, B\{\phi^+\})$ coincide, which establishes the equality between functors $\Pi (A, B) \{ \phi \} = \Pi (A\{\phi\}, B\{\phi^+\}) : \Delta \to \mathcal{PGE}$.

\paragraph{\textsc{$\lambda$-Subst.}} For any $\psi : \widehat{\Pi} (\widehat{\Sigma} (\Gamma, A), B)$, it is not hard to see that:
\begin{align*}
&\lambda_{A, B} (\psi) \{ \phi \} = \lambda_{A, B} (\psi) \bullet \phi = \lambda_{A\{\phi\}, B\{\phi^+\}} (\psi \bullet \langle \phi \bullet \mathit{p}(A \{ \phi \}), \mathit{v}_{A\{\phi\}} \rangle) = \lambda_{A\{\phi\}, B\{\phi^+\}} (\psi \{ \phi^+ \}) \\
&\lambda_{A, B} (\psi) \{ \phi \}^= = \lambda_{A, B} (\psi)^= \bullet \phi^= = \lambda_{A\{\phi\}, B\{\phi^+\}} (\psi^= \bullet \langle \phi \bullet \mathit{p}(A \{ \phi \}), \mathit{v}_{A\{\phi\}} \rangle^=) = \lambda_{A\{\phi\}, B\{\phi^+\}} (\psi \{ \phi^+ \})^=
\end{align*}
showing that $\lambda_{A, B} (\psi) \{ \phi \}$ and $\lambda_{A\{\phi\}, B\{\phi^+\}} (\psi \{ \phi^+ \})$ are the same ep-strategy.

\paragraph{\textsc{App-Subst.}} Moreover, it is easy to see that:
\begin{align*}
\mathit{App}_{A, B} (\varphi, \alpha) \{ \phi \} &= (\lambda_{A, B}^{-1} (\varphi) \bullet \langle \mathit{der}_\Gamma, \alpha \rangle) \bullet \phi \\
&= \lambda_{A, B}^{-1} (\varphi) \bullet (\langle \mathit{der}_\Gamma, \alpha \rangle \bullet \phi) \\
&= \lambda_{A, B}^{-1} (\varphi) \bullet \langle \phi, \alpha \bullet \phi \rangle \\
&= \lambda_{A\{ \phi \}, B\{ \phi^+ \}}^{-1} (\varphi \bullet \phi) \bullet \overline{\alpha \bullet \phi} \\
&= \mathit{App}_{A\{ \phi \}, B\{ \phi^+ \}} (\varphi \{ \phi \}, \alpha \{ \phi \})
\end{align*}
as well as:
\begin{align*}
\mathit{App}_{A, B} (\varphi, \alpha) \{ \phi \}^= &= (\lambda_{A, B}^{-1} (\varphi)^= \bullet \langle \mathit{der}_\Gamma, \alpha \rangle^=) \bullet \phi^= \\
&= \lambda_{A, B}^{-1} (\varphi)^= \bullet (\langle \mathit{der}_\Gamma^=, \alpha^= \rangle \bullet \phi^=) \\
&= \lambda_{A, B}^{-1} (\varphi)^= \bullet \langle \phi^=, \alpha \{ \phi \}^= \rangle \\
&= \lambda_{A\{ \phi \}, B\{ \phi^+ \}}^{-1} (\varphi \bullet \phi)^= \bullet \overline{\alpha \{ \phi \}}^= \\
&= \mathit{App}_{A\{ \phi \}, B\{ \phi^+ \}} (\varphi \{ \phi \}, \alpha \{ \phi \})^=
\end{align*}
where $\overline{\tau \{ \phi \}} \stackrel{\mathrm{df. }}{=} \langle \mathit{der}_\Delta, \tau \{ \phi \} \rangle : \widehat{\Sigma}(\Delta, A \{ \phi \})$. 
Thus, we have shown that $\mathit{App}_{A, B} (\varphi, \alpha) \{ \phi \}$ and $\mathit{App}_{A\{ \phi \}, B\{ \phi^+ \}} (\varphi \{ \phi \}, \alpha \{ \phi \})$ are the same ep-strategy.

\paragraph{\textsc{$\lambda$-Unique.}} Finally, if $\mu : \widehat{\Pi}(\widehat{\Sigma}(\Gamma, A), \Pi(A, B)\{\mathit{p}(A)\})$ in $\mathcal{PGE}$, then we clearly have:
\begin{align*}
&\lambda (\mathit{App}(\mu, \mathit{v}_A)) = \lambda (\lambda^{-1}(\mu) \bullet \langle \mathit{der}_{\widehat{\Sigma}(\Gamma, A)}, \mathit{v}_A \rangle) = \lambda (\lambda^{-1} (\mu)) = \mu \\
&\lambda (\mathit{App}(\mu, \mathit{v}_A))^= = \lambda (\lambda^{-1}(\mu)^= \bullet \langle \mathit{der}_{\widehat{\Sigma}(\Gamma, A)}^=, \mathit{v}_A^= \rangle) = \lambda (\lambda^{-1} (\mu^=)) = \mu^=
\end{align*}
showing that $\lambda (\mathit{App}(\mu, \mathit{v}_A))$ and $\mu$ are the same ep-strategy. 
\end{proof}

\subsubsection{Game-semantic Dependent Pair Types}
\label{Sum}
Next, we consider $\Sigma$-types. Again, we begin with the general, categorical definition:
\begin{definition}[\textsf{CwF}s with $\Sigma$-types \cite{hofmann1997syntax}]
A \textsf{CwF} $\mathcal{C}$ \emph{\bfseries supports $\bm{\Sigma}$-types} if:
\begin{itemize}

\item {\bfseries \sffamily $\bm{\Sigma}$-Form.} For any $\Gamma \in \mathcal{C}$, $A \in \mathit{Ty}(\Gamma)$, $B \in \mathit{Ty}(\Gamma . A)$, there is a type
\begin{equation*}
\textstyle \Sigma (A, B) \in \mathit{Ty}(\Gamma).
\end{equation*}

\item {\bfseries \sffamily $\bm{\Sigma}$-Intro.} There is a morphism in $\mathcal{C}$
\begin{equation*}
\textstyle \mathit{Pair}_{A, B} : \Gamma . A . B \to \Gamma . \Sigma (A, B).
\end{equation*}

\item {\bfseries \sffamily $\bm{\Sigma}$-Elim.} For any $P \in \mathit{Ty}(\Gamma . \Sigma (A, B))$, $\psi \in \mathit{Tm}(\Gamma . A . B, P \{ \mathit{Pair}_{A, B} \})$, there is a term 
\begin{equation*}
\textstyle \mathcal{R}^{\Sigma}_{A, B, P}(\psi) \in \mathit{Tm}(\Gamma . \Sigma (A, B), P). 
\end{equation*}

\item {\bfseries \sffamily $\bm{\Sigma}$-Comp.} $\mathcal{R}^{\Sigma}_{A, B, P}(\psi) \{ \mathit{Pair}_{A, B}\} = \psi$ for all $\psi \in \mathit{Tm}(\Gamma . A . B, P \{ \mathit{Pair}_{A, B} \})$.

\item {\bfseries \sffamily $\bm{\Sigma}$-Subst.} For any $\Delta \in \mathcal{C}$, $\phi : \Delta \to \Gamma$ in $\mathcal{C}$, we have:
\begin{equation*}
\textstyle \Sigma (A, B) \{ \phi \} = \Sigma (A\{\phi\}, B\{\phi^+\})
\end{equation*}
where $\phi^+ \stackrel{\mathrm{df. }}{=} \langle \phi \circ \mathit{p}(A\{ \phi \}), \mathit{v}_{A\{\phi\}} \rangle_A : \Delta . A\{\phi\} \to \Gamma . A$.

\item {\bfseries \sffamily Pair-Subst.} Under the same assumption, we have:
\begin{align*}
\textstyle \mathit{p}(\Sigma (A, B)) \circ \mathit{Pair}_{A, B} &= \mathit{p}(A) \circ \mathit{p}(B) \\
\phi^* \circ \mathit{Pair}_{A\{\phi\}, B\{\phi^+\}} &= \mathit{Pair}_{A, B} \circ \phi^{++} 
\end{align*}
where $\phi^* \stackrel{\mathrm{df. }}{=} \langle \phi \circ \mathit{p}(\Sigma(A,B)\{\phi\}), \mathit{v}_{\Sigma(A,B)\{\phi\}} \rangle_{\Sigma(A,B)} : \Delta . \Sigma(A,B)\{\phi\} \to \Gamma . \Sigma(A,B)$ and $\phi^{++} \stackrel{\mathrm{df. }}{=} \langle \phi^+ \circ \mathit{p}(B\{\phi^+\}), \mathit{v}_{B\{\phi^+\}} \rangle_B : \Delta . A\{\phi\} . B\{\phi^+\} \to \Gamma . A . B$.

\item {\bfseries \sffamily$\bm{\mathcal{R}^{\Sigma}}$-Subst.} Finally, under the same assumption, we have:
\begin{equation*}
\textstyle \mathcal{R}^{\Sigma}_{A, B, P}(\psi) \{\phi^*\} = \mathcal{R}^{\Sigma}_{A\{\phi\}, B\{\phi^+\}, P\{\phi^*\}} (\psi \{ \phi^{++} \}).
\end{equation*}

\end{itemize}
Moreover, $\mathcal{C}$ \emph{\bfseries supports $\bm{\Sigma}$-types in the strict sense} if it additionally satisfies:
\begin{itemize}
\item {\bfseries \sffamily$\bm{\mathcal{R}^{\Sigma}}$-Uniq.} If any $\psi \in \mathit{Tm}(\Gamma . A . B, P \{ \mathit{Pair}_{A, B} \})$, $\varphi \in \mathit{Tm}(\Gamma . \Sigma(A, B), P)$ satisfy the equation $\varphi \{ \mathit{Pair}_{A, B} \} = \psi$, then $\varphi = \mathcal{R}^{\Sigma}_{A, B, P}(\psi)$.
\end{itemize}
\end{definition}

We now present our interpretation of $\Sigma$-types:
\begin{theorem}[Game-semantic $\Sigma$-types]
\label{ThmGameTheoreticSumTypes}
The CwF $\mathcal{PGE}$ supports $\Sigma$-types.
\end{theorem}
\begin{proof}
Let $\Gamma \in \mathcal{PGE}$, and $A : \Gamma \to \mathcal{PGE}$, $B : \widehat{\Sigma}(\Gamma, A) \to \mathcal{PGE}$ be DGwEs.

\paragraph{\textsc{$\Sigma$-Form.}} Similarly to the case of $\Pi$-types, we define the DGwG $\Sigma(A, B) : \Gamma \to \mathcal{PGE}$ by:
\begin{align*}
\textstyle \Sigma(A, B)(\gamma) &\stackrel{\mathrm{df. }}{=} \textstyle \widehat{\Sigma}(A(\gamma), B_\gamma) \\
\textstyle \Sigma(A, B)(\rho) &\stackrel{\mathrm{df. }}{=} \rho_{\Sigma(A, B)} : \textstyle \widehat{\Sigma}(A(\gamma), B_\gamma) \to \widehat{\Sigma}(A(\gamma'), B_{\gamma'})
\end{align*}
for all $\gamma, \gamma' : \Gamma$, $\rho : \gamma =_\Gamma \gamma'$, where $\rho_{\Sigma(A, B)}$ is the ep-strategy defined by:
\begin{align*}
\rho_{\Sigma(A, B)} &\stackrel{\mathrm{df. }}{=} \&_{\langle \sigma, \tau \rangle : \widehat{\Sigma}(A(\gamma), B_\gamma)} \langle \sigma, \tau \rangle \leftrightarrows \langle A(\rho) \bullet \sigma, B_\rho \bullet \tau \rangle \\
\rho_{\Sigma(A, B)}^= &\stackrel{\mathrm{df. }}{=} \&_{\langle \sigma_1, \tau_1 \rangle, \langle \sigma_2, \tau_2 \rangle : \widehat{\Sigma}(A(\gamma), B_\gamma), \langle \varrho, \vartheta \rangle : \langle \sigma_1, \tau_1 \rangle =_{\widehat{\Sigma}(A(\gamma), B_\gamma)} \langle \sigma_2, \tau_2 \rangle} \langle \varrho, \vartheta \rangle \leftrightarrows \langle A(\rho)^= \bullet \varrho, B_\rho^= \bullet \vartheta \rangle.
\end{align*}
It is straightforward to show the functoriality of $\rho_{\Sigma(A, B)}$: Let $\langle \sigma_1, \tau_1 \rangle, \langle \sigma_2, \tau_2 \rangle, \langle \sigma_3, \tau_3 \rangle : \widehat{\Sigma}(A(\gamma), B_\gamma)$, $\langle \varrho_1, \vartheta_1 \rangle : \langle \sigma_1, \tau_1 \rangle =_{\widehat{\Sigma}(A(\gamma), B_{\gamma})} \langle \sigma_2, \tau_2 \rangle$, $\langle \varrho_2, \vartheta_2 \rangle : \langle \sigma_2, \tau_2 \rangle =_{\widehat{\Sigma}(A(\gamma), B_{\gamma})} \langle \sigma_3, \tau_3 \rangle$.
\begin{enumerate}

\item As $A(\rho)^= \bullet \varrho_1 : A(\rho) \bullet \sigma_1 =_{A(\gamma')} A(\rho) \bullet \sigma_2$ and $B_\rho^= \bullet \vartheta_1 : B_{\gamma'}\{A(\rho)\}(\varrho_1) \bullet B_\rho \bullet \tau_1 =_{B_{\gamma'}(\sigma_2)} B_\rho \bullet \tau_2$, we have $\langle A(\rho)^= \bullet \varrho_1, B_\rho^= \bullet \vartheta_1 \rangle : \langle A(\rho) \bullet \sigma_1, B_\rho \bullet \tau_1 \rangle =_{\widehat{\Sigma}(A(\gamma'), B_{\gamma'})} \langle A(\rho) \bullet \sigma_2, B_\rho \bullet \tau_2 \rangle$, i.e.,
\begin{equation*}
\rho_{\Sigma(A, B)}^= \bullet \langle \varrho_1, \vartheta_1 \rangle : \rho_{\Sigma(A, B)} \bullet \langle \sigma_1, \tau_1 \rangle =_{\widehat{\Sigma}(A(\gamma'), B_{\gamma'})} \rho_{\Sigma(A, B)} \bullet \langle \sigma_2, \tau_2 \rangle.
\end{equation*}
Thus, $\rho_{\Sigma(A, B)}$ respects domain and codomain.

\item $\rho_{\Sigma(A, B)}$ respects composition:
\begin{align*}
\rho_{\Sigma(A, B)}^= \bullet (\langle \varrho_2, \vartheta_2 \rangle \bullet \langle \varrho_1, \vartheta_1 \rangle) &= \rho_{\Sigma(A, B)}^= \bullet \langle \varrho_2 \bullet \varrho_1, \vartheta_2 \odot \vartheta_1 \rangle \\
&= \langle A(\rho)^= \bullet (\varrho_2 \bullet \varrho_1), B_\rho^= \bullet (\vartheta_2 \odot \vartheta_1) \rangle \\
&= \langle (A(\rho)^= \bullet \varrho_2) \bullet (A(\rho)^= \bullet \varrho_1), (B_\rho^= \bullet \vartheta_2) \odot (B_\rho^= \bullet \vartheta_1) \rangle \\
&= \langle A(\rho)^= \bullet \varrho_2, B_\rho^= \bullet \vartheta_2 \rangle \bullet \langle A(\rho)^= \bullet \varrho_1, B_\rho^= \bullet \vartheta_1 \rangle \\
&= (\rho_{\Sigma(A, B)}^= \bullet \langle \varrho_2, \vartheta_2 \rangle) \bullet (\rho_{\Sigma(A, B)}^= \bullet \langle \varrho_1, \vartheta_1 \rangle).
\end{align*}

\item $\rho_{\Sigma(A, B)}$ respects identities:
\begin{align*}
\rho_{\Sigma(A, B)}^= \bullet \mathit{id}_{\langle \sigma_1, \tau_1 \rangle} &= \rho_{\Sigma(A, B)}^= \bullet \langle \mathit{id}_{\sigma_1}, \mathit{id}_{\tau_1} \rangle \\
&= \langle A(\rho)^= \bullet \mathit{id}_{\sigma_1}, B_\rho^= \bullet \mathit{id}_{\tau_1} \rangle \\
&= \langle \mathit{id}_{A(\rho) \bullet \sigma_1}, \mathit{id}_{B_\rho \bullet \tau_1} \rangle \\
&= \mathit{id}_{\langle A(\rho) \bullet \sigma_1, B_\rho \bullet \tau_1 \rangle} \\
&= \mathit{id}_{\rho_{\Sigma(A, B)} \bullet \langle \sigma_1, \tau_1 \rangle}.
\end{align*}

\end{enumerate}
Note that $\rho_{\Sigma(A, B)}$, $\rho_{\Sigma(A, B)}^=$ are both total, innocent, wb and noetherian by the same argument as the case of $\rho_{\Pi(A, B)}$, $\rho_{\Pi(A, B)}^=$.
Therefore we have shown that $\rho_{\Sigma(A, B)}$ is a well-defined morphism $\widehat{\Sigma}(A(\gamma), B_\gamma) \to \widehat{\Sigma}(A(\gamma'), B_{\gamma'})$ in $\mathcal{PGE}$.

Now, we show the functoriality of $\Sigma(A, B)$.
Let $\gamma, \gamma', \gamma'' : \Gamma$, $\rho : \gamma =_\Gamma \gamma'$, $\rho' : \gamma' =_\Gamma \gamma''$. 
We have already seen that $\Sigma(A, B)$ respects domain and codomain.
It remains to verify that $\Sigma(A, B)$ respects composition and identities. For all $\langle \sigma_1, \tau_1 \rangle, \langle \sigma_2, \tau_2 \rangle : \widehat{\Sigma}(A(\gamma), B_\gamma)$, $\langle \varrho, \vartheta \rangle : \langle \sigma_1, \tau_1 \rangle =_{\widehat{\Sigma}(A(\gamma), B_\gamma)} \langle \sigma_2, \tau_2 \rangle$, we have:
\begin{align*}
(\rho' \bullet \rho)_{\Sigma(A, B)}^= \bullet \langle \varrho, \vartheta \rangle &= \langle A(\rho' \bullet \rho)^= \bullet \varrho, B_{\rho' \bullet \rho}^= \bullet \vartheta \rangle \\
&= \langle A(\rho')^= \bullet A(\rho)^= \bullet \varrho, B_{\rho'}^= \bullet B_\rho^= \bullet \vartheta \rangle \\
&= \rho'_{\Sigma(A, B)} \bullet \langle A(\rho)^= \bullet \varrho, B_\rho^= \bullet \vartheta \rangle \\
&= \rho'^=_{\Sigma(A, B)} \bullet (\rho_{\Sigma(A, B)}^= \bullet \langle \varrho, \vartheta \rangle)
\end{align*}
as well as:
\begin{align*}
(\mathit{id}_\gamma)_{\Sigma(A, B)}^= \bullet \langle \varrho, \vartheta \rangle &= \langle A(\mathit{id}_\gamma)^= \bullet \varrho, B_{\mathit{id}_\gamma}^= \bullet \vartheta \rangle \\
&= \langle \mathit{id}_{A(\gamma)}^= \bullet \varrho, \mathit{id}_{B_\gamma}^= \bullet \vartheta \rangle \\
&= \langle \mathit{id}_{=_{A(\gamma)}} \bullet \varrho, \mathit{id}_{=_{B_\gamma}} \bullet \vartheta \rangle \\
&= \langle \varrho, \vartheta \rangle.
\end{align*}
Thus, we have shown that $\Sigma(A, B)$ is a well-defined DGwE over $\Gamma$.

\paragraph{\textsc{$\Sigma$-Intro.}} As shown in \cite{yamada2016game}, there is the obvious correspondence:
\begin{equation*}
\textstyle \mathsf{ob}(\widehat{\Sigma}(\widehat{\Sigma}(\Gamma, A), B)) \cong \mathsf{ob}(\widehat{\Sigma}(\Gamma, \Sigma(A, B)))
\end{equation*}
We may extend this correspondence to morphisms as well in the obvious way.
Accordingly, we may define the ep-strategy
\begin{equation*}
\mathit{Pair}_{A, B} : \textstyle \widehat{\Sigma}(\widehat{\Sigma}(\Gamma, A), B) \to \widehat{\Sigma}(\Gamma, \Sigma(A, B))
\end{equation*}
to be the identity morphism in the category $\mathcal{PGE}$ up to tags for disjoint union.

\paragraph{\textsc{$\Sigma$-Elim.}} Given $P \in \mathit{Ty}(\widehat{\Sigma}(\Gamma, \Sigma(A, B)))$, $\psi : \widehat{\Pi}(\widehat{\Sigma}(\widehat{\Sigma}(\Gamma, A), B), P\{ \mathit{Pair}_{A, B} \})$, we define:
\begin{equation*}
\mathcal{R}^{\Sigma}_{A, B, P}(\psi) \stackrel{\mathrm{df. }}{=} \psi \{ \mathit{Pair}_{A, B}^{-1} \} : \textstyle \widehat{\Pi}(\widehat{\Sigma}(\Gamma, \Sigma(A, B)), P).
\end{equation*}

\paragraph{\textsc{$\Sigma$-Comp.}} We then have:
\begin{align*}
\mathcal{R}^{\Sigma}_{A, B, P}(\psi)\{ \mathit{Pair}_{A, B} \} &= \psi \{ \mathit{Pair}_{A, B}^{-1} \} \{ \mathit{Pair}_{A, B} \} \\
&= \psi \{ \mathit{Pair}_{A, B}^{-1} \bullet \mathit{Pair}_{A, B} \} \\
&= \psi \{ \mathit{id}_{\widehat{\Sigma}(\widehat{\Sigma}(\Gamma, A), B)} \}\\
&= \psi.
\end{align*}

\paragraph{\textsc{$\Sigma$-Subst.}} Completely analogous to the case of $\Pi$-types.

\paragraph{\textsc{Pair-Subst.}}
As shown in \cite{yamada2016game}, we clearly have $\mathit{p}(\Sigma(A, B)) \bullet \mathit{Pair}_{A, B} = \mathit{p}(A) \bullet \mathit{p}(B)$, and:
\begin{align*}
\phi^* \bullet \mathit{Pair}_{A\{\phi\}, B\{\phi^+\}} &= \textstyle \langle \phi \bullet \mathit{p}(\Sigma(A,B)\{\phi\}), \mathit{v}_{\Sigma(A,B)\{\phi\}} \rangle \bullet \mathit{Pair}_{A\{\phi\}, B\{\phi^+\}} \\
&= \textstyle \langle \phi \bullet \mathit{p}(\Sigma(A\{\phi\},B\{\phi^+\})) \bullet \mathit{Pair}_{A\{\phi\}, B\{\phi^+\}}, \mathit{v}_{\Sigma(A,B)\{\phi\}} \bullet \mathit{Pair}_{A\{\phi\}, B\{\phi^+\}} \rangle \\
&= \langle \phi \bullet \mathit{p}(A\{\phi\}) \bullet \mathit{p}(B\{\phi^+\}), \mathit{v}_{\Sigma(A\{\phi\},B\{\phi^+\})} \bullet \mathit{Pair}_{A\{\phi\}, B\{\phi^+\}} \rangle \\
&= \mathit{Pair}_{A, B} \bullet \langle \langle \phi \bullet \mathit{p}(A\{ \phi \}), \mathit{v}_{A\{\phi\}} \rangle \bullet \mathit{p}(B\{\phi^+\}), \mathit{v}_{B\{\phi^+\}} \rangle \\
&= \mathit{Pair}_{A, B} \bullet \langle \phi^+ \bullet \mathit{p}(B\{\phi^+\}), \mathit{v}_{B\{\phi^+\}} \rangle \\
&= \mathit{Pair}_{A, B} \bullet \phi^{++} 
\end{align*}
where $\phi^* \stackrel{\mathrm{df. }}{=} \langle \phi \bullet \mathit{p}(\Sigma(A,B)\{\phi\}), \mathit{v}_{\Sigma(A,B)\{\phi\}} \rangle$, $\phi^{++} \stackrel{\mathrm{df. }}{=} \langle \phi^+ \bullet \mathit{p}(B\{\phi^+\}), \mathit{v}_{B\{\phi^+\}} \rangle$.

It is completely analogous to establish the equations on equality-preservations: 
\begin{align*}
\textstyle (\mathit{p}(\Sigma(A, B)) \bullet \mathit{Pair}_{A, B})^= &= (\mathit{p}(A) \bullet \mathit{p}(B))^= \\
\textstyle (\phi^* \bullet \mathit{Pair}_{A\{\phi\}, B\{\phi^+\}})^= &= (\mathit{Pair}_{A, B} \bullet \phi^{++})^=.
\end{align*}

\paragraph{\textsc{$\mathcal{R}^{\Sigma}$-Subst.}} Clearly, we have:
\begin{align*}
\textstyle \mathcal{R}^{\Sigma}_{A, B, P}(\psi) \{\phi^*\} &= \textstyle \psi \bullet \mathit{Pair}_{A, B}^{-1} \bullet \langle \phi \bullet \mathit{p}(\Sigma(A,B)\{\phi\}, \mathit{v}_{\Sigma(A,B)\{\phi\}} \rangle \\
&= (\psi \bullet \langle \phi^+ \bullet \mathit{p}(B\{\phi^+\}), \mathit{v}_{B\{\phi^+\}} \rangle) \bullet \mathit{Pair}_{A\{\phi\}, B\{\phi^+\}}^{-1} \\
&= \mathcal{R}^{\Sigma}_{A\{\phi\}, B\{\phi^+\}, P\{\phi^*\}} (\psi \bullet \langle \phi^+ \bullet \mathit{p}(B\{\phi^+\}), \mathit{v}_{B\{\phi^+\}} \rangle ) \\
&= \mathcal{R}^{\Sigma}_{A\{\phi\}, B\{\phi^+\}, P\{\phi^*\}} (\psi \{ \phi^{++} \}).
\end{align*}

And again, it is similar to show the equality of equality-preservations:
\begin{equation*}
\mathcal{R}^{\Sigma}_{A, B, P}(\psi) \{\phi^*\}^= = \mathcal{R}^{\Sigma}_{A\{\phi\}, B\{\phi^+\}, P\{\phi^*\}} (\psi \{ \phi^{++} \})^=.
\end{equation*}

\paragraph{\textsc{$\mathcal{R}^{\Sigma}$-Uniq.}} If $\psi \in \mathit{Tm}(\Gamma . A . B, P \{ \mathit{Pair}_{A, B} \})$ and $\varphi \in \mathit{Tm}(\Gamma . \Sigma(A, B), P)$ satisfy $\varphi \{ \mathit{Pair}_{A, B} \} = \psi$, i.e., $\varphi \bullet \mathit{Pair}_{A, B} = \psi$ and $\varphi^= \bullet \mathit{Pair}_{A, B}^= = \psi^=$, then 
\begin{align*}
\varphi &= \psi \bullet \mathit{Pair}_{A, B}^{-1} = \mathcal{R}^{\Sigma}_{A, B, P}(\psi) \\
\varphi^= &= \psi^= \bullet (\mathit{Pair}_{A, B}^=)^{-1} = \mathcal{R}^{\Sigma}_{A, B, P}(\psi)^=.
\end{align*}

\end{proof}

\subsubsection{Game-semantic Identity Types}
\label{Id}
Next, we consider identity types. Again, we first review the general, categorical interpretation.
\begin{definition}[\textsf{CwF}s with identity types \cite{hofmann1997syntax}]
A \textsf{CwF} $\mathcal{C}$ \emph{\bfseries supports identity types} if:
\begin{itemize}

\item {\bfseries \sffamily Id-Form.} For any $\Gamma \in \mathcal{C}$, $A \in \mathit{Ty}(\Gamma)$, there is a type
\begin{equation*}
\mathsf{Id}_A \in \mathit{Ty}(\Gamma . A . A^+)
\end{equation*}
where $A^+ \stackrel{\mathrm{df. }}{=} A\{\mathit{p}(A)\} \in \mathit{Ty}(\Gamma . A)$.

\item {\bfseries \sffamily Id-Intro.} There is a morphism in $\mathcal{C}$
\begin{equation*}
\mathit{Refl}_A : \Gamma . A \to \Gamma . A . A^+ . \mathsf{Id}_A.
\end{equation*}

\item {\bfseries \sffamily Id-Elim.} For each $B \in \mathit{Ty}(\Gamma . A . A^+ . \mathsf{Id}_A)$, $\tau \in \mathit{Tm}(\Gamma . A, B\{\mathit{Refl}_A\})$, there is a term
\begin{equation*}
\mathcal{R}^{\mathsf{Id}}_{A,B}(\tau) \in \mathit{Tm}(\Gamma . A . A^+ . \mathsf{Id}_A, B).
\end{equation*}

\item {\bfseries \sffamily Id-Comp.} $\mathcal{R}^{\mathsf{Id}}_{A,B}(\tau)\{\mathit{Refl}_A\} = \tau$ for all $\tau \in \mathit{Tm}(\Gamma . A, B\{\mathit{Refl}_A\})$.

\item {\bfseries \sffamily Id-Subst.} For any $\Delta \in \mathcal{C}$, $\phi : \Delta \to \Gamma$ in $\mathcal{C}$, we have:
\begin{equation*}
\mathsf{Id}_A\{\phi^{++}\} = \mathsf{Id}_{A\{\phi\}} \in \mathit{Ty}(\Delta . A\{\phi\} . A\{\phi\}^+)
\end{equation*}
where $A\{\phi\}^+ \stackrel{\mathrm{df. }}{=} A\{\phi\}\{\mathit{p}(A\{\phi\})\} \in \mathit{Ty}(\Delta . A\{\phi\})$, $\phi^+ \stackrel{\mathrm{df. }}{=} \langle \phi \circ \mathit{p}(A\{ \phi \}), \mathit{v}_{A\{\phi\}} \rangle_A : \Delta . A\{\phi\} \to \Gamma . A$, and $\phi^{++} \stackrel{\mathrm{df. }}{=} \langle \phi^+ \circ \mathit{p}(A^+\{\phi^+\}), \mathit{v}_{A^+\{\phi^+\}} \rangle_{A^+} : \Delta . A\{\phi\} . A^+\{\phi^+\} \to \Gamma . A . A^+$.

\item {\bfseries \sffamily Refl-Subst.} Under the same assumption, the following equation holds
\begin{equation*}
\mathit{Refl}_A \circ \phi^+ = \phi^{+++} \circ \mathit{Refl}_{A\{\phi\}} : \Delta . A\{\phi\} \to \Gamma . A . A^+ . \mathsf{Id}_A
\end{equation*}
where $\phi^{+++} \stackrel{\mathrm{df. }}{=} \langle \phi^{++} \circ \mathit{p}(\mathsf{Id}_A\{\phi^{++}\}), \mathit{v}_{\mathsf{Id}_A\{\phi^{++}\}} \rangle_{\mathsf{Id}_A} : \Delta . A\{\phi\} . A^+\{\phi^+\} . \mathsf{Id}_{A\{\phi\}} \to \Gamma . A . A^+ . \mathsf{Id}_A$. Note that $\mathsf{Id}_A\{\phi^{++}\} = \mathsf{Id}_{A\{\phi\}}$, $A^+\{\phi^+\} = A\{\phi\}^+$.

\item {\bfseries \sffamily $\bm{\mathcal{R}^{\mathsf{Id}}}$-Subst.} $\mathcal{R}^{\mathsf{Id}}_{A,B}(\tau)\{\phi^{+++}\} = \mathcal{R}^{\mathsf{Id}}_{A\{\phi\}, B\{\phi^{+++}\}}(\tau\{\phi^+\})$ under the same assumption

\end{itemize}
\end{definition}

Now, let us give our interpretation of \textsf{Id}-types:
\begin{theorem}[Game-semantic \textsf{Id}-types]
\label{ThmGameTheoreticIdTypes}
The CwF $\mathcal{PGE}$ supports \textsf{Id}-types.
\end{theorem}
\begin{proof}
Let $\Gamma \in \mathcal{PGE}$, $A : \Gamma \to \mathcal{PGE}$, $A^+ \stackrel{\mathrm{df. }}{=} A \{ \mathit{p}(A) \} \in \mathit{Ty}(\widehat{\Sigma}(\Gamma, A))$.

\paragraph{\textsc{$\mathsf{Id}$-Form.}} We define the DGwE $\mathsf{Id}_A : \widehat{\Sigma}(\widehat{\Sigma}(\Gamma, A), A^+) \to \mathcal{PGE}$ by:
\begin{align*}
\mathsf{Id}_A((\gamma, \sigma), \tilde{\sigma}) &\stackrel{\mathrm{df. }}{=} \widehat{\mathsf{Id}}_{A(\gamma)}(\sigma, \tilde{\sigma}) \\
\mathsf{Id}_A((\rho \& \varrho) \& \tilde{\varrho}) &\stackrel{\mathrm{df. }}{=} \&_{\alpha : \widehat{\mathsf{Id}}_{A(\gamma)}(\sigma, \tilde{\sigma})} \alpha \leftrightarrows \tilde{\varrho} \bullet (A(\rho)^= \bullet \alpha) \bullet (\varrho^{-1})
\end{align*}
for all $\gamma, \gamma' : \Gamma$, $\rho : \gamma =_\Gamma \gamma'$, $\sigma, \tilde{\sigma} : A(\gamma)$, $\sigma', \tilde{\sigma}' : A(\gamma')$, $\varrho : A(\rho) \bullet \sigma =_{A(\gamma')} \sigma'$, $\tilde{\varrho} : A(\rho) \bullet \tilde{\sigma} =_{A(\gamma')} \tilde{\sigma}'$, where $\&_{\alpha : \widehat{\mathsf{Id}}_{A(\gamma)}(\sigma, \tilde{\sigma})} \alpha \leftrightarrows \tilde{\varrho} \bullet (A(\rho)^= \bullet \alpha) \bullet (\varrho^{-1}) : \widehat{\mathsf{Id}}_{A(\gamma)}(\sigma, \tilde{\sigma}) \to \widehat{\mathsf{Id}}_{A(\gamma')}(\sigma', \tilde{\sigma}')$ is an ep-strategy equipped with the trivial equality-preservation.
Following the same pattern as before, it is easy to see that $\&_{\alpha : \widehat{\mathsf{Id}}_{A(\gamma)}(\sigma, \tilde{\sigma})} \alpha \leftrightarrows (\rho^{-1})^\dagger ; (\alpha^\dagger ; A(\rho)^=)^\dagger ; \tilde{\varrho}$ is a strategy on $\widehat{\mathsf{Id}}_{A(\gamma)}(\sigma, \tilde{\sigma}) \to \widehat{\mathsf{Id}}_{A(\gamma')}(\sigma', \tilde{\sigma}')$; and then it is trivially ep.

We now show the functoriality of $\mathsf{Id}_A$.
Let $\gamma, \gamma', \gamma'' : \Gamma$, $\rho : \gamma =_\Gamma \gamma'$, $\rho' : \gamma' =_\Gamma \gamma''$, $\sigma, \tilde{\sigma} : A(\gamma)$, $\sigma', \tilde{\sigma}' : A(\gamma')$, $\sigma'', \tilde{\sigma}'' : A(\gamma'')$, $\varrho : A(\rho) \bullet \sigma =_{A(\gamma')} \sigma'$, $\tilde{\varrho} : A(\rho) \bullet \tilde{\sigma} =_{A(\gamma')} \tilde{\sigma}'$, $\varrho' : A(\rho') \bullet \sigma' =_{A(\gamma'')} \sigma''$, $\tilde{\varrho}' : A(\rho') \bullet \tilde{\sigma'} =_{A(\gamma'')} \tilde{\sigma}''$.
First, $\mathsf{Id}_A$ clearly respects domain and codomain.
It is also easy to see that $\mathsf{Id}_A$ respects composition:
\begin{align*}
& \ \mathsf{Id}_A (((\rho' \& \varrho') \& \tilde{\varrho}') \bullet ((\rho \& \varrho) \& \tilde{\varrho})) \\
= & \ \mathsf{Id}_A (((\rho' \bullet \rho) \& (\varrho' \odot \varrho)) \& (\tilde{\varrho}' \odot \tilde{\varrho})) \\
= & \ \&_{\alpha : \widehat{\mathsf{Id}}_{A(\gamma)}(\sigma, \tilde{\sigma})} \alpha \leftrightarrows (\tilde{\varrho}' \odot \tilde{\varrho}) \bullet (A(\rho' \bullet \rho)^= \bullet \alpha) \bullet (\varrho' \odot \varrho)^{-1} \\
= & \ \&_{\alpha : \widehat{\mathsf{Id}}_{A(\gamma)}(\sigma, \tilde{\sigma})} \alpha \leftrightarrows \tilde{\varrho}' \bullet (A(\rho')^= \bullet \tilde{\varrho}) \bullet (A(\rho')^= \bullet A(\rho)^= \bullet \alpha) \bullet (\varrho' \bullet (A(\rho')^= \bullet \varrho))^{-1} \\
= & \ \&_{\alpha : \widehat{\mathsf{Id}}_{A(\gamma)}(\sigma, \tilde{\sigma})} \alpha \leftrightarrows \tilde{\varrho}' \bullet (A(\rho')^= \bullet (\tilde{\varrho} \bullet (A(\rho)^= \bullet \alpha))) \bullet (A(\rho')^= \bullet \varrho^{-1}) \bullet \varrho'^{-1} \\
= & \ \&_{\alpha : \widehat{\mathsf{Id}}_{A(\gamma)}(\sigma, \tilde{\sigma})} \alpha \leftrightarrows \tilde{\varrho}' \bullet (A(\rho')^= \bullet (\tilde{\varrho} \bullet (A(\rho)^= \bullet \alpha) \bullet \varrho^{-1})) \bullet \varrho'^{-1} \\
= & \ \&_{\alpha : \widehat{\mathsf{Id}}_{A(\gamma)}(\sigma, \tilde{\sigma})} \alpha \leftrightarrows \mathsf{Id}_A((\rho' \& \varrho') \& \tilde{\varrho}') \bullet (\mathsf{Id}_A((\rho \& \varrho) \& \tilde{\varrho}) \bullet \alpha)) \\
= & (\&_{\alpha' : \widehat{\mathsf{Id}}_{A(\gamma')}(\sigma', \tilde{\sigma}')} \alpha' \leftrightarrows \tilde{\varrho}' \bullet (A(\rho')^= \bullet \alpha') \bullet (\varrho'^{-1})) \bullet (\&_{\alpha : \widehat{\mathsf{Id}}_{A(\gamma)}(\sigma, \tilde{\sigma})} \alpha \leftrightarrows \tilde{\varrho} \bullet (A(\rho)^= \bullet \alpha) \bullet (\varrho^{-1})) \\
= & \mathsf{Id}_A((\rho' \& \varrho') \& \tilde{\varrho}') \bullet \mathsf{Id}_A((\rho \& \varrho) \& \tilde{\varrho})
\end{align*}
where note that the case for equality-preservation is trivial.
Similarly, $\mathsf{Id}_A$ respects identities:
\begin{align*}
\mathsf{Id}_A(\mathit{id}_{\langle \langle \gamma, \sigma \rangle, \tilde{\sigma} \rangle}) &= \mathsf{Id}_A((\mathit{id}_\gamma \& \mathit{id}_\sigma) \& \mathit{id}_{\tilde{\sigma}}) \\
&= \&_{\alpha : \widehat{\mathsf{Id}}_{A(\gamma)}(\sigma, \tilde{\sigma})} \alpha \leftrightarrows \mathit{id}_{\tilde{\sigma}} \bullet (A(\mathit{id}_\gamma)^= \bullet \alpha) \bullet \mathit{id}_\sigma^{-1} \\
&= \&_{\alpha : \widehat{\mathsf{Id}}_{A(\gamma)}(\sigma, \tilde{\sigma})} \alpha \leftrightarrows \mathit{id}_{A(\gamma)}^= \bullet \alpha \\
&= \&_{\alpha : \widehat{\mathsf{Id}}_{A(\gamma)}(\sigma, \tilde{\sigma})} \alpha \leftrightarrows \alpha \\
&= \mathit{id}_{\widehat{\mathsf{Id}}_{A(\gamma)}(\sigma, \tilde{\sigma})} \\
&= \mathit{id}_{\mathsf{Id}_{A}((\gamma, \sigma), \tilde{\sigma})}.
\end{align*}

\paragraph{\textsc{$\mathsf{Id}$-Intro.}} The ep-strategy $\mathit{Refl}_A : \widehat{\Sigma}(\Gamma, A) \to \widehat{\Sigma}(\widehat{\Sigma} (\widehat{\Sigma} (\Gamma, A), A^+), \mathsf{Id}_A)$ is defined by:
\begin{align*}
\mathit{Refl}_A &\stackrel{\mathrm{df. }}{=} \&_{\langle \gamma, \sigma \rangle : \widehat{\Sigma}(\Gamma, A)} \langle \gamma, \sigma \rangle \leftrightarrows \langle \langle \langle \gamma, \sigma \rangle, \sigma \rangle, \sigma \rangle \\
\mathit{Refl}_A^= &\stackrel{\mathrm{df. }}{=} \&_{\langle \gamma, \sigma \rangle, \langle \gamma', \sigma' \rangle : \widehat{\Sigma}(\Gamma, A), \rho \& \varrho : \langle \gamma, \sigma \rangle =_{\widehat{\Sigma}(\Gamma, A)} \langle \gamma', \sigma' \rangle} \rho \& \varrho \leftrightarrows ((\rho \& \varrho) \& \varrho) \& \varrho.
\end{align*}
We omit the verification of the functoriality of $\mathit{Refl}_A$ as it is just straightforward.
Note that $\mathit{Refl}_A$ has its inverse $\mathit{Refl}_A^{-1}$, which is just the identity up to tags for disjoint union.

\paragraph{\textsc{$\mathsf{Id}$-Elim.}} Given $\psi : \widehat{\Pi}(\widehat{\Sigma}(\Gamma, A), B\{\mathit{Refl}_A\})$, we define:
\begin{equation*}
\mathcal{R}^{\mathsf{Id}}_{A, B}(\psi) \stackrel{\mathrm{df. }}{=} \psi \{ \mathit{Refl}_A^{-1} \} : \textstyle \widehat{\Pi}(\widehat{\Sigma}(\widehat{\Sigma}(\widehat{\Sigma}(\Gamma, A), A), \mathsf{Id}_A), B).
\end{equation*}

\paragraph{\textsc{$\mathsf{Id}$-Comp.}} Then we have:
\begin{align*}
\mathcal{R}^{\mathsf{Id}}_{A, B}(\psi) \{ \mathit{Refl}_A \} &= \psi \{ \mathit{Refl}_A^{-1} \} \{ \mathit{Refl}_A \} \\
&= \psi \{ \mathit{Refl}_A^{-1} \bullet \mathit{Refl}_A \} \\
&= \psi \{ \mathit{id}_{\widehat{\Sigma}(\Gamma, A)} \} \\
&= \psi.
\end{align*}

\paragraph{\textsc{$\mathsf{Id}$-Subst.}} For any game $\Delta \in \mathcal{WPG}$ and strategy $\phi : \Delta \to \Gamma$ in $\mathcal{WPG}$, we have:
\begin{align*}
& \ \textsf{Id}_A\{\phi^{++}\} \\
= & \ \textstyle \{ \widehat{\textsf{Id}}_A(\sigma, \sigma') \ \! | \ \! \langle \langle \gamma, \sigma \rangle, \sigma' \rangle : \widehat{\Sigma}(\widehat{\Sigma}(\Gamma, A), A^+) \} \{\phi^{++}\} \\
= & \ \textstyle \{ \textsf{Id}_{A\{\phi\}}(\phi^{++} \bullet \langle \langle \delta, \tau \rangle, \tau' \rangle) \ \! | \ \! \langle \langle \delta, \tau \rangle, \tau' \rangle : \widehat{\Sigma}(\widehat{\Sigma}(\Delta, A\{\phi\}), A\{\phi\}^+) \} \\
= & \ \textstyle \{ \textsf{Id}_{A\{\phi\}}(\langle\phi \bullet \delta, \tau \rangle, \tau' \rangle) \ \! | \ \! \langle \langle \delta, \tau \rangle, \tau' \rangle : \widehat{\Sigma}(\widehat{\Sigma}(\Delta, A\{\phi\}), A\{\phi\}^+) \} \\
= & \ \textstyle \{ \widehat{\textsf{Id}}_{A\{\phi\}}(\tau, \tau') \ \! | \ \! \langle \langle \delta, \tau \rangle, \tau' \rangle : \widehat{\Sigma}(\widehat{\Sigma}(\Delta, A\{\phi\}), A\{\phi\}^+) \} \\
= & \ \textsf{Id}_{A\{\phi\}} 
\end{align*}
where $\phi^+ \stackrel{\mathrm{df. }}{=} \langle \phi \bullet \mathit{p}(A\{ \phi \}), \mathit{v}_{A\{\phi\}} \rangle : \widehat{\Sigma}(\Delta, A\{\phi\}) \to \widehat{\Sigma}(\Gamma, A)$, $\phi^{++} \stackrel{\mathrm{df. }}{=} \langle \phi^+ \bullet \mathit{p}(A^+\{\phi^+\}), \mathit{v}_{A^+\{\phi^+\}} \rangle : \widehat{\Sigma}(\widehat{\Sigma}(\Delta, A\{\phi\}), A^+\{\phi^+\}) \to \widehat{\Sigma}(\widehat{\Sigma}(\Gamma, A), A^+)$. Note that
\begin{align*}
\phi^{++} \bullet \langle \langle \delta, \tau \rangle, \tau' \rangle &= \langle \phi^+ \bullet \mathit{p}(A^+\{\phi^+\}), \mathit{v}_{A^+\{\phi^+\}} \rangle \bullet \langle \langle \delta, \tau \rangle, \tau' \rangle \\
&= \langle \phi^+ \bullet \mathit{p}(A^+\{\phi^+\}) \bullet \langle \langle \delta, \tau \rangle, \tau' \rangle, \mathit{v}_{A^+\{\phi^+\}} \bullet \langle \langle \delta, \tau \rangle, \tau' \rangle \rangle \\
&= \langle \langle \phi \bullet \delta, \tau \rangle, \tau' \rangle.
\end{align*}

\paragraph{\textsc{Refl-Subst.}}Also, the following equation holds:
\begin{align*}
&\mathit{Refl}_A \bullet \phi^+ \\
= \ &\mathit{Refl}_A \bullet \langle \phi \bullet \mathit{p}(A\{ \phi \}), \mathit{v}_{A\{\phi\}} \rangle \\
= \ &\langle \langle \langle \phi \bullet \mathit{p}(A\{ \phi \}), \mathit{v}_{A\{\phi\}} \rangle \bullet \mathit{p}(A^+\{\phi^+\}) \bullet \mathit{p}(\mathsf{Id}_A\{\phi^{++}\}), \mathit{v}_{A^+\{\phi^+\}} \bullet \mathit{p}(\mathsf{Id}_A\{\phi^{++}\}) \rangle, \mathit{v}_{\mathsf{Id}_A\{\phi^{++}\}} \rangle \bullet \mathit{Refl}_{A\{\phi\}} \\
= \ &\langle \langle \phi^+ \bullet \mathit{p}(A^+\{\phi^+\}), \mathit{v}_{A^+\{\phi^+\}} \rangle \bullet \mathit{p}(\mathsf{Id}_A\{\phi^{++}\}), \mathit{v}_{\mathsf{Id}_A\{\phi^{++}\}} \rangle \bullet \mathit{Refl}_{A\{\phi\}} \\
= \ &\langle \phi^{++} \bullet \mathit{p}(\mathsf{Id}_A\{\phi^{++}\}), \mathit{v}_{\mathsf{Id}_A\{\phi^{++}\}} \rangle \bullet \mathit{Refl}_{A\{\phi\}} \\
= \ &\phi^{+++} \bullet \mathit{Refl}_{A\{\phi\}}
\end{align*}
where $\phi^{+++} \stackrel{\mathrm{df. }}{=} \langle \phi^{++} \bullet \mathit{p}(\mathsf{Id}_A\{\phi^{++}\}), \mathit{v}_{\mathsf{Id}_A\{\phi^{++}\}} \rangle$.

\paragraph{\textsc{$\mathcal{R}^{\mathsf{Id}}$-Comp.}} Finally, we have:
\begin{align*}
&\mathcal{R}^{\mathsf{Id}}_{A,B}(\tau)\{\phi^{+++}\} \\
= \ &(\tau \bullet \mathit{Refl}_A^{-1}) \bullet \langle \langle \langle \phi \bullet \mathit{p}(A\{ \phi \}), \mathit{v}_{A\{\phi\}} \rangle \bullet \mathit{p}(A^+\{\phi^+\}), \mathit{v}_{A^+\{\phi^+\}} \rangle \bullet \mathit{p}(\mathsf{Id}_A\{\phi^{++}\}), \mathit{v}_{\mathsf{Id}_A\{\phi^{++}\}} \rangle \\
= \ &\tau \bullet \langle \phi \bullet \mathit{p}(A\{ \phi \}), \mathit{v}_{A\{\phi\}} \rangle \bullet \mathit{Refl}_{A\{\phi\}}^{-1} \\
= \ &\mathcal{R}^{\mathsf{Id}}_{A\{\phi\}, B\{\phi^{+++}\}}(\tau \bullet \langle \phi \bullet \mathit{p}(A\{ \phi \}), \mathit{v}_{A\{\phi\}} \rangle) \\ 
= \ &\mathcal{R}^{\mathsf{Id}}_{A\{\phi\}, B\{\phi^{+++}\}}(\tau \bullet \phi^+) \\
= \ &\mathcal{R}^{\mathsf{Id}}_{A\{\phi\}, B\{\phi^{+++}\}}(\tau\{\phi^+\}).
\end{align*}

\end{proof}

Our interpretation of \textsf{Id}-types accommodates non-identity morphisms in the same manner as the groupoid model \cite{hofmann1998groupoid}, and so it refutes UIP essentially by the same argument as follows.
Recall that UIP states: For any type $\mathsf{A}$, the following type can be inhabited
\begin{equation*}
\mathsf{\Pi_{a_1, a_2 : A} \Pi_{p, q : a_1 = a_2} p = q}.
\end{equation*}
Consider the GwE $\mathscr{B}$ whose positions are prefixes of the sequences $q_{\mathit{tt}} . \mathit{tt}, q_{\mathit{ff}} . \mathit{ff}$ with all isomorphism strategies between strategies on $\mathscr{B}$ as morphisms. Let us write $\bullet : \mathsf{B}$ for the unique total strategy. Explicitly, the morphisms are the dereliction $\mathit{der}_{\bullet}$ and the ``reversing" strategy $\mathit{rv}_{\bullet}$. We then have $\textit{der}_{\bullet} \neq_{\bullet =_{\mathscr{B}} \bullet} \mathit{rv}_{\bullet}$ because morphisms in $\bullet =_{\mathscr{B}} \bullet$ are only the trivial ones.

Since the ``atomic games'' such as the natural number game $N$ can be seen as discrete wf-GwEs, we may inherit the results in \cite{yamada2016game} so that $\mathcal{PGE}$ supports $\mathsf{N}$- $\mathsf{1}$- and $\mathsf{0}$-types as well as universes.
By the \emph{soundness} of CwFs, we have established:

\begin{corollary}
There is a (sound) model of \textsf{MLTT} with $\Pi$- $\Sigma$- \textsf{Id}- \textsf{N}- \textsf{1}- and \textsf{0}-types as well as universes in the CwF $\mathcal{PGE}$ of wf-GwEs and ep-strategies.
\end{corollary}

In the light of UA, however, universe games $\mathcal{U}$ should not be discrete.
Note that as $\mathcal{PGE}$ has no higher morphisms, its interpretation of UA is: For any $A, B \in \mathcal{PGE}$, there is an isomorphism $ua : (A \cong B) \stackrel{\sim}{\Rightarrow} (\underline{A} =_{\mathcal{U}} \underline{B})$\footnote{But the argument below holds even if $\mathcal{PGE}$ has higher morphisms.}.
As Theorem~\ref{ThmIsomThm} indicates, we may fix an isomorphism $\phi_0 : A \stackrel{\sim}{\Rightarrow} B$ and represent any $\phi : A \stackrel{\sim}{\Rightarrow} B$ by $\phi = \beta \bullet \phi_0 \bullet \alpha^{-1}$, where $\alpha : A \stackrel{\sim}{\Rightarrow} A$, $\beta : B \stackrel{\sim}{\Rightarrow} B$.
This suggests a solution to refine the notion of a GwE by equipping it with an \emph{order on its maximal positions of the same length}, writing $A \equiv A'$ if $A, A' \in \mathcal{PGE}$ are the same up to the order, and define $\mathsf{st}(\underline{A} =_{\mathcal{U}} \underline{B}) \stackrel{\mathrm{df. }}{=} \bigcup_{A' \equiv A, B' \equiv B} \mathsf{st}(\underline{A'} \stackrel{\sim}{\Rightarrow} \underline{B'})$ if $A \cong B$ and $\mathsf{st}(\underline{A} =_{\mathcal{U}} \underline{B}) \stackrel{\mathrm{df. }}{=} \emptyset$ otherwise, where the name $\mathcal{N}(A)$ of each $A \in \mathcal{PGE}$ incorporates the order on positions in $A$.

Nevertheless, from a \emph{computational} viewpoint, even this equality in $\mathcal{U}$ challenges UA.
The direction $ua^{-1} : (\underline{A} =_{\mathcal{U}} \underline{B}) \stackrel{\sim}{\Rightarrow} (A \cong B)$ is no problem as $ua^{-1}$ may learn the given $\alpha : A \stackrel{\sim}{\Rightarrow} A$, $\beta : B \stackrel{\sim}{\Rightarrow} B$ from the domain and construct $\beta \bullet \phi_0 \bullet \alpha^{-1} : A \stackrel{\sim}{\Rightarrow} B$.
In contrast, $ua : (A \cong B) \stackrel{\sim}{\Rightarrow} (\underline{A} =_{\mathcal{U}} \underline{B})$ would be intractable as it needs to determine the given $A' \cong A$, $B' \cong B$ by a \emph{finite} interaction with them.
In any case, since UA implies FunExt \cite{hottbook} and our model refutes FunExt, it must refute UA as well.
\begin{corollary}[Intensionality]
The model in $\mathcal{PGE}$ refutes UIP, FunExt and UA.
\end{corollary}

\section{Conclusion and Future Work}
\label{ConclusionAndFutureWork}
We have presented the first game semantics for \textsf{MLTT} that refutes UIP.
Its algebraic structure is very similar to the groupoid model \cite{hofmann1998groupoid}, but in contrast it is intensional, refuting FunExt and UA.
Hence, in some sense, we have given a negative answer to the computational nature of UA.
This view stands in a sharp contrast to the cubical set model \cite{bezem2014model,cohen2016cubical} that interprets propositional equalities as paths and validates UA.

For future work, we plan to generalize the notion of GwEs to $\omega$-groupoids in order to interpret higher equalities in a non-trivial manner. 
Moreover, it would be interesting to see connections between our game model and the cubical set model, which may shed a new light on relations between computation and topology.

\subsection*{Acknowledgements}
The author was supported by Funai Overseas Scholarship. 
Also, he is grateful to Samson Abramsky for fruitful discussions.

\bibliographystyle{alpha}
\bibliography{GamesAndStrategies,HoTT,Recursion,Categories,CategoricalLogic,TypeTheoriesAndProgrammingLanguages,PCF}

\end{document}